\documentclass[reqno]{amsart}
\usepackage{fullpage,amsfonts,amssymb,amsthm,mathrsfs,graphicx,color,framed,esint,stackrel,amsmath}
\usepackage{cancel,bm}
\usepackage{tikz}
\usetikzlibrary{decorations.markings}
\usepackage{booktabs,longtable}
\usepackage[initials]{amsrefs}

\usepackage[colorlinks=true, pdfstartview=FitV, linkcolor=blue, citecolor=blue, urlcolor=blue]{hyperref}

\renewcommand{\d}{{\mathrm d}}

\newcommand{\im}{\mathrm{i}}
\newcommand{\e}{\mathrm{e}}

\def\pf{\mathop{\mathrm{pf}}\limits}

\newtheorem{theo}{Theorem}[section]

\newtheorem{lem}[theo]{Lemma}
\newtheorem{rem}[theo]{Remark}
\newtheorem{problem}[theo]{Riemann-Hilbert Problem}

\newtheorem{prop}[theo]{Proposition} 
\newtheorem{cor}[theo]{Corollary} 
\newtheorem{definition}[theo]{Definition}

\begin{document}
\title{On Fredholm Pfaffians and Riemann-Hilbert problems}
\author{Thomas Bothner}
\address{School of Mathematics, University of Bristol, Fry Building, Woodland Road, Bristol, BS8 1UG, United Kingdom}
\email{thomas.bothner@bristol.ac.uk}
\author{Amari Jaconelli}
\address{School of Mathematics, University of Bristol, Fry Building, Woodland Road, Bristol, BS8 1UG, United Kingdom}
\email{amari.jaconelli@bristol.ac.uk}

\date{\today}


\dedicatory{Dedicated with pleasure to Percy Deift on his 80th birthday}


%\thanks{}

\begin{abstract}
It is shown how classes of Fredholm Pfaffians can be computed in terms of canonical, auxiliary Riemann-Hilbert problems as soon as the main kernel in the Pfaffian is either of additive Hankel composition or of truncated Wiener-Hopf type. Akhiezer-Kac asymptotic results for the Fredholm Pfaffians are then derived as natural consequences of the Riemann-Hilbert characterisation.
\end{abstract}

\maketitle

\section{Introduction}

In September $2025$, we celebrated Percy Deift's $80$th birthday. In this paper, we will use two of his favourite tools to derive some interesting results about Fredholm Pfaffians. Namely, we utilise the commutation formula for two bounded linear transformations $A\in\mathcal{L}(X,Y),B\in\mathcal{L}(Y,X)$ between Banach spaces $X,Y$, i.e.
\begin{equation}\label{c1}
	(I_Y-AB)^{-1}=I_Y+A(I_X-BA)^{-1}B,
\end{equation}
in the sense that if $I_X-BA$ is invertible, then so is $I_Y-AB$ and $I_Y+A(I_X-BA)^{-1}B$ gives $(I_Y-AB)^{-1}$. Formula \eqref{c1} is very useful in mathematical physics, cf. \cite{D}, and we will apply its most obvious consequence
\begin{equation}\label{c2}
	\sigma(AB)\setminus\{0\}=\sigma(BA)\setminus\{0\},
\end{equation}
about the spectra of $AB\in\mathcal{L}(Y)$ and $BA\in\mathcal{L}(X)$. In fact, we will study integral operators $AB$ and $BA$ with sufficiently good spectral properties, and for those \eqref{c2} yields equality of their respective, possibly regularised, Fredholm Pfaffians. So, thanks to \eqref{c1}, we will be in a position to compute Fredholm Pfaffians by transforming them to Fredholm determinants of simpler operators. Still, even those simpler operator determinants are mostly challenging to analyse, especially when asymptotic questions are posed. It is at that moment when we make contact with another one of Percy's favourites: Riemann-Hilbert matrix factorisation problems (RHPs). For, assuming that our kernel functions in the simpler determinants posses convenient algebraic features (that is Hankel composition or truncated Wiener-Hopf features), we will be able to compute the determinants in terms of a canonical RHP. What results is an opportune pathway to deriving asymptotics for the original Fredholm Pfaffian via Deift-Zhou nonlinear steepest descent techniques, cf. \cite{DZ}. 
\begin{rem}
What we achieve in the following, by applying Percy's teachings, is an analytic spin on the beautiful random walker analysis of Fredholm Pfaffians by FitzGerald-Tribe-Zaboronski \cite{FTZ} - itself inspired by calculations of Mark Kac \cite{Ka} in the $1950$s. We obtain, using an analytic approach based on commutation and RHPs, our main results in \eqref{c19},\eqref{c20} and in \eqref{c33},\eqref{c34}. Out of those, \eqref{c20} and \eqref{c34} are partially known from \cite{FTZ}, \eqref{c19} and \eqref{c33}, as well as the intermediate results \eqref{c3},\eqref{c7} are novel.
\end{rem}
We begin by introducing our main players:
%\subsection{Main definitions}
consider a $2\times 2$ matrix-valued kernel
\begin{equation}\label{a1}
	M(x,y)=\begin{bmatrix}M_{11}(x,y) & M_{12}(x,y)\\ M_{21}(x,y) & M_{22}(x,y)\end{bmatrix},\ \ \ \ \ \ x,y\in\Delta:=\bigcup_{j=1}^m(a_{2j-1},a_{2j})\subset\mathbb{R},\ \ m\in\mathbb{N},
\end{equation}
with $-\infty<a_1<a_2<\ldots<a_{2m-1}<a_{2m}<\infty$ and let $M$ denote the integral operator on $L^2(\Delta)\oplus L^2(\Delta)$ with kernel $M(x,y)$. Assuming the diagonal $M_{11},M_{22}$ are trace class integral operators on $L^2(\Delta)$ and the off-diagonal $M_{12},M_{21}$ are Hilbert-Schmidt class integral operators on $L^2(\Delta)$, we consider the \textit{regularised $2$-determinant}, cf. \cite[Chapter $9$]{S}, of $M$ in the form
\begin{equation}\label{a2}
	D_{M,2}(\Delta):=\det\Big((I-M)\e^{M}\Big)\e^{-\textnormal{tr}(M_{11}+M_{22})},
\end{equation}
with ``$\det$\,'' as the ordinary Fredholm determinant of $(I-M)\e^{M}=I+\textnormal{trace class on}\,L^2(\Delta)\oplus L^2(\Delta)$. Since $M_{jj}\in\textnormal{trace class on}\,L^2(\Delta)$, the exponential on the far right in \eqref{a2} is well-defined and if all four $M_{ij}\in\textnormal{trace class on}\,L^2(\Delta)$, then
\begin{equation*}
	D_{M,2}(\Delta)=\det(I-M)=:D_M(\Delta),
\end{equation*}
so the regularised $2$-determinant \eqref{a2} extends the definition of $D_M(\Delta)$ to the class of block operators of the form \eqref{a1} with trace class diagonal and Hilbert-Schmidt class off-diagonal. Moving ahead, let us further consider a $2\times 2$ matrix-valued skew-symmetric kernel
\begin{equation}\label{a3}
	K(x,y)=\begin{bmatrix}K_{11}(x,y) & K_{12}(x,y)\smallskip\\ K_{21}(x,y) & K_{22}(x,y)\end{bmatrix},\ \ \ x,y\in\Delta,
\end{equation}
that induces an integral operator $K$ on $L^2(\Delta)\oplus L^2(\Delta)$. Skew-symmetry of $K(x,y)$ means
\begin{equation}\label{a4}
	K_{11}(x,y)=-K_{11}(y,x),\ \ \ \ K_{22}(x,y)=-K_{22}(y,x),\ \ \ \ K_{21}(x,y)=-K_{12}(y,x)\ \ \ \textnormal{on}\ \ \Delta\times\Delta.
\end{equation}
Assuming the off-diagonal $K_{12},K_{21}$ are trace class integral operators on $L^2(\Delta)$ and $K_{11},K_{22}$ are Hilbert-Schmidt on the same space, we define the \textit{regularised $2$-Pfaffian}, cf. \cite[(B.21)]{OQR} of $K$ in the form
\begin{equation}\label{a5}
	P_{K,2}(\Delta):=\pf\Big(\e^{-\frac{1}{2}JK}(J+JKJ)\e^{-\frac{1}{2}KJ}\Big)\e^{\frac{1}{2}\textnormal{tr}(K_{21}-K_{12})},
\end{equation}
with ``$\textnormal{pf}$\,'' as the ordinary Fredholm Pfaffian, see \cite[Lemma $8.1$]{R}, \cite[Proposition B.4]{OQR} or \eqref{p1}, of 
\begin{equation*}
	\e^{-\frac{1}{2}JK}(J+JKJ)\e^{-\frac{1}{2}KJ}=J+\textnormal{trace class on}\,L^2(\Delta)\oplus L^2(\Delta).
\end{equation*}
In \eqref{a5}, $J$ denotes the integral operator on $L^2(\Delta)\oplus L^2(\Delta)$ with distributional kernel 
\begin{equation*}
	J(x,y)=\begin{bmatrix}0&\delta_y(x)\\ -\delta_y(x) & 0\end{bmatrix},\ \ \ \ \ \ \ \ \ \int_{\Delta}f(x)\delta_y(x)\d x=f(y),\ \ \ y\in\Delta\cup\partial\Delta,
\end{equation*}
the exponent in the far right of \eqref{a5} is well-defined by assumption, and if all $K_{ij}\in\textnormal{trace class on}\,L^2(\Delta)$, then
\begin{equation*}
	P_{K,2}(\Delta)\stackrel{\eqref{a5}}{=}\textnormal{pf}\Big(\e^{-\frac{1}{2}JK}(J+JKJ)\e^{-\frac{1}{2}KJ}\Big)\e^{\frac{1}{2}\textnormal{tr}(JK)}=\textnormal{pf}(J+JKJ)=\textnormal{pf}(J-K),
\end{equation*}
where \cite[(B.19),(B.20)]{OQR} are used in the second and third equality together with the relation $(JK)^{\ast}=KJ$ for the real adjoint of $JK$. Consequently, the regularised $2$-Pfaffian \eqref{a5} extends the definition of $\textnormal{pf}(J-K)$ to the class of block operators \eqref{a3} with trace class off-diagonal and Hilbert-Schmidt class diagonal. What's more, the square of the regularised $2$-Pfaffian \eqref{a5} is equal to a regularised $2$-determinant \eqref{a2}, namely
\begin{equation}\label{a6}
	\big(P_{K,2}(\Delta)\big)^2=D_{M,2}(\Delta),\ \ \ \ \ \ M(x,y):=-(JK)(x,y)=\begin{bmatrix}-K_{21}(x,y) & -K_{22}(x,y)\smallskip\\ K_{11}(x,y) & K_{12}(x,y)\end{bmatrix},\ \ \ \ x,y\in\Delta,
\end{equation}
which is a direct consequence of \cite[Proposition B.4]{OQR}. In short, $P_{K,2}$ provides a square root of $D_{M,2}$, generalising Cayley's $1847$ result on Pfaffians of skew-symmetric matrices.
\section{Main objectives and first results}
We seek classes of integral operators $K_{ij}$ on $L^2(\Delta)$ for which the regularised $2$-Pfaffian, or $2$-determinant, in \eqref{a6} admits a Riemann-Hilbert characterisation; in the sense that we want the same operator Pfaffian, for those $K_{ij}$, to be computable in terms of a canonical RHP. That this ought to be possible for suitable kernel classes was hinted at in recent years in theoretical physics, cf. \cite{Kra}, in probability theory, cf. \cite{FTZ}, and in integrable systems theory, cf. \cite{Bo}. And prior to those works, several special cases of the general framework that we offer here were also investigated - we will cite those special cases as we progress. For now we begin by singling out two classes of $M$ or $K$ in \eqref{a6} for which commutation \eqref{c2} simplifies the underlying Fredholm Pfaffian.

\subsection*{The symplectic derived class}
Assume $K$ to be of \textit{symplectic derived type}, cf. \cite[Definition $3.9.18$]{AGZ}, which says that $K_{ij}$ in \eqref{a3} are of the form
\begin{equation}\label{z3}
	K_{21}(x,y)=-S(x,y),\ \ \ \ K_{22}(x,y)=\frac{\partial}{\partial y}S(x,y),\ \ \ \ \frac{\partial}{\partial x}K_{11}(x,y)=S(x,y)\ \ \ \textnormal{on}\ \ \Delta\times\Delta,
\end{equation}
in terms of a function $S:\Delta\times\Delta\rightarrow\mathbb{R}$, the \textit{main kernel}, that is continuously differentiable on $\Delta\times\Delta$, that obeys the dominance assumptions
\begin{equation}\label{z4}
	|S(x,y)|\leq g_1(y),\ \ \ \bigg|\frac{\partial}{\partial x}S(y,x)\bigg|\leq g_2(y)\ \ \ \ \ \ \forall\,(x,y)\in \Delta\times \Delta\ \ \ \textnormal{with}\ \ g_k\in L^2(\Delta),
\end{equation}
which is such that $K_{ij}$ satisfy \eqref{a4}, and such that the kernels $K_{ij}(x,y)$ induce trace class integral operators on $L^2(\Delta)$. In turn, the operator $M$ induced by the kernel $M(x,y)$ in \eqref{a6} can be written in block form
\begin{equation}\label{z5}
	M=\begin{bmatrix}S & G\\ H& S^{\ast}\end{bmatrix}:L^2(\Delta)\oplus L^2(\Delta)\rightarrow L^2(\Delta)\oplus L^2(\Delta),
\end{equation}
where $G,H:L^2(\Delta)\rightarrow L^2(\Delta)$ have skew-symmetric, trace class kernels 
\begin{equation}\label{z6}
	G(x,y)=-\frac{\partial}{\partial y}S(x,y),\ \ \ \ \ \ \frac{\partial}{\partial x}H(x,y)=S(x,y)\ \ \ \ \textnormal{on}\ \Delta\times\Delta,
\end{equation}
and $S^{\ast}$ is the real adjoint of $S$. The family of kernels \eqref{z3} is distinguished in that, by assumption $M_{ij}\in\textnormal{trace class on}\,L^2(\Delta)$, the regularised $2$-determinant of $M$ equals its ordinary Fredholm determinant
\begin{equation*}
	D_{M,2}(\Delta)=D_M(\Delta).
\end{equation*}
Furthermore we can algebraically simplify the Fredholm determinant of $M$ via commutation \eqref{c2} and we can relate it to Riemann-Hilbert problems for certain types of main kernels, compare Sections \ref{sec13} and \ref{sec14} below. 
\begin{rem} Our upcoming simplification of \eqref{a6} and connection to RHPs was first put forth in \cite[Result $2.13$]{Kra} for one specific class of $S$, and with $\Delta$ a single semi-infinite interval, see \cite[$(2.32)$]{Kra} and \eqref{c8}. On the other hand, for any bounded single interval $\Delta$ and any symplectic derived kernel \eqref{z3}, \cite[$(3.9.75)$]{AGZ} presents an algebraic simplification of \eqref{a6}, but no subsequent RHP connection. Of course, the first concrete $S$ on a single interval $\Delta$ for which \eqref{a6} was algebraically simplified using \eqref{c2} is due to Tracy and Widom in \cite{TW}. 
\end{rem}
Commutation \eqref{c2} gives the following.

\begin{prop} Let $D_M(\Delta)$ denote the Fredholm determinant of $M$ in \eqref{z5} on $L^2(\Delta)\oplus L^2(\Delta)$, with the aforementioned assumptions placed on $M_{ij}$. Then we have the identity
\begin{equation}\label{c3}
	D_M(\Delta)=D_{2S}(\Delta)\det\big(\delta_{jk}-F_{jk}(\Delta)\big)_{j,k=1}^{2m},\ \ \ \ \ \textnormal{with}\ \ F_{jk}(\Delta):=(-1)^k\big((I-2S^{\ast})^{-1}H\big)(a_k,a_j),
\end{equation}
provided $I-2S$ is invertible on $L^2(\Delta)$ and where $D_{2S}(\Delta)$ is the Fredholm determinant of the trace class operator $2S$ acting on $L^2(\Delta)$.
\end{prop}
One can surely view \eqref{c3} as an algebraic simplification of the initial $D_M(\Delta)$, as the right-hand side in \eqref{c3} involves a finite-size determinant and a Fredholm determinant on $L^2(\Delta)$ instead of one on $L^2(\Delta)\oplus L^2(\Delta)$. Still, any further simplification of \eqref{c3} appears to be out of reach unless more assumptions are placed on the main kernel $S(x,y)$, precisely as we will do it in Sections \ref{sec13} and \ref{sec14}. For now, we discuss our second class of $M$ or $K$ in \eqref{a6} to which commutation is applicable.

\subsection*{The orthogonal derived class}
Assume $K$ to be of \textit{orthogonal derived type}, which says that $K_{ij}$ are of the form
\begin{align}
	K_{21}(x,y)=-S(x,y),\ \ \ \ \ K_{22}(x,y)=&\,\frac{\partial}{\partial y}S(x,y),\nonumber\\
	K_{11}(x,y)=&\,H(x,y)-\epsilon(x-y)\ \ \ \ \textnormal{where}\ \ \frac{\partial}{\partial x}H(x,y)=S(x,y),\label{c4}
\end{align}
in terms of a function $S:\Delta\times\Delta\rightarrow\mathbb{R}$, the \textit{main kernel}, that is continuously differentiable on $\Delta\times\Delta$, that obeys assumptions \eqref{z4}, which is such that $K_{ij}$ satisfy \eqref{a4}, and such that the kernels $K_{21}(x,y),K_{22}(x,y)$ and $H(x,y)$ induce trace class integral operators on $L^2(\Delta)$. In \eqref{c4},
\begin{equation*}
	\epsilon(x):=\frac{1}{2}\textnormal{sgn}(x)=\begin{cases}+\frac{1}{2},&x>0\smallskip\\ -\frac{1}{2},&x<0\end{cases},\ \ \ \ \ \ \ \ \epsilon(0):=0.
\end{equation*}
What results is the following block composition form for $M$ in \eqref{a6},
\begin{equation}\label{c5}
	M=\begin{bmatrix}S&G\\ H-\epsilon & S^{\ast}\end{bmatrix}:L^2(\Delta)\oplus L^2(\Delta)\rightarrow L^2(\Delta)\oplus L^2(\Delta),
\end{equation}
where $\epsilon:L^2(\Delta)\rightarrow L^2(\Delta)$ has the Hilbert-Schmidt kernel $\epsilon(x-y)$ and where $G,H:L^2(\Delta)\rightarrow L^2(\Delta)$ have skew-symmetric kernels
\begin{equation}\label{c6}
	G(x,y)=-\frac{\partial}{\partial y}S(x,y),\ \ \ \ \ \ \ \frac{\partial}{\partial x}H(x,y)=S(x,y)\ \ \ \ \ \textnormal{on}\ \Delta\times\Delta.
\end{equation}
Also, by assumption, $M_{jj}\in\textnormal{trace class on}\,L^2(\Delta)$ and $M_{12},M_{21}\in\textnormal{Hilbert-Schmidt class on}\,L^2(\Delta)$; i.e. the regularised $2$-determinant $D_{M,2}(\Delta)$ of $M$ is well-defined and we now set out to simplify the same algebraically. 
\begin{rem} Special cases of our forthcoming algebraic simplification \eqref{c7} have previously appeared in \cite[$(2.34)$]{Kra} for a specific class of $S$, with $\Delta$ a single semi-infinite interval and in \cite[$(3.9.74)$]{AGZ} for any orthogonal derived kernel \eqref{c4} on a bounded single interval $\Delta$. The first concrete $S$ on a single interval $\Delta$ for which \eqref{c7} was computed appeared in \cite{TW}.
\end{rem}
Commutation \eqref{c2} gives the following.
\begin{prop} Let $D_{M,2}(\Delta)$ denote the regularised $2$-determinant of $M$ in \eqref{c5} on $L^2(\Delta)\oplus L^2(\Delta)$, with the aforementioned assumptions placed on $M_{ij}$. Then we have the identity
\begin{equation}\label{c7}
	D_{M,2}(\Delta)=D_S(\Delta)\det\big(\delta_{jk}-G_{jk}(\Delta)\big)_{j,k=1}^{2m},
\end{equation}
with
\begin{equation*}
	G_{jk}(\Delta):=(-1)^k\big((I-S^{\ast})^{-1}H\big)(a_j,a_k)-\frac{1}{2}\sum_{\ell=1}^{2m}\sigma_{\ell}(k)\big((I-S^{\ast})^{-1}H\big)(a_j,a_{\ell}),
\end{equation*}
provided $I-S$ is invertible on $L^2(\Delta)$ and where $D_S(\Delta)$ is the Fredholm determinant of the trace class operator $S$ on $L^2(\Delta)$. The sign coefficients $\sigma_j(k)\in\{\pm 1\}$ are defined in \eqref{z555}.
\end{prop}
Similar to \eqref{c3}, any further simplification of \eqref{c7} is seemingly hard to achieve unless additional assumptions are placed on $S$. We will do that momentarily.
\section{Hankel composition type and second results}\label{sec13} Both \eqref{c3} and \eqref{c7} can be further simplified, subsequently connected to a canonical RHP, and thus subjected to an asymptotic analysis, provided $\Delta=(t,\infty)\subset\mathbb{R}$ with $t\in\mathbb{R}$ and that the main kernel $S(x,y)$ in \eqref{z3},\eqref{c4} is of \textit{additive Hankel composition} type. What that means is laid out below; first for the symplectic derived class \eqref{z5}, then for the orthogonal derived one \eqref{c5}.
\subsection*{In the symplectic derived class}
Consider in \eqref{z5} the rule
\begin{equation}\label{c8}
	S(x,y):=\frac{1}{2}\bigg(\int_0^{\infty}\phi(x+u)\phi(u+y)\d u-\frac{1}{2}\phi(x)\int_y^{\infty}\phi(z)\d z\bigg),\ \ \ \ \ x,y\in\Delta=(t,\infty),
\end{equation}
with $\phi:\mathbb{R}\rightarrow\mathbb{R}$ continuously differentiable, both $\phi$ and its derivative $D\phi$ decaying exponentially fast at $+\infty$ and $D\phi$ obeying the dominance assumption
\begin{equation}\label{c9}
	|D\phi(x+y)|\leq g_3(x)\ \ \ \ \ \forall\,(x,y)\in(0,\infty)\times\mathbb{R}\ \ \ \ \textnormal{with}\ \ \ \ g_3\in L^2(0,\infty).
\end{equation}
In addition, let $H(x,y)$ in \eqref{z6} be of the form
\begin{equation}\label{c10}
	H(x,y):=-\int_x^{\infty}S(z,y)\d z,\ \ \ \ \ x,y\in\Delta=(t,\infty).
\end{equation}
If $S,G,H:L^2(t,\infty)\rightarrow L^2(t,\infty)$ denote the integral operators with kernels \eqref{c8},\eqref{z6} and \eqref{c10}, then Proposition \ref{pz3} establishes that all three are trace class and that $G(x,y),H(x,y)$ are skew-symmetric. So, \eqref{c8},\eqref{z6} and \eqref{c10} are an admissible choice in \eqref{z5} and we can simplify \eqref{c3} further.
\begin{cor}\label{beta4} Let $D_M(t)=\lim_{a_2\rightarrow\infty}D_M((t,a_2)),t\in\mathbb{R}$ denote the Fredholm determinant of $M$ in \eqref{z5} with $S,G,H$ as in \eqref{c8},\eqref{z6} and \eqref{c10}, with the aforementioned assumptions placed on $\phi$. Then
\begin{equation}\label{c11}
	D_M(t)=D_Q(t)\left[1+\frac{1}{2}\int_t^{\infty}\big((I-Q)^{-1}\phi\big)(x)\Phi(x)\d x\right],\ \ \ \ \ \Phi(x):=\int_x^{\infty}\phi(z)\d z,
\end{equation}
provided $I-2S$ and $I-Q$ are invertible on $L^2(t,\infty)$. Here, $Q:L^2(t,\infty)\rightarrow L^2(t,\infty)$ is the integral operator with trace class kernel
\begin{equation}\label{c12}
	Q(x,y):=\int_0^{\infty}\phi(x+u)\phi(u+y)\d u,\ \ \ \ \ x,y\in(t,\infty),
\end{equation}
and $D_Q(t)$ is the Fredholm determinant of $Q$ acting on $L^2(t,\infty)$.
\end{cor}
Result \eqref{c11} was first derived in theoretical physics, see \cite[$(2.32)$]{Kra}, for the class \eqref{c8}. Tracy and Widom knew \eqref{c11} for the Airy function choice of $\phi$ since the mid $1990$s, see \cite{TW}. One can derive \eqref{c11}, for the Airy function choice, from \cite[Lemma $3.9.41$]{AGZ}, too. We obtain \eqref{c11} directly from \eqref{c3} by an application of the Sherman-Morrison identity. 

\subsection*{In the orthogonal derived class}
Consider, in \eqref{c5}, the rule
\begin{equation}\label{c13}
	S(x,y):=\int_0^{\infty}\phi(x+u)\phi(u+y)\d u+\frac{1}{2}\phi(x)\bigg(1-\int_y^{\infty}\phi(z)\d z\bigg),\ \ \ \ \ x,y\in\Delta=(t,\infty),
\end{equation}
with $\phi:\mathbb{R}\rightarrow\mathbb{R}$ continuously differentiable, both $\phi$ and and its derivative $D\phi$ decaying exponentially fast at $+\infty$, and $D\phi$ obeying the dominance assumption \eqref{c9}. In addition, let $H(x,y)$ in \eqref{c6} be given by
\begin{equation}\label{c14}
	H(x,y):=-\int_x^{\infty}\left[\int_0^{\infty}\phi(z+u)\phi(u+y)\d u\right]\d z+\frac{1}{2}\int_y^x\phi(z)\d z+\frac{1}{2}\int_x^{\infty}\phi(z)\d z\int_y^{\infty}\phi(z)\d z,
\end{equation}
for $x,y\in\Delta=(t,\infty)$. What appears in \eqref{c13}, different from \eqref{c8}, is the map
\begin{equation*}
	(t,\infty)\ni x\mapsto 1-\int_x^{\infty}\phi(z)\d z,
\end{equation*}
which is in $L^{\infty}(t,\infty)$, but \textit{not} in $L^2(t,\infty)$. Also, $\epsilon:L^2(\Delta)\rightarrow L^2(\Delta)$ in \eqref{c5} is no longer Hilbert-Schmidt once $\Delta=(t,\infty)$. So, we first have to make sense of the regularised $2$-determinant of $M$; this cannot be achieved by simply taking $m=1,a_1=t$ and letting $a_2\rightarrow\infty$ in \eqref{c7} (as we did it in Corollary \ref{beta4} in the context of \eqref{c3}). Rather, compare \cite[Section $4.1.3$]{DG}, we first choose any bounded $\rho:\mathbb{R}\rightarrow(0,\infty)$ such that $\rho\in L^2(\mathbb{R})$ and $\frac{1}{\rho}$ is polynomially bounded. After that, instead of \eqref{c5}, we work with
\begin{equation}\label{c15}
	M=\begin{bmatrix}\rho^{-1}S\rho & \rho^{-1}G\rho^{-1}\\ \rho H\rho-\rho\epsilon\rho & \rho S^{\ast}\rho^{-1}\end{bmatrix}:L^2(t,\infty)\oplus L^2(t,\infty)\rightarrow L^2(t,\infty)\oplus L^2(t,\infty),
\end{equation}
where $\rho,\rho^{-1}$ act via multiplication by $\rho(x),\frac{1}{\rho(x)}$. Here, $\rho\epsilon\rho:L^2(t,\infty)\rightarrow L^2(t,\infty)$ has Hilbert-Schmidt kernel $\rho(x)\epsilon(x-y)\rho(y)$ (fortunately!) and the kernels
\begin{equation*}
	\frac{1}{\rho(x)}S(x,y)\rho(y),\ \ \ \ \frac{1}{\rho(x)}G(x,y)\frac{1}{\rho(y)},\ \ \ \ \rho(x)H(x,y)\rho(y),
\end{equation*}
in terms of \eqref{c13},\eqref{c6} and \eqref{c14}, induce trace class operators $\rho^{-1}S\rho,\rho^{-1}G\rho^{-1},\rho H\rho:L^2(t,\infty)\rightarrow L^2(t,\infty)$. This is shown in Proposition \ref{conjkernel} below and it makes the regularised $2$-determinant $D_{M,2}(t)$ of $M$ in \eqref{c15} on $L^2(t,\infty)\oplus L^2(t,\infty)$ well-defined. What's more, we can simplify $D_{M,2}(t)$  as follows, utilising \eqref{c7}:
\begin{cor} Let $D_{M,2}(t),t\in\mathbb{R}$ denote the regularised $2$-determinant of $M$ in \eqref{c15} with $S,G,H$ as in \eqref{c13},\eqref{c6},\eqref{c14}, with the aforementioned assumptions place on $\phi$. Then
\begin{equation}\label{c16}
	D_{M,2}(t)=D_Q(t)\left[1-\int_t^{\infty}\big((I-Q)^{-1}\phi\big)(x)\big(1-\Phi(x)\big)\d x\right],\ \ \ \ \Phi(x)=\int_x^{\infty}\phi(z)\d z,
\end{equation}
provided $I-S$ is invertible on $L^2(t,\infty)$ and $\|Q\|<1$ in operator norm on $L^2(t,\infty)$. Here, $D_Q(t)$ is the Fredholm determinant of $Q$, with kernel \eqref{c12}, acting on $L^2(t,\infty)$.
\end{cor}
A non-rigorous derivation of \eqref{c16} can be found in \cite[$(2.34)$]{Kra}, for the class \eqref{c13}. The special case of \eqref{c16} for $\phi$ the Airy function is due to \cite{TW} and can also be distilled from \cite[Lemma $3.9.41$]{AGZ}. We obtain \eqref{c16} through \eqref{c7}, via a limit argument,
\begin{equation*}
	D_{M,2}(t)=\lim_{a_2\rightarrow\infty}D_{M,2}((t,a_2)),
\end{equation*}
where $D_{M,2}((t,a_2))$ is the regularised $2$-determinant of $M$ in \eqref{c15} on $L^2(t,a_2)\oplus L^2(t,a_2)$ with $a_2>t$ finite. Expectedly, our derivation shows that $D_{M,2}(t)$ is independent of any concrete choice of $\rho$ in \eqref{c15}, and hence $\rho$ is suppressed from our notations.
\subsection*{RHP and asymptotic results}
%Next to the algebraic simplifications \eqref{c11} and \eqref{c16}, 
The Hankel composition structure \eqref{c12} that appears in both \eqref{c11} and \eqref{c16} makes it possible to access the factors in the right hand sides of these equations through a canonical RHP. This was originally observed in \cite[Section $5.2$]{Kra} and further worked out in \cite{Bo}. To formulate the relevant RHP we require the Sobolev space
\begin{equation*}
	W^{1,1}(\mathbb{R}):=\big\{f:\,\,f,Df\in L^1(\mathbb{R})\big\},
\end{equation*}
with $D$ as weak derivative and $f$ real-valued.
\begin{problem}[Hankel RHP]\label{HankelRHP} Let $t\in\mathbb{R}$ and $\phi\in W^{1,1}(\mathbb{R})\cap L^{\infty}(\mathbb{R})$. Find $X(z)=X(z;t,\phi)\in\mathbb{C}^{2\times 2}$ so that
\begin{enumerate}
	\item[(1)] $z\mapsto X(z)$ is analytic for $z\in\mathbb{C}\setminus\mathbb{R}$.
	\item[(2)] $z\mapsto X(z)$ admits continuous pointwise limits $X_{\pm}(z):=\lim_{\epsilon\downarrow 0}X(z\pm\im\epsilon),z\in\mathbb{R}$ that satisfy
	\begin{equation*}
		X_+(z)=X_-(z)\begin{bmatrix}1-|r(z)|^2 & -\bar{r}(z)\e^{-\im tz}\smallskip\\ r(z)\e^{\im tz} & 1\end{bmatrix};\ \ \ \ \ \ \ r(z):=-\im\int_{-\infty}^{\infty}\phi(y)\e^{-\im zy}\d y,\ z\in\mathbb{R}.
	\end{equation*}
	\item[(3)] As $z\rightarrow\infty$ in $\mathbb{C}\setminus\mathbb{R}$,
	\begin{equation*}
		X(z)=\mathbb{I}+\frac{1}{z}X_1+o\big(z^{-1}\big);\ \ \ \ \ \ \ X_1=\big[X_1^{jk}(t,\phi)\big]_{j,k=1}^2.
	\end{equation*}
\end{enumerate}
\end{problem}
Provided $I-Q$ is invertible on $L^2(t,\infty)$, RHP \ref{HankelRHP} is uniquely solvable by \cite[Theorem $2.7$]{Bo} and we have
\begin{equation}\label{c17}
	\frac{\d}{\d t}\ln D_Q(t)=\im X_1^{11}(t,\phi),\hspace{1cm}X_1^{21}(t,\phi)=X_1^{12}(t,\phi)=\big((I-Q)^{-1}\phi\big)(t)=:q(t).
\end{equation}
Consequently, the first factor in the right hand sides of \eqref{c11} and \eqref{c16} admits a RHP characterisation, and the same goes for the second factors, too. Indeed, by \cite[Theorem $2.12$]{Bo}, for any $t\in\mathbb{R}$,
\begin{align*}
	1+\frac{1}{2}\int_t^{\infty}\big((I-Q)^{-1}\phi\big)(x)\Phi(x)\d x=\cosh^2\bigg(\frac{\omega(t)}{2}\bigg),\
	1-\int_t^{\infty}\big((I-Q)^{-1}\phi\big)(x)\big(1-\Phi(x)\big)\d x=\e^{-\omega(t)},
\end{align*}
in terms of $\omega(t):=\int_t^{\infty}q(s)\d s$. The RHP formulation above is convenient for asymptotic studies, especially for $D_M(t),D_{M,2}(t)$ in \eqref{c11},\eqref{c16} when $t\rightarrow-\infty$. The regime $t\rightarrow+\infty$ is much easier and can be dealt with by trace norm estimates and Neumann series arguments. We will not address this regime in the following. Instead, using a Deift-Zhou nonlinear steepest descent analysis \cite{DZ}, we obtain our first two asymptotic results.
\begin{theo}\label{Kactheo0} Suppose $\phi\in W^{1,1}(\mathbb{R})\cap L^{\infty}(\mathbb{R})$ satisfies
\begin{equation}\label{c18}
	|\phi(x)|\leq\e^{-a|x|},\ \ \ \ \ |D\phi(x)|\leq\e^{-a|x|}\ \ \ \ \ \forall\,x\in\mathbb{R},
\end{equation}
with $a\geq 2+\epsilon$ for some $\epsilon>0$. If $\|Q\|<1$ in operator norm on $L^2(t,\infty)$, then $D_M(t)$ in \eqref{c11} is given asymptotically as $t\rightarrow-\infty$ by
\begin{align}
	\ln D_M(t)=ts(0)+\int_0^{\infty}xs(x)s(-x)\d x-\frac{1}{2\pi}&\,\int_{-\infty}^{\infty}\Im\bigg\{\frac{r'(\lambda)}{r(\lambda)}\bigg\}\ln\big(1-|r(\lambda)|^2\big)\d\lambda\nonumber\\
	&+2\ln\Bigg(\frac{1}{2}\sqrt[4]{\frac{1+\im r(0)}{1-\im r(0)}}+\frac{1}{2}\sqrt[4]{\frac{1-\im r(0)}{1+\im r(0)}}\Bigg)+\mathcal{O}\big(t^{-\infty}\big),\label{c19}
\end{align}
and $D_{M,2}(t)$ in \eqref{c16} by
\begin{align}
	\ln D_{M,2}(t)=ts(0)+\int_0^{\infty}xs(x)s(-x)\d x-\frac{1}{2\pi}&\,\int_{-\infty}^{\infty}\Im\bigg\{\frac{r'(\lambda)}{r(\lambda)}\bigg\}\ln\big(1-|r(\lambda)|^2\big)\d\lambda\nonumber\\
	&+\frac{1}{2}\ln\bigg(\frac{1-\im r(0)}{1+\im r(0)}\bigg)+\mathcal{O}\big(t^{-\infty}\big),\label{c20}
\end{align}
where throughout, for $x,\lambda\in\mathbb{R}$,
\begin{equation*}
	s(x)=-\frac{1}{2\pi}\int_{-\infty}^{\infty}\ln\big(1-|r(\lambda)|^2\big)\e^{\im x\lambda}\d\lambda,\ \ \ \ \ \ \ r(\lambda)=-\im\int_{-\infty}^{\infty}\phi(y)\e^{-\im\lambda y}\d y.
\end{equation*}
\end{theo}
Estimates \eqref{c19} and \eqref{c20} constitute Akhiezer-Kac type asymptotic expansions for $D_M(t)$ and $D_{M,2}(t)$. We derive them through a nonlinear steepest descent analysis of RHP \ref{HankelRHP}, for which \eqref{c18} is a convenient assumption. But \eqref{c18} is by no means necessary to obtain the leading terms in \eqref{c19},\eqref{c20}, as first discovered in \cite{FTZ} for \eqref{c20}.
\begin{rem} In order to compare \eqref{c20} with \cite[Theorem $3$]{FTZ} one first has to compare \eqref{a3},\eqref{c5},\eqref{c23} to \cite[$(10)$]{FTZ}. What results is $\rho(x)=\phi(x)$ for $\rho$ in \cite[$(10)$]{FTZ} which needs to obey $\int_{-\infty}^{\infty}\rho(x)\d x=1$. That, in turn, would mean $r(0)=-\im$ which is inadmissible in Theorem \ref{Kactheo} because of \eqref{c18}, compare the discussion immediately below \eqref{c35}. In words, there is no direct identification between $\phi$ in Theorem \ref{Kactheo} and $(p,\rho)$ in \cite{FTZ}. This explains why the fourth term in \eqref{c20} is different from the non-integral terms in \cite[Theorem $3$]{FTZ}, for we are working with a class of $\phi=\rho$ that is inadmissible in \cite[Theorem $3$]{FTZ}. Still, our first three terms in \eqref{c20} match  exactly with the sum and integral terms in \cite[Theorem $3$(i)]{FTZ}. This is to be expected by \cite[$(96)$]{FTZ} and was previously observed in \cite[Remark $2.18$]{Bo}.

\end{rem}

\begin{rem} From \eqref{c19} and \eqref{c20}, together with \eqref{a6}, we easily deduce the $t\rightarrow-\infty$ expansions for the corresponding Pfaffians $P_K(t)=P_K((t,\infty))$ and $P_{K,2}(t)=P_{K,2}((t,\infty))$.
\end{rem}

\section{Truncated Wiener-Hopf type and third results}\label{sec14} There exists at least one other class for which \eqref{c3} and \eqref{c7} can be further simplified and subsequently connected to a RHP. Namely, it works for $\Delta=(-t,t)\subset\mathbb{R}$ with $t>0$ and for a main kernel $S(x,y)$ in \eqref{z3},\eqref{c4} that is of \textit{truncated Wiener-Hopf} type. In detail, we consider the following two cases.
\subsection*{In the symplectic derived class}
Back in \eqref{z5}, choose
\begin{equation}\label{c21}
	S(x,y):=\frac{1}{4\pi}\int_{-\infty}^{\infty}\phi(u)\e^{\im u(x-y)}\d u,\ \ \ \ \ x,y\in\Delta=(-t,t),
\end{equation}
where $\phi:\mathbb{R}\rightarrow\mathbb{R}$ is an even function such that $\phi\in L^1(\mathbb{R})\cap L^{\infty}(\mathbb{R})$ and $\mathbb{R}\ni x\mapsto x\phi(x)\in L^1(\mathbb{R})\cap L^{\infty}(\mathbb{R})$. In addition, we take $H(x,y)$ in \eqref{z6} as
\begin{equation}\label{c22}
	H(x,y):=\int_y^xS(z,y)\d z,\ \ \ \ \ \ x,y\in\Delta=(-t,t).
\end{equation}
With $S,G,H:L^2(-t,t)\rightarrow L^2(-t,t)$ as integral operators induced by the kernels \eqref{c21},\eqref{z6},\eqref{c22}, Proposition \ref{pz4} ensures that they are trace class and that $G(x,y),H(x,y)$ are skew-symmetric kernels. Thus, \eqref{c21},\eqref{z6} and \eqref{c22} are admissible in \eqref{c5}, and \eqref{c3} simplifies as follows.
\begin{cor}\label{beta5} Let $D_M(t)=D_M((-t,t)),t>0$ denote the Fredholm determinant of $M$ in \eqref{z5} with $S,G,H$ as in \eqref{c21},\eqref{z6},\eqref{c22}, with the aforementioned assumptions placed on $\phi$. Then
\begin{align}
	D_M(t)=D_Q(t)&\left[1+\frac{1}{2}\int_{-t}^t\big((I-Q)^{-1}Q\big)(x,t)\d x\right]\nonumber\\
	&\hspace{0.5cm}\times\left[1-\frac{1}{2}\int_{-t}^tQ(x,t)\d x-\frac{1}{2}\int_{-t}^t\left(\int_{-x}^xQ(z,t)\d z\right)\big((I-Q)^{-1}Q\big)(x,t)\d x\right],\label{c23}
\end{align}
provided $I-2S$ is invertible on $L^2(-t,t)$. Here, $Q:L^2(-t,t)\rightarrow L^2(-t,t)$ is the integral operator with trace class kernel
\begin{equation}\label{c24}
	Q(x,y):=\frac{1}{2\pi}\int_{-\infty}^{\infty}\phi(u)\e^{\im u(x-y)}\d u\stackrel{\eqref{c21}}{=}2S(x,y),\ \ \ \ \ \ x,y\in(-t,t),
\end{equation}
and $D_Q(t)$ is the Fredholm determinant of $Q$ acting on $L^2(-t,t)$.
\end{cor}
A special case of \eqref{c23}, namely $\phi(x)=1$ for $x\in[-1,1]$ and $\phi(x)=0$ for $x\notin[-1,1]$, was known from \cite[$(3.9.62),(3.9.82)$]{AGZ}. We obtain \eqref{c23} for the full class of truncated Wiener-Hopf operators \eqref{c21}, a seemingly novel result, from \eqref{c3}.

\subsection*{In the orthogonal derived class}

Instead of \eqref{c21}, take
\begin{equation}\label{c25}
	S(x,y):=\frac{1}{2\pi}\int_{-\infty}^{\infty}\phi(u)\e^{\im u(x-y)}\d u\stackrel{\eqref{c24}}{=}Q(x,y),\ \ \ \ \ x,y\in\Delta=(-t,t),
\end{equation}
where $\phi:\mathbb{R}\rightarrow\mathbb{R}$ is even so that $\phi\in L^1(\mathbb{R})\cap L^{\infty}(\mathbb{R})$ and $x\phi(x)\in L^1(\mathbb{R})\cap L^{\infty}(\mathbb{R})$. The integral operators $S,G,H:L^2(-t,t)\rightarrow L^2(-t,t)$ with kernels \eqref{c25},\eqref{z6} and \eqref{c22} are trace class, compare Proposition \ref{pz4}, and $G(x,y),H(x,y)$ are skew-symmetric kernels. So, the regularised $2$-determinant $D_{M,2}(t)=D_{M,2}((-t,t))$ of $M$ in \eqref{c15} is well-defined without needing any further steps, and the same simplifies as follows.
\begin{cor}\label{beta6} Let $D_{M,2}(t)=D_{M,2}((-t,t)),t>0$ denote the regularised $2$-determinant of $M$ in \eqref{c5} with $S,G,H$ as in \eqref{c25},\eqref{z6},\eqref{c22}, with the aforementioned assumptions placed on $\phi$. Then
\begin{equation}\label{c26}
	D_{M,2}(t)=D_Q(t)\left[1-\int_{-t}^tQ(x,t)\d x-\int_{-t}^t\left(\int_{-x}^xQ(z,t)\d z\right)\big((I-Q)^{-1}Q\big)(x,t)\d x\right],
\end{equation}
provided $I-Q$ is invertible on $L^2(-t,t)$. Here, $D_Q(t)$ is the Fredholm determinant of $Q$, with kernel \eqref{c24}, acting on $L^2(-t,t)$.
\end{cor}
A derivation of \eqref{c26}, for the aforementioned indicator function $\phi(x)=1$ for $x\in[-1,1]$ and $\phi(x)=0$ for $x\notin[-1,1]$, can be found in \cite[$(3.9.62),(3.9.81)$]{AGZ}. The general case above, in the family \eqref{c25}, is seemingly novel and will follow from \eqref{c7}.

\subsection*{RHP and asymptotic results}
We exploit the translation-invariant Wiener-Hopf structure \eqref{c21},\eqref{c25} in the derivations of \eqref{c23} and \eqref{c26}, and the same makes it also possible to access all factors in the right hand sides of \eqref{c23},\eqref{c26} through a canonical RHP. The relevant matrix factorisation problem appeared first in \cite{Ko}, to the best of our knowledge, and it reads as follows.
\begin{problem}[Wiener-Hopf RHP]\label{WHRHP} Let $t>0$ and $\phi\in L^1(\mathbb{R})\cap L^{\infty}(\mathbb{R})$ be an even real-valued function such that $x\phi(x)\in L^1(\mathbb{R})\cap L^{\infty}(\mathbb{R})$. Find $X(z)=X(z;t,\phi)\in\mathbb{C}^{2\times 2}$ so that
\begin{enumerate}
	\item[(1)] $z\mapsto X(z)$ is analytic for $z\in\mathbb{C}\setminus\mathbb{R}$.
	\item[(2)] $z\mapsto X(z)$ admits pointwise limits $X_{\pm}(z):=\lim_{\epsilon\downarrow 0}X(z\pm\im\epsilon),z\in\mathbb{R}$ almost everywhere that satisfy
	\begin{equation*}
		X_+(z)=X_-(z)\begin{bmatrix}1-\phi(z) & \phi(z)\e^{2\im tz}\smallskip\\ -\phi(z)\e^{-2\im tz} & 1+\phi(z)\end{bmatrix}.
	\end{equation*}
	\item[(3)] As $z\rightarrow\infty$ in $\mathbb{C}\setminus\mathbb{R}$,
	\begin{equation*}
		X(z)=\mathbb{I}+\frac{1}{z}X_1+o\big(z^{-1}\big);\ \ \ \ \ \ \ X_1=\big[X_1^{jk}(t,\phi)\big]_{j,k=1}^2.
	\end{equation*}
\end{enumerate}
\end{problem}
It is known that, provided $I-Q$ is invertible on $L^2(-t,t)$, RHP \ref{WHRHP} is uniquely solvable and
\begin{equation}\label{c27}
	\frac{\d}{\d t}\ln D_Q(t)=-2\im X_1^{11}(t,\phi),\hspace{1cm}X_1^{21}(t,\phi)=-X_1^{12}(t,\phi)=-\im\big((I-Q)^{-1}Q\big)(-t,t)=:\frac{q(t)}{2\im}.
\end{equation}
In turn, by \eqref{c27}, the first factor in the right hand side of \eqref{c23},\eqref{c26} can be accessed via RHP \ref{WHRHP}.
\begin{rem} The first identity in \eqref{c27} follows from \cite[Lemma $1$]{Ko}. Indeed, by using the integral identity
\begin{equation*}
	\frac{1}{\pi}\int_{-\infty}^{\infty}\sin(a\xi)\e^{-\im\xi x}\frac{\d\xi}{\xi}=\begin{cases}1,&|x|<a\\ 0,&|x|>a\end{cases},\ \ \ \ \ \ a>0,\ \ x\in\mathbb{R},
\end{equation*}
see for example \cite[$17.23$]{GR}, one finds on $L^2(-t,t)$ the factorisation
\begin{equation}\label{c28}
	Q=\mathcal{F}MVM^{-1}\mathcal{F}^{-1}
\end{equation}
for $Q:L^2(-t,t)\rightarrow L^2(-t,t)$ with kernel \eqref{c24}, where $\mathcal{F}:L^2(\mathbb{R})\rightarrow L^2(\mathbb{R})$ and $\mathcal{F}^{-1}:L^2(\mathbb{R})\rightarrow L^2(\mathbb{R})$ are the bounded linear Fourier and inverse Fourier operators given by
\begin{equation*}
	(\mathcal{F}f)(x):=\int_{-\infty}^{\infty}f(y)\e^{-\im xy}\d y,\ \ \ \ \ (\mathcal{F}^{-1}f)(y):=\frac{1}{2\pi}\int_{-\infty}^{\infty}f(x)\e^{\im xy}\d x,
\end{equation*}
where $M:L^2(\mathbb{R})\rightarrow L^2(\mathbb{R})$ is multiplication $(Mf)(x):=\sqrt{\phi(x)}f(x)$, with some fixed branch for the square root, and $V:L^2(\mathbb{R})\rightarrow L^2(\mathbb{R})$ has kernel
\begin{equation}\label{c29}
	V(x,y):=\sqrt{\phi(x)}\,\frac{\sin t(x-y)}{\pi(x-y)}\sqrt{\phi(y)}.
\end{equation}
Note that $VM^{-1}:L^2(\mathbb{R})\rightarrow L^2(\mathbb{R})$ is a bounded linear operator, because $\phi\in L^1(\mathbb{R})$, and $V$ is trace class, for
\begin{equation*}
	V(x,y)=\frac{t}{2\pi}\int_{-1}^1\sqrt{\phi(x)}\e^{\im t(x-y)s}\sqrt{\phi(y)}\d s
\end{equation*}
can be realised as a Hilbert-Schmidt composition kernel. Consequently, using commutation \eqref{c2}, $D_Q(t)$ is equal to the Fredholm determinant of $V$ on $L^2(\mathbb{R})$. But $V(x,y)$ in \eqref{c29} is an integrable kernel, and so RHP \ref{WHRHP}, as well as the first identity in \eqref{c27}, follow from \cite{IIKS}. The second identity in \eqref{c27} also follows from \eqref{c28} and \cite{IIKS}.
\end{rem}
To access the second and third factors in \eqref{c23} and \eqref{c26}, we make use of $q(t)$ introduced in \eqref{c27}.
\begin{lem} Let $t>0$ and assume $I-Q$ is invertible on $L^2(-t,t)$, with the kernel of $Q$ being written in \eqref{c24} with the aforementioned assumptions placed on $\phi$. Then, with $\omega(t):=\int_0^tq(s)\d s$, we have
\begin{align}\label{c30}
	1+\frac{1}{2}&\,\int_{-t}^t\big((I-Q)^{-1}Q\big)(x,t)\d x\nonumber\\
	&\times\left[1-\frac{1}{2}\int_{-t}^tQ(x,t)\d x-\frac{1}{2}\int_{-t}^t\bigg(\int_{-x}^xQ(z,t)\d z\bigg)\big((I-Q)^{-1}Q\big)(x,t)\d x\right]=\cosh^2\bigg(\frac{\omega(t)}{2}\bigg),
\end{align}
followed by
\begin{align}\label{c31}
	1-\int_{-t}^tQ(x,t)\d x-\int_{-t}^t\bigg(\int_{-x}^xQ(z,t)\d z\bigg)\big((I-Q)^{-1}Q\big)(x,t)\d x=\e^{-\omega(t)}.
\end{align}
\end{lem}
Identities \eqref{c30},\eqref{c31} are analogues of the identities immediately below \eqref{c17}, for the Wiener-Hopf kernel \eqref{c24}. We will derive them in Section \ref{cool}. Once done, \eqref{c30},\eqref{c31} and \eqref{c27}, through a nonlinear steepest descent resolution of RHP \ref{WHRHP}, yield our next results.
\begin{theo}\label{Kactheo} Suppose $\phi\in L^1(\mathbb{R})\cap L^{\infty}(\mathbb{R})$ is an even function with $x\phi\in L^1(\mathbb{R})\cap L^{\infty}(\mathbb{R})$, with $|\phi(x)|<1$ for $x\in\mathbb{R}$ and %such that $\phi$ is H\"older continuous on $\mathbb{R}$ and 
such that
\begin{equation}\label{c32}
	z\mapsto \phi(z)\,\ \textnormal{is analytic in some neighbourhood $U$ of}\,\,\mathbb{R}\,\,\textnormal{with}\,\lim_{\substack{|z|\rightarrow\infty\\ z\in U}}\phi(z)=0.
\end{equation}
If $\|Q\|<1$ in operator norm on $L^2(-t,t)$, then, asymptotically as $t\rightarrow+\infty$, $D_M(t)$ in \eqref{c23} is given by
\begin{equation}\label{c33}
	\ln D_M(t)=-ts(0)+\frac{1}{4}\int_0^{\infty}xs(x)s(-x)\d x+2\ln\bigg(\frac{1}{2}\sqrt[4]{1-\phi(0)}+\frac{1}{2\sqrt[4]{1-\phi(0)}}\bigg)+\mathcal{O}\big(t^{-\infty}\big),
\end{equation}
and $D_{M,2}(t)$ in \eqref{c26} by
\begin{equation}\label{c34}
	\ln D_{M,2}(t)=-ts(0)+\frac{1}{4}\int_0^{\infty}xs(x)s(-x)\d x+\frac{1}{2}\ln\big(1-\phi(0)\big)+\mathcal{O}\big(t^{-\infty}\big),
\end{equation}
where throughout, for $x\in\mathbb{R}$,
\begin{equation*}
	s(x):=-\frac{1}{\pi}\int_{-\infty}^{\infty}\ln\big(1-\phi(\lambda)\big)\e^{\im x\lambda}\d\lambda.
\end{equation*}
\end{theo}
We derive the Akhiezer-Kac expansions \eqref{c33} and \eqref{c34} through a nonlinear steepest descent analysis of RHP \ref{WHRHP}, for which \eqref{c32} will prove a convenient assumption. As was the case for \eqref{c18}, \eqref{c32} is not a necessary assumption for the leading terms in \eqref{c33} and \eqref{c34} to hold, cf. \cite{FTZ} for \eqref{c34}.
\begin{rem} One has to be careful when comparing \eqref{c34} to \cite[Theorem $1$]{FTZ}. First, comparing \cite[$(8)$]{FTZ} with \eqref{a3},\eqref{c5},\eqref{c25} and \eqref{c22}, we relate $\rho$ in \cite[$(9)$]{FTZ} to our $\phi$ via
\begin{equation*}
	\rho(z)=\frac{1}{2\pi}\int_{-\infty}^{\infty}\phi(u)\e^{\im uz}\d u,\ \ \ z\in\mathbb{R}.
\end{equation*}
Thus, for reasonably well-behaved $\phi$, resp. $\rho$,
\begin{equation*}
	\phi(x)=\int_{-\infty}^{\infty}\rho(y)\e^{-\im yx}\d y,\ \ \ x\in\mathbb{R}.
\end{equation*}
However, $\int_{-\infty}^{\infty}\rho(x)\d x=1$ is crucial to the workings in \cite{FTZ}, so we would necessarily want $\phi(0)=1$, which in turn is inadmissible in Theorem \ref{Kactheo}. Moreover, no choice of $p\in[0,1]$ in the Fredholm Paffian in \cite[Theorem $1$]{FTZ} would remedy this obstable. So, in brief, there is no direct identification between $\phi$ in Theorem \ref{Kactheo} and $(p,\rho)$ in \cite{FTZ}, and that is the reason why the third term in \eqref{c34} is different from the first term in \cite[$(14)$]{FTZ}. Conversely, our first and second term in \eqref{c34} match exactly with the integral terms in \cite[$(13),(14)$]{FTZ} upon the identification
\begin{equation*}
	\phi(z)=4p(1-p)\hat{\rho}(z),\ \ \ \ p=\frac{1}{2}\Big(1-\sqrt{1-\phi(0)}\Big)\in\Big(0,\frac{1}{2}\Big).
\end{equation*}
This particular matching can be anticipated from \cite[$(46)$]{FTZ}, after recalling \eqref{a6} and the fact that \cite{FTZ} has their operators acting on $L^2(0,L)$.
\end{rem}
\begin{rem} Via substitution in \eqref{a6}, \eqref{c33} and \eqref{c34} yield $t\rightarrow+\infty$ expansions for the Pfaffians $P_K(t)=P_K((-t,t))$ and $P_{K,2}(t)=P_{K,2}((-t,t))$, respectively.
\end{rem}
\begin{rem} The leading two terms in \eqref{c32} and \eqref{c34} are exactly as in the classical Akhiezer-Kac asymptotic expansion of $D_Q(t)$, as $t\rightarrow+\infty$, cf. \cite{Ka,A}. This the justification for our wording surrounding asymptotics in the abstract and throughout.

\end{rem}
\begin{rem} The Riemann-Hilbert problem \ref{WHRHP} that underwrites \eqref{c29} is known in integrable systems theory not only from \cite{Ko} but much earlier from \cite[$(1.4)$]{IIKS}, albeit for only one concrete $\phi$. Kernels of the form \eqref{c29} have since been coined finite-temperature sine kernels due to their appearances in non-interacting fermionic systems. Fredholm determinants of such finite-temperature sine kernels enjoy rich integrable features, cf. \cite{CT}, and asymptotic studies for concrete $\phi$ have been undertaken in \cite{KBI,X}. In fact, the two common leading terms in \eqref{c32},\eqref{c33}, which originate from the common $D_Q(t)$ in \eqref{c23},\eqref{c26}, match exactly with \cite[$(2.1),(2.2)$]{X}, once we specialise our $\phi$ to \cite[$(1.5)$]{X}.

\end{rem}
In the remainder of this paper, we first use commutation \eqref{c2} to prove \eqref{c3} and \eqref{c7} in Section \ref{sec5}. This is achieved by deriving operator identities involving the entries of \eqref{z5} and \eqref{c5}, which in turn yields the desired factors $A,B$ needed for \eqref{c2}. All steps pertaining to Hankel type kernels, in particular the derivation of \eqref{c11},\eqref{c16} and \eqref{c19},\eqref{c20} are contained in Section \ref{sec6}. These steps are a mixture of purely algebraic manipulations and nonlinear steepest descent techniques applied to RHP \ref{HankelRHP}. In a way, the derivation of \eqref{c23},\eqref{c26} and \eqref{c30},\eqref{c31},\eqref{c33},\eqref{c34} is easier when compared to their analogues for Hankel type kernels; in particular, no further regularisation of the form \eqref{c15} is needed. All relevant steps are demonstrated in Section \ref{cool}. We cconclude the paper with a discussion on a third type of main kernels to which \eqref{c2} is applicable, see Section \ref{notknow}. Appendix \ref{appA} lists the Sherman-Morrison identity utilised in the derivation of \eqref{c11}, as well as a very short primer on Fredholm Pfaffians.

\section{Proof of \eqref{c3} and \eqref{c7} via commutation}\label{sec5}
To utilise \eqref{c2}, we must first ensure that our operators \eqref{z5} and \eqref{c5} have good spectral properties.
\subsection*{The symplectic derived class}
Via exploitation of the compactness of $M$ in \eqref{z5} on $L^2(\Delta)\oplus L^2(\Delta)$, we identify
%We exploit compactness of $M$ in \eqref{z5} on $L^2(\Delta)\oplus L^2(\Delta)$ and identify 
every non-zero element in the spectrum of $M$ with a non-zero element in the spectrum of a simpler, trace class integral operator $N$ on $L^2(\Delta)\oplus L^2(\Delta)$. That this is possible follows from two sets of operator identities, see \eqref{z7} and \eqref{z12} below.
\begin{lem} We have on $L^2(\Delta)$, with $G,H,S:L^2(\Delta)\rightarrow L^2(\Delta)$ defined by their trace class kernels in \eqref{z3},\eqref{z5},
\begin{equation}\label{z7}
	S^{\ast}H=HS-\sum_{k=1}^{2m}(-1)^kH(\delta_{a_k}\otimes\delta_{a_k})H,\ \ \ \ \ \ \ HG=(S^{\ast})^2+\sum_{k=1}^{2m}(-1)^kH(\delta_{a_k}\otimes\delta_{a_k})S^{\ast}.
\end{equation}
Here $\alpha\otimes\beta$ is the integral operator with separable kernel $\alpha(x)\beta(y)$ when $\alpha,\beta\in L^2(\Delta)$ and 
\begin{equation*}
	\int_{\Delta}f(x)\delta_a(x)\d x=f(a),\ \ \ \ a\in\partial\Delta.
\end{equation*}
\end{lem}
\begin{proof} We perform a direct calculation of the underlying kernels, using integration by parts and the general identity $A(\alpha\otimes\beta)B=(A\alpha)\otimes(B^{\ast}\beta)$, valid for any admissible integral operators $A,B$. First,
\begin{align*}
	(S^{\ast}H)(x,y)=&\,\int_{\Delta}S(z,x)H(z,y)\d z\stackrel{\eqref{z6}}{=}H(z,x)H(z,y)\Big|_{\partial\Delta}-\int_{\Delta}H(z,x)S(z,y)\d z\\
	=&\,(HS)(x,y)+\sum_{k=1}^{2m}(-1)^kH(a_k,x)H(a_k,y)=(HS)(x,y)-\sum_{k=1}^{2m}(-1)^k\Big(H(\delta_{a_k}\otimes\delta_{a_k})H\Big)(x,y),
\end{align*}
and second,
\begin{align*}
	(HG)(x,y)=&\,\int_{\Delta}H(x,z)G(z,y)\d z=\int_{\Delta}H(z,x)G(y,z)\d z\stackrel{\eqref{z6}}{=}-H(z,x)S(y,z)\Big|_{\partial\Delta}+\int_{\Delta}S(z,x)S(y,z)\d z\\
	=&\,(S^{\ast})^2(x,y)-\sum_{k=1}^{2m}(-1)^kH(a_k,x)S(y,a_k)=(S^{\ast})^2(x,y)+\sum_{k=1}^{2m}(-1)^k\Big(H(\delta_{a_k}\otimes\delta_{a_k})S^{\ast}\Big)(x,y).
\end{align*}
The proof of \eqref{z7} is complete.
\end{proof}
What results from \eqref{z7} for the trace class integral operator
\begin{equation}\label{z8}
	N:=\sigma_3\begin{bmatrix}S^{\ast} & -S^{\ast}-\sum_{k=1}^{2m}(-1)^kH(\delta_{a_k}\otimes\delta_{a_k})\smallskip\\
	S^{\ast} & -S^{\ast}-\sum_{k=1}^{2m}(-1)^kH(\delta_{a_k}\otimes\delta_{a_k})\end{bmatrix}:L^2(\Delta)\oplus L^2(\Delta)\rightarrow L^2(\Delta)\oplus L^2(\Delta),
\end{equation}
with $\sigma_3:=\bigl[\begin{smallmatrix}1&0\\ 0 & -1\end{smallmatrix}\bigr]$, is the following useful fact regarding the spectra $\sigma(M)$ and $\sigma(N)$ of $M$ and $N$. 
\begin{prop}\label{pz1} We have that $\sigma(M)\setminus\{0\}\subset\sigma(N)$.
\end{prop}
\begin{proof}
Recall that $M$ and $N$ are compact operators on $L^2(\Delta)\oplus L^2(\Delta)$, so $\sigma(M)\setminus\{0\}$, resp. $\sigma(N)\setminus\{0\}$, are at most discrete subsets of $\mathbb{C}\setminus\{0\}$ consisting of eigenvalues $\{\lambda_j(M)\}_{j=1}^{\infty}$ of $M$, resp. eigenvalues $\{\lambda_j(N)\}_{j=1}^{\infty}$ of $N$, by the Riesz-Schauder theorem. Now notice that \eqref{z7} yields the operator composition identity
\begin{equation}\label{z9}
	\begin{bmatrix}H & S^{\ast}\\ -H&-S^{\ast}\end{bmatrix}M=N\begin{bmatrix}H & S^{\ast}\\ -H& -S^{\ast}\end{bmatrix}\ \ \ \ \textnormal{on}\ \ \ L^2(\Delta)\oplus L^2(\Delta).
\end{equation}
Hence, if $\lambda\in\sigma(M)\setminus\{0\}$ then $M[f_1,f_2]^{\top}=\lambda[f_1,f_2]^{\top}$ with $f_k\in L^2(J)$ not both zero a.e., and so by \eqref{z9},
\begin{equation}\label{z10}
	N\begin{bmatrix}g_1\\ g_2\end{bmatrix}=\lambda\begin{bmatrix}g_1\\ g_2\end{bmatrix},\ \ \ \ \ \ \ \ \begin{bmatrix}g_1\\ g_2\end{bmatrix}:=\begin{bmatrix}H&S^{\ast}\\ -H&-S^{\ast}\end{bmatrix}\begin{bmatrix}f_1\\ f_2\end{bmatrix}\in L^2(\Delta)\oplus L^2(\Delta).
\end{equation}
Here $g_k\neq 0\in L^2(\Delta)$, for if we had $g_1=0\in L^2(\Delta)$, then $Hf_1+S^{\ast}f_2=0\in L^2(\Delta)$ by \eqref{z10} while at the same time from the eigenvector-value equation of $M$, 
\begin{equation}\label{z11}
	Sf_1+Gf_2=\lambda f_1\ \ \ \ \textnormal{and}\ \ \ \ Hf_1+S^{\ast}f_2=\lambda f_2.
\end{equation}
So, since $\lambda\neq 0$, necessarily $f_2=0\in L^2(\Delta)$ and hence, again by \eqref{z11}, $Hf_1=0\in L^2(\Delta)$. In turn, by \eqref{z6} and the dominance assumption \eqref{z4}, after differentiation, $Sf_1=0\in L^2(\Delta)$. Consequently, by \eqref{z11} $\lambda f_1=0$ which yields, with $\lambda\neq 0$, that $f_1=0\in L^2(\Delta)$. This would contradict that $f_k$ cannot both be zero a.e., and thus $g_1\neq 0\in L^2(\Delta)$. Likewise, if we had $g_2=0\in L^2(\Delta)$, then by \eqref{z10} $g_2=-g_1$ and so we can use the same argument as before. In summary, $\lambda\in\sigma(N)$ if $\lambda\in\sigma(M)\setminus\{0\}$, and so indeed $\sigma(M)\setminus\{0\}\subset\sigma(N)$, as claimed.
\end{proof}
Our second set of operator identities involves the Sobolev space $W^{1,2}(\Delta):=\big\{f:\,f,Df\in L^2(\Delta)\big\}$ with $D$ as weak derivative. This space is natural in view of the mapping properties of $S^{\ast}$ that transpire from the dominance assumption \eqref{z4}.
 \begin{lem} We have on $W^{1,2}(\Delta)$, with $G,H,S:L^2(\Delta)\rightarrow L^2(\Delta)$ defined by their trace class kernels in \eqref{z3},\eqref{z5},
 \begin{equation}\label{z12}
 	DS^{\ast}=G,\ \ \ \ \ SD=DS^{\ast}+\sum_{k=1}^{2m}(-1)^kDH(\delta_{a_k}\otimes\delta_{a_k}),\ \ \ \ \ HD=S^{\ast}+\sum_{k=1}^{2m}(-1)^kH(\delta_{a_k}\otimes\delta_{a_k}).
\end{equation}
 \end{lem}
 \begin{proof} By direct computation, using the dominated convergence theorem and \eqref{z4}, we have
 \begin{equation*}
 	(DS^{\ast}f)(x)=\frac{\d}{\d x}\int_{\Delta}S(y,x)f(y)\d y\stackrel{\eqref{z4}}{=}\int_{\Delta}\frac{\partial}{\partial x}S(y,x)f(y)\d y\stackrel{\eqref{z6}}{=}-\int_{\Delta}G(y,x)f(y)\d y=(Gf)(x),\ \ f\in L^2(\Delta),
\end{equation*}
followed by
\begin{equation*}
	\big(DH(\delta_a\otimes\delta_a)f\big)(x)=\frac{\partial}{\partial x}H(x,a)f(a)\stackrel{\eqref{z6}}{=}S(x,a)f(a),\ \ \ \ \ \ \ f\in L^2(\Delta),\ \ a\in\partial\Delta.
\end{equation*}
Therefore, for any $f\in W^{1,2}(\Delta)$,
\begin{align*}
	(SDf)(x)=&\,\int_{\Delta}S(x,y)(Df)(y)\d y\\
	=&\,S(x,y)f(y)\Big|_{\partial\Delta}-\int_{\Delta}\bigg[\frac{\partial}{\partial y}S(x,y)\bigg]f(y)\d y
	\stackrel{\eqref{z6}}{=}(Gf)(x)+\sum_{k=1}^{2m}(-1)^kS(x,a_k)f(a_k)\\
	=&\,(Gf)(x)+\sum_{k=1}^{2m}(-1)^k\big(DH(\delta_{a_k}\otimes\delta_{a_k})f\big)(x)=(DS^{\ast}f)(x)+\sum_{k=1}^{2m}(-1)^k\big(DH(\delta_{a_k}\otimes\delta_{a_k})f\big)(x).
\end{align*}
Lastly, with $f\in W^{1,2}(\Delta)$,
\begin{align*}
	(HDf)(x)=&\,\int_{\Delta}H(x,y)(Df)(y)\d y\\
	=&\,H(x,y)f(y)\Big|_{\partial\Delta}+\int_{\Delta}\bigg[\frac{\partial}{\partial y}H(y,x)\bigg]f(y)\d y\stackrel{\eqref{z6}}{=}(S^{\ast}f)(x)+\sum_{k=1}^{2m}(-1)^kH(x,a_k)f(a_k)\\
	=&\,(S^{\ast}f)(x)+\sum_{k=1}^{2m}(-1)^k\big(H(\delta_{a_k}\otimes\delta_{a_k})f\big)(x).
\end{align*}
This concludes our proof of \eqref{z12}.
 \end{proof}
 What results from \eqref{z12} is the following companion to Proposition \ref{pz1}.
 \begin{prop}\label{pz2} We have that $\sigma(N)\setminus\{0\}\subset\sigma(M)$.
 \end{prop}
\begin{proof} From \eqref{z12}, we distill the operator composition formula
 \begin{equation}\label{z13}
 	\begin{bmatrix}0&-D\\ I & 0\end{bmatrix}N=M\begin{bmatrix}0 & -D\\ I&0\end{bmatrix}\ \ \ \ \ \textnormal{on}\ \ \ W^{1,2}(\Delta)\oplus W^{1,2}(\Delta),
\end{equation}
with $I$ acting as identity on $W^{1,2}(\Delta)$. Next, by the dominated convergence theorem and \eqref{z4},
\begin{equation*}
	S^{\ast}\big(L^2(\Delta)\big)\subset W^{1,2}(\Delta),\ \ \ \ \ \ H(\delta_a\otimes\delta_a)\big(L^2(\Delta)\big)\subset W^{1,2}(\Delta),\ \ \ \ \ a\in\partial\Delta.
\end{equation*}
So, if $N[f_1,f_2]^{\top}=\lambda[f_1,f_2]^{\top}$, with $f_k\in L^2(\Delta)$ not both zero a.e., then necessarily $f_k\in W^{1,2}(\Delta)$, and we can thus define
\begin{equation}\label{z14}
	\begin{bmatrix}g_1\\ g_2\end{bmatrix}:=\begin{bmatrix}-Df_2\\ f_1\end{bmatrix}\in L^2(\Delta)\oplus L^2(\Delta).
\end{equation}
Observe that \eqref{z14} solves
\begin{equation}\label{z15}
	M\begin{bmatrix}g_1\\ g_2\end{bmatrix}=\lambda\begin{bmatrix}g_1\\ g_2\end{bmatrix},
\end{equation}
because of \eqref{z13} and the eigenvector-value equation of $N$. More further deduce that, $g_k\neq 0\in L^2(J)$, for if we had $g_1=0\in L^2(\Delta)$, then $Df_2=0\in L^2(\Delta)$ by \eqref{z14}. At the same time, however, from \eqref{z15},
\begin{equation*}
	Sg_1+Gg_2=\lambda g_1\ \ \ \ \textnormal{and}\ \ \ \ Hg_1+S^{\ast}g_2=\lambda g_2.
\end{equation*}
This yields $Gg_2=0\in L^2(\Delta)$ and $S^{\ast}g_2=\lambda g_2$, equivalently $S^{\ast}f_1=\lambda f_1$ by \eqref{z14}. In turn $0=Gg_2=DS^{\ast}g_2=\lambda Dg_2=\lambda Df_1$, by \eqref{z12} in the second equality and \eqref{z14} in the last. So, with $\lambda\neq 0$, $Df_1=0\in L^2(\Delta)$, and thus both $f_k$ are constant a.e. on each connected component of $\Delta$. But this is impossible in light of $S^{\ast}f_1=\lambda f_1$ with $\lambda\neq 0$, unless $f_1=0\in L^2(\Delta)$. But $N[f_1,f_2]^{\top}=\lambda[f_1,f_2]^{\top}$ yields $f_1=-f_2$, compare \eqref{z8}, and so both $f_k$ would be zero a.e.. Consequently, $g_1\neq 0\in L^2(\Delta)$. Likewise, if we had $g_2=0\in L^2(\Delta)$, then at once $g_2=f_1=0\in L^2(\Delta)$ by \eqref{z14}. Since again $f_1=-f_2$, both $f_k$ would vanish a.e. on $\Delta$. In summary, $\lambda\in\sigma(M)$ if $\lambda\in\sigma(N)\setminus\{0\}$, and so indeed $\sigma(N)\setminus\{0\}\subset\sigma(M)$, as claimed.
\end{proof}
The combination of Proposition \ref{pz1} and \ref{pz2} implies the following result; the key step in establishing \eqref{c3}.
\begin{cor} We have
\begin{equation}\label{z16}
	D_M(\Delta)=\prod_{j=1}^{\infty}\big(1-\lambda_j(M)\big)=\prod_{j=1}^{\infty}\big(1-\lambda_j(N)\big),
\end{equation}
where $D_M(\Delta)$ is the Fredholm determinant on $L^2(\Delta)\oplus L^2(\Delta)$ of the operator $M$ in \eqref{z5} and $\{\lambda_j(M)\}_{j=1}^{\infty}$, resp. $\{\lambda_j(N)\}_{j=1}^{\infty}$, is a listing of the nonzero eigenvalues, taking algebraic multiplicities into account, of the operators $M$ and $N$, compare \eqref{z8}.
 \end{cor}
 \begin{proof} By Proposition \ref{pz1} and \ref{pz2}, $\sigma(M)\setminus\{0\}=\sigma(N)\setminus\{0\}$, where $M,N$ are trace class on $L^2(\Delta)\oplus L^2(\Delta)$. Hence, \eqref{z16} follows from \cite[Theorem $3.7$]{S}, provided we can show
\begin{align}
 	\textnormal{algebraic multiplicity of}\ &\lambda\ \textnormal{for}\ M=\textnormal{rank}\big(P_{\lambda,M}\big)\nonumber\\
	=&\,\textnormal{rank}\big(P_{\lambda,N}\big)=\textnormal{algebraic multiplicity of}\ \lambda\ \textnormal{for}\ N\ \ \ \ \ \forall\,\lambda\in\sigma(M)\setminus\{0\},\label{z17}
\end{align}
where $P_{\lambda,T}$ is the spectral projection of a bounded operator $T$ on a Banach space at $\lambda$. The first and third equality in \eqref{z17} are standard, compare \cite[$(3.3.21)$]{S0}, so we need only establish the second one. For that, we recall that $M,N$ are compact operators on $L^2(\Delta)\oplus L^2(\Delta)$, and so any $\lambda\in\sigma(M)\setminus\{0\}$ is an isolated point of the spectrum, making $P_{\lambda,M}$ and $P_{\lambda,N}$, themselves, compact. In particular, cf. \cite[Proposition $3.3.8$]{S0},
 \begin{equation*}
	m:=\textnormal{rank}\big(P_{\lambda,M}\big)<\infty\ \ \ \ \textnormal{and}\ \ \ \  n:=\textnormal{rank}\big(P_{\lambda,N}\big)<\infty\ \ \ \ \ \forall\,\lambda\in\sigma(M)\setminus\{0\}.
\end{equation*}
 We now construct a linear injection $A:\textnormal{Ran}(P_{\lambda,M})\rightarrow\textnormal{Ran}(P_{\lambda,N})$. To that end, recall from \cite[$(2.3.41)$]{S0},
 \begin{equation*}
 	\textnormal{Ran}(P_{\lambda,M})=\big\{f\in L^2(\Delta)\oplus L^2(\Delta):\ (M-\lambda I)^mf=0\big\},
\end{equation*}
with $I$ acting as identity on $L^2(\Delta)\oplus L^2(\Delta)$. Thus, if $f\in\textnormal{Ran}(P_{\lambda,M})\setminus\{0\}$, by the proof workings of Proposition \ref{pz1}, utilizing foremost \eqref{z9},
\begin{equation*}
	(N-\lambda I)^mg=0,\ \ \ \ g:=\begin{bmatrix}H & S^{\ast}\\ -H&-S^{\ast}\end{bmatrix}f\in L^2(\Delta)\oplus L^2(\Delta)\setminus\{0\},
\end{equation*}
 and so $A:f\mapsto g=Af$ maps $\textnormal{Ran}(P_{\lambda,M})$ into $\textnormal{Ran}(P_{\lambda,N})$ and $A$ is linear and injective. Consequently,
 \begin{equation}\label{z18}
 	m=\textnormal{rank}(A)\leq\textnormal{rank}\big(P_{\lambda,N}\big)=n,
\end{equation}
by the rank-nullity theorem. Conversely, we define $B:\textnormal{Ran}(P_{\lambda,N})\rightarrow\textnormal{Ran}(P_{\lambda,M})$ as follows: pick $f\in\textnormal{Ran}(P_{\lambda,N})\setminus\{0\}$ and put as in \eqref{z14}, with $I$ acting as identity on $W^{1,2}(\Delta)$,
\begin{equation*}
	h:=\begin{bmatrix}0&-D\\ I&0\end{bmatrix}f\in L^2(\Delta)\oplus L^2(\Delta)\setminus\{0\}.
\end{equation*}
 Then $B:f\mapsto h=Bf$ maps $\textnormal{Ran}(P_{\lambda,N})$ into $\textnormal{Ran}(P_{\lambda,M})$ in a linear and injective fashion. So, again by the rank-nullity theorem,
 \begin{equation}\label{z19}
 	n=\textnormal{rank}(B)\leq\textnormal{rank}\big(P_{\lambda,M}\big)=m.
\end{equation}
Combining \eqref{z18} and \eqref{z19}, $m=n$ which yields the second equality in \eqref{z17}. Our proof of \eqref{z16} is complete.
 \end{proof}
 Since $N$ is trace class on $L^2(\Delta)\oplus L^2(\Delta)$, the right hand side in \eqref{z16} equals the Fredholm determinant $D_N(\Delta)$ of $N$ on the same space. Seeing \eqref{z8}, this can then be simplified by adding linear combinations of rows and columns in the matrix valued operator. What results is the sought after \eqref{c3}:
 
 \begin{proof}[Proof of \eqref{c3}]
By \eqref{z16} and \cite[Theorem $3.5$]{S0}, using $I$ to denote the identity operator on $L^2(\Delta)$,
\begin{align*}
	D_M(\Delta)=&\,\det\begin{bmatrix}I-S^{\ast}& S^{\ast}+\sum_{k=1}^{2m}(-1)^kH(\delta_{a_k}\otimes\delta_{a_k})\smallskip\\ S^{\ast}&I-S^{\ast}-\sum_{k=1}^{2m}(-1)^kH(\delta_{a_k}\otimes\delta_{a_k})\end{bmatrix}=\det\begin{bmatrix}I-S^{\ast} & S^{\ast}+\sum_{k=1}^{2m}(-1)^kH(\delta_{a_k}\otimes\delta_{a_k})\smallskip\\ I&I\end{bmatrix}\\
	=&\,\det\begin{bmatrix}I-2S^{\ast}-\sum_{k=1}^{2m}(-1)^kH(\delta_{a_k}\otimes\delta_{a_k}) & S^{\ast}+\sum_{k=1}^{2m}(-1)^kH(\delta_{a_k}\otimes\delta_{a_k})\smallskip\\ 0 & I\end{bmatrix}\\
	=&\,\det\bigg(I-2S^{\ast}-\sum_{k=1}^{2m}(-1)^kH(\delta_{a_k}\otimes\delta_{a_k})\bigg)=D_{2S}(\Delta)\det\bigg(I-\sum_{k=1}^{2m}(-1)^k(I-2S^{\ast})^{-1}H(\delta_{a_k}\otimes\delta_{a_k})\bigg),
\end{align*}
where the last Fredholm determinant is the Fredholm determinant of a rank-$(2m)$ integral operator on $L^2(\Delta)$. As such, \eqref{c3} follows from \cite[Chapter I, Theorem $3.2$]{GGK}, seeing that the Fredholm determinant of $2S^{\ast}$ and $2S$, both acting on $L^2(\Delta)$, are the same.
\end{proof}

\begin{rem} The reader will have noticed that we established \eqref{z16} without identifying the $A,B$ factors in \eqref{c2}. Rather, we worked with the operator identities \eqref{z9},\eqref{z13} that provide explicit identities between $M$ and $N$. Such explicit relations might not always be within reach, see our workings around \eqref{z55}. A shorter route to \eqref{z16} that fully relies on $A,B$ factors goes through the factorisation
\begin{equation*}
	M=\begin{bmatrix}D&0\\ 0 & I\end{bmatrix}\begin{bmatrix}H&S^{\ast}\\ H&S^{\ast}\end{bmatrix},\ \ \ I:W^{1,2}(\Delta)\rightarrow L^2(\Delta),\ \ If:=f,
\end{equation*}
where the right factor is a bounded linear transformation from $L^2(\Delta)\oplus L^2(\Delta)$ to $W^{1,2}(\Delta)\oplus W^{1,2}(\Delta)$ and the left from $W^{1,2}(\Delta)\oplus W^{1,2}(\Delta)$ to $L^2(\Delta)\oplus L^2(\Delta)$. Since
\begin{equation}\label{i1}
	\begin{bmatrix}H&S^{\ast}\\ H&S^{\ast}\end{bmatrix}\begin{bmatrix}D&0\\ 0 & I\end{bmatrix}=\begin{bmatrix}S^{\ast}+\sum_{k=1}^{2m}(-1)^kH(\delta_{a_k}\otimes\delta_{a_k}) & S^{\ast}\smallskip\\ S^{\ast}+\sum_{k=1}^{2m}(-1)^kH(\delta_{a_k}\otimes\delta_{a_k}) & S^{\ast}\end{bmatrix}
\end{equation}
is a trace class operator on $W^{1,2}(\Delta)\oplus W^{1,2}(\Delta)$, \eqref{c3} follows now from commutation \eqref{c2}, and determinant manipulations. Note that \eqref{i1} is equal to $\sigma_1\sigma_3N\sigma_3\sigma_1$, with $N$ in \eqref{z8} and $\sigma_1=\bigl[\begin{smallmatrix}0&1\\ 1&0\end{smallmatrix}\bigr]$.
\end{rem}

\subsection*{The orthogonal derived class} When working with \eqref{c5} instead of \eqref{z5}, we first note that $\epsilon(L^2(\Delta))\subset W^{1,2}(\Delta)$ since $\Delta$ is a union of \textit{bounded} intervals. Also, since $G=DS^{\ast}$ by \eqref{z12} and $S=DH$ by \eqref{c6}, both on $L^2(\Delta)$, one can decompose $M$ as
\begin{equation}\label{z52}
	M=\begin{bmatrix}D&0\\ 0 & I\end{bmatrix}\begin{bmatrix}H& S^{\ast}\\ H-\epsilon & S^{\ast}\end{bmatrix},\ \ \ I:W^{1,2}(\Delta)\rightarrow L^2(\Delta),\ \ If:=f,
\end{equation}
where the right factor is a bounded linear transformation from $L^2(\Delta)\oplus L^2(\Delta)$ to $W^{1,2}(\Delta)\oplus W^{1,2}(\Delta)$, by \eqref{z4} and the dominated convergence theorem, and the left is a bounded linear transformation from $W^{1,2}(\Delta)\oplus W^{1,2}(\Delta)$ to $L^2(\Delta)\oplus L^2(\Delta)$. Then, remembering \eqref{c2}, one might be tempted to deduce, for the regularised $2$-determinant $D_{M,2}(\Delta)$ of $M$ in \eqref{c5} on $L^2(\Delta)\oplus L^2(\Delta)$,
\begin{equation}\label{z53}
	D_{M,2}(\Delta)=\det_2\bigg(I-\begin{bmatrix}HD & S^{\ast}\\ HD-\epsilon D & S^{\ast}\end{bmatrix}\bigg),
\end{equation}
where $HD,S^{\ast}:W^{1,2}(\Delta)\rightarrow W^{1,2}(\Delta)$ are trace class, compare \eqref{z12}, and $I$ denotes the identity on $W^{1,2}(\Delta)\oplus W^{1,2}(\Delta)$. However, $\epsilon D:W^{1,2}(\Delta)\rightarrow W^{1,2}(\Delta)$ in \eqref{z53} is not even Hilbert-Schmidt, for integration by parts says
\begin{equation}\label{z54}
	\epsilon D=I+\sum_{k=1}^{2m}(-1)^k\epsilon(\delta_{a_k}\otimes\delta_{a_k})\ \ \ \ \textnormal{on}\ \ W^{1,2}(\Delta).
\end{equation}
Consequently the right hand side in \eqref{z53} is ill-defined and so commutation \eqref{c2} is not directly applicable to the regularised $2$-determinant of \eqref{c5}! One way to overcome this hurdle is via an approximation of the discontinuous kernel $\epsilon(x-y)$ by a smooth skew-symmetric kernel $\eta_n(x,y)$ that converges pointwise to $\epsilon(x-y)$ as $n\rightarrow\infty$ and which is uniformly bounded, as done in \cite[page $753$]{TW}. The resulting $M_n$, which is $M$ in \eqref{c5} but with $\eta_n$ instead of $\epsilon$, is trace class on $L^2(\Delta)\oplus L^2(\Delta)$ and $D_{M,2}(\Delta)$ is equal to the limit of $\det(I-M_n)$ as $n\rightarrow\infty$. With $\eta_n$ in place, one can now simplify the Fredholm determinant of $M_n$ as follows:
\begin{prop} For the Fredholm determinant $D_{M_n}(\Delta)$ of $M_n$ on $L^2(\Delta)\oplus L^2(\Delta)$ one has
\begin{equation}\label{z55}
	D_{M_n}(\Delta)=\det\bigg(I-S^{\ast}-\sum_{k=1}^{2m}\bigg[(-1)^kH\delta_{a_k}-\frac{1}{2}\sum_{j=1}^{2m}\sigma_j(k)H\delta_{a_j}\bigg]\otimes\delta_{a_k}
	+S^{\ast}(\eta_n-\epsilon)D\bigg),
\end{equation}
where the operators in the Fredholm determinant in the right hand side act on $W^{1,2}(\Delta)$, and $\sigma_j(k)\in\{\pm 1\}$ is defined in \eqref{z555}.
\end{prop}
\begin{proof} Note that $M_n$ is trace class on $L^2(\Delta)\oplus L^2(\Delta)$ as, compare \eqref{z52},
\begin{equation*}
	M_n=\begin{bmatrix}D&0\\ 0 & I\end{bmatrix}\begin{bmatrix}H&S^{\ast}\\ H-\eta_n& S^{\ast}\end{bmatrix},\ \ \ \ I:W^{1,2}(\Delta)\rightarrow L^2(\Delta),\ \ If:=f,
\end{equation*}
where the right factor is a trace class linear transformation from $L^2(\Delta)\oplus L^2(\Delta)$ to $W^{1,2}(\Delta)\oplus W^{1,2}(\Delta)$, utilising the smoothness of $\eta_n(x,y)$, and the left is a bounded linear transformation from $W^{1,2}(\Delta)\oplus W^{1,2}(\Delta)$ to $L^2(\Delta)\oplus L^2(\Delta)$. Consequently, $N_n$ on $W^{1,2}(\Delta)\oplus W^{1,2}(\Delta)$, with
\begin{equation*}
	N_n:=\begin{bmatrix}HD & S^{\ast}\\ HD-\eta_n D& S^{\ast}\end{bmatrix}=\begin{bmatrix}H&S^{\ast}\\ H-\eta_n & S^{\ast}\end{bmatrix}\begin{bmatrix}D&0\\ 0 & I\end{bmatrix},
\end{equation*}
is trace class on $W^{1,2}(\Delta)\oplus W^{1,2}(\Delta)$ and, by \eqref{c2},
\begin{equation*}
	D_{M_n}(\Delta)=\det(I-N_n),
\end{equation*}
which is the Fredholm determinant of $N_n$ on $W^{1,2}(\Delta)\oplus W^{1,2}(\Delta)$ in the right hand side. The same can be simplified via block determinant manipulations. Namely, with $I$ as identity on $W^{1,2}(\Delta)$,
\begin{align*}
	D_{M_n}(\Delta)=\det\begin{bmatrix}I-HD & -S^{\ast}\\ -HD+\eta_nD & I-S^{\ast}\end{bmatrix}=\det\begin{bmatrix}I-HD & -S^{\ast} \\ \eta_nD-I& I\end{bmatrix}=\det(I-S^{\ast}-HD+S^{\ast}\eta_nD).
\end{align*}
Note that, with $f\in W^{1,2}(\Delta)$ and $x\in\Delta$,
\begin{eqnarray*}
	\big(S^{\ast}\epsilon Df\big)(x)&\stackrel{\eqref{z54}}{=}&\int_{\Delta}S(y,x)\bigg[f(y)+\sum_{k=1}^{2m}(-1)^k\epsilon(y-a_k)f(a_k)\bigg]\d y\\
	&=&(S^{\ast}f)(x)+\sum_{j=1}^{m}\int_{a_{2j-1}}^{a_{2j}}\bigg[\frac{\partial}{\partial y}H(y,x)\bigg]\sum_{k=1}^{2m}(-1)^k\epsilon(y-a_k)f(a_k)\d y,
\end{eqnarray*}
where, for $x\in(a_{2j-1},a_{2j})$ with $j\in\{1,\ldots,m\}$, we have
\begin{equation*}
	\epsilon(x-a_k)=\frac{1}{2}\begin{cases}+1,&k\in\{1,2,\ldots,2j-1\}\\ -1,&k\in\{2j,2j+1,\ldots,2m\}\end{cases}.
\end{equation*}
Consequently, for any $x\in\Delta$ and $f\in W^{1,2}(\Delta)$,
\begin{align*}
	\big(S^{\ast}&\,\epsilon Df\big)(x)=(S^{\ast}f)(x)\\
	&+\frac{1}{2}\sum_{j=1}^m\Bigg(\sum_{k=1}^{2j-1}(-1)^kf(a_k)\int_{a_{2j-1}}^{a_{2j}}\bigg[\frac{\partial}{\partial y}H(y,x)\bigg]\d y-\sum_{k=2j}^{2m}(-1)^kf(a_k)\int_{a_{2j-1}}^{a_{2j}}\bigg[\frac{\partial}{\partial y}H(y,x)\bigg]\d y\Bigg)\\
	&\hspace{1.35cm}=(S^{\ast}f)(x)+\frac{1}{2}\sum_{j=1}^m\big((H\delta_{a_{2j-1}})(x)-(H\delta_{a_{2j}})(x)\big)\Bigg[\sum_{k=1}^{2j-1}(-1)^kf(a_k)-\sum_{k=2j}^{2m}(-1)^kf(a_k)\Bigg].
\end{align*}
Thus on $W^{1,2}(\Delta)$,
\begin{equation}\label{z555}
	S^{\ast}\epsilon D=S^{\ast}+\frac{1}{2}\sum_{k=1}^{2m}\bigg(\sum_{j=1}^{2m}\sigma_j(k)H\delta_{a_j}\bigg)\otimes\delta_{a_k},
\end{equation}
with the sign coefficients $\sigma_j(1)=\sigma_j(2m)=(-1)^j$ and, written in row vector form for $\ell\in\{1,\ldots,m-1\}$,
\begin{align*}
	(\sigma_j(2\ell))_{j=1}^{2m}=&\,\,(-1,+1,-1,+1,\ldots,-1,+1,\textcolor{red}{+1},-1+1,-1,+1,-1,\ldots,+1,-1),\\
	(\sigma_j(2\ell+1))_{j=1}^{2m}=&\,\,(+1,-1,+1,-1,\ldots,+1,-1,\textcolor{red}{-1},+1,-1,+1,-1,+1,\ldots,-1,+1),
\end{align*}
where the red transposition occurs in entry $2\ell+1$, counted from the left. Consequently, on $W^{1,2}(\Delta)$,
\begin{align}\label{z56}
	-HD+S^{\ast}\eta_nD=-HD+S^{\ast}\epsilon D+S^{\ast}(\eta_n-\epsilon)D\stackrel[\eqref{z555}]{\eqref{z12}}{=}-\sum_{k=1}^{2m}\Bigg(&\,(-1)^kH\delta_{a_k}-\frac{1}{2}\sum_{j=1}^{2m}\sigma_j(k)H\delta_{a_j}\Bigg)\otimes\delta_{a_k}\nonumber\\
	&\,+S^{\ast}(\eta_n-\epsilon)D,
\end{align}
and so \eqref{z55} follows.
\end{proof}
What's crucial for the approximation argument is that for any $f\in W^{1,2}(\Delta)$, in the norm $\|\cdot\|$ on $L^2(\Delta)$, by Cauchy-Schwarz inequality,
\begin{equation*}
	\|(\eta_n-\epsilon)Df\|\leq\sqrt{\|f\|^2+\|Df\|^2}\sqrt{\int_{\Delta^2}\big|\eta_n(y,z)-\epsilon(y-z)\big|^2\d z\,\d y},
\end{equation*}
and so $(\eta_n-\epsilon)D\rightarrow 0$ as $n\rightarrow\infty$ uniformly in operator norm $\|\cdot\|_{\textnormal{op}}$ from $W^{1,2}(\Delta)$ to $L^2(\Delta)$. In turn, in the trace norm $\|\cdot\|_{\textnormal{tr}}$, on the respective spaces,
\begin{equation}\label{z57}
	\|S^{\ast}(\eta_n-\epsilon)D\|_{\textnormal{tr}}\leq\|S^{\ast}\|_{\textnormal{tr}}\|(\eta_n-\epsilon)D\|_{\textnormal{op}}\rightarrow 0,\ \ \textnormal{as}\ n\rightarrow\infty.
\end{equation}
This, combined with \eqref{z55} and the short discussion preceding \eqref{z55}, yields the below proof of \eqref{c7}.
\begin{proof}[Proof of \eqref{c7}] By \eqref{z57} and \cite[Theorem $3.4$]{S},
\begin{eqnarray*}
	D_{M,2}(\Delta)\!\!&=&\!\!\lim_{n\rightarrow\infty}D_{M_n}(\Delta)\stackrel{\eqref{z55}}{=}\det\bigg(I-S^{\ast}-\sum_{k=1}^{2m}\bigg[(-1)^kH\delta_{a_k}-\frac{1}{2}\sum_{j=1}^{2m}\sigma_j(k)H\delta_{a_j}\bigg]\otimes\delta_{a_k}\bigg)\\
	&=&\!\!D_S(\Delta)\det\bigg(I-\sum_{k=1}^{2m}f_k\otimes g_k\bigg),
\end{eqnarray*}
with $\{f_j,g_k\}$ equal to
\begin{equation*}
	f_j:=	(-1)^j(I-2S^{\ast})^{-1}H\delta_{a_j}-\frac{1}{2}\sum_{\ell=1}^{2m}\sigma_{\ell}(j)(I-2S^{\ast})^{-1}H\delta_{a_{\ell}},\ \ \ \ \ \ \ \ g_k:=\delta_{a_k}.
\end{equation*}
Since the Fredholm determinants of $S^{\ast}$ and $S$ on $W^{1,2}(\Delta)\subset L^2(\Delta)$ are the same, \eqref{c7} follows once more by recalling \cite[Chapter I, Theorem $3.2$]{GGK}.
\end{proof}

\section{Hankel composition type kernels - proof of \eqref{c11},\eqref{c16} and \eqref{c19},\eqref{c20}}\label{sec6}
We now begin working with \eqref{c8} and \eqref{c13}.
\subsection*{The symplectic derived class} The Fredholm determinant $D_M(t)=\lim_{a_2\rightarrow\infty}D_M((t,a_2))$ of $M$ in \eqref{z5}, with elements given by \eqref{z6},\eqref{c8}, and \eqref{c10}, is well-defined as shown below.

\begin{prop}\label{pz3} If $S,G,H:L^2(t,\infty)\rightarrow L^2(t,\infty)$ denote the integral operators with kernels \eqref{z6},\eqref{c8}, and \eqref{c10}, with the aforementioned assumptions placed on $\phi$, then all three are trace class and we have that 
\begin{equation*}
	G(x,y)=-G(y,x),\ \ \ \ H(x,y)=-H(y,x)\ \ \ \ \textnormal{on}\ \ \ (t,\infty)\times(t,\infty).
\end{equation*}
\end{prop}
\begin{proof}
By assumption on $\phi$,
\begin{equation*}
	\int_0^{\infty}u|\phi(u+t)|^2\d u<\infty,\ \ \ \ \ \int_0^{\infty}u|\phi(u+t)|\d u<\infty\ \ \ \ \ \forall\,\, t\in\mathbb{R},\ \ \ \ \ \textnormal{and}\ \ \ \phi\in L^1(t,\infty)\cap L^2(t,\infty).
\end{equation*}
Consequently, $S$ is a linear combination of a Hilbert-Schmidt composition integral operator and a rank-$1$ integral operator, both on $L^2(t,\infty)$, and so trace class on $L^2(t,\infty)$. Noting that, by \eqref{c9},
\begin{equation}\label{z23}
	G(x,y)=-\frac{1}{2}\int_0^{\infty}\phi(x+u)(D\phi)(u+y)\d u-\frac{1}{4}\phi(x)\phi(y),
\end{equation}
the same goes for the operator $G$ on $L^2(t,\infty)$ and its kernel is skew-symmetric,
\begin{equation*}
	G(x,y)+G(y,x)\stackrel{\eqref{z23}}{=}-\frac{1}{2}\phi(x+u)\phi(u+y)\Big|_{u=0}^{\infty}-\frac{1}{2}\phi(x)\phi(y)=0\ \ \ \ \textnormal{on}\ \ (t,\infty)\times(t,\infty).
\end{equation*}
Lastly, by Fubini's theorem,
\begin{equation}\label{z24}
	H(x,y)=\frac{1}{2}\int_0^{\infty}\Phi(x+u)(D\Phi)(u+y)\d u+\frac{1}{4}\Phi(x)\Phi(y),\ \ \ \ \ \Phi(x):=\int_x^{\infty}\phi(z)\d z,
	%\bigg[\int_{x+u}^{\infty}\phi(z)\d z\bigg]\phi(u+y)\d u+\frac{1}{4}\int_x^{\infty}\phi(z)\d z\int_y^{\infty}\phi(z)\d z,
\end{equation}
with $\Phi:\mathbb{R}\rightarrow\mathbb{R}$ continuously differentiable, $\Phi$ and $D\Phi$ exponentially fast decaying at $+\infty$ and $\Phi\in L^1(t,\infty)\cap L^2(t,\infty)$. This proves $H$ to be of trace class on $L^2(t,\infty)$ and \eqref{z24} yields skew-symmetry for its kernel,
\begin{equation*}
	H(x,y)+H(y,x)=\frac{1}{2}\Phi(x+u)\Phi(u+y)\Big|_{u=0}^{\infty}+\frac{1}{2}\Phi(x)\Phi(y)=0\ \ \ \ \textnormal{on}\ \ (t,\infty)\times(t,\infty).
\end{equation*}
The proof of the Proposition is now complete.
\end{proof}
For the application of \eqref{c3} to \eqref{c8}, we take $m=1,a_1=t$ and then let $a_2\rightarrow\infty$. This is admissible by Proposition \ref{pz3}, and noting that by the decay properties of $\phi$ at $+\infty$, pointwise in $x\in(t,\infty)$,
\begin{equation*}
	\big((I-2S^{\ast})^{-1}H\big)(x,a_2)\rightarrow 0\ \ \ \ \ \textnormal{as}\ \ a_2\rightarrow\infty,
\end{equation*}
the right hand side in \eqref{c3} becomes, with $D_M(t)=D_M((t,\infty))=\lim_{a_2\rightarrow\infty}D_M((t,a_2))$,
\begin{equation}\label{z25}
	D_M(t)=D_{2S}(t)\big(1+\tau(t)\big),\hspace{1.5cm} \tau(t):=\big((I-2S^{\ast})^{-1}H\big)(t,t),
\end{equation}
where $D_{2S}(t)$ is the Fredholm determinant of $2S$ on $L^2(t,\infty)$ with the kernel of $S$ written in \eqref{c8}. By our assumptions on $\phi$, $\lim_{t\rightarrow+\infty}\tau(t)=0$. More is true, however, which transpires from the Sherman-Morrison identity, cf. \cite[page $37$]{Kra}.
\begin{prop} Assume $I-2S^{\ast}$ is invertible on $L^2(t,\infty)$. Then $\tau(t)$ satisfies
\begin{equation}\label{z26}
	\tau(t)=-\tau(t)+\frac{2\tau^2(t)}{1+2\tau(t)},\ \ t\in\mathbb{R}.
\end{equation}
\end{prop}
\begin{proof} By \eqref{z7}, we have on $L^2(\Delta)$ that
\begin{equation}\label{z27}
	(I-2S^{\ast})H=H\bigg(I-2S+2\sum_{k=1}^{2m}(-1)^k(\delta_{a_k}\otimes\delta_{a_k})H\bigg),
\end{equation}
and so invertibility of $I-2S^{\ast}$ yields invertibility of $I-2S+2\sum_{k=1}^{2m}(-1)^k(\delta_{a_k}\otimes\delta_{a_k})H$, by the Fredholm alternative. Consequently, specialising to $m=1,a_1=t,a_2=\infty$, Lemma \ref{SM} yields
\begin{equation*}
	1+2\tau(t)=1+2\big((I-2S^{\ast})^{-1}H\big)(t,t)=1+2\int_t^{\infty}(H\delta_t)(x)\big((I-2S)^{-1}\delta_t\big)(x)\d x\neq 0,
\end{equation*}
which allows us to apply \eqref{e17}. We obtain thus, utilizing also skew-symmetry $H^{\ast}=-H$,
\begin{align*}
	\tau(t)=\big(&\,(I-2S^{\ast})^{-1}H\big)(t,t)\stackrel{\eqref{z27}}{=}\Big(H\big(I-2S-2(\delta_t\otimes\delta_t)H\big)^{-1}\Big)(t,t)\\
	&\stackrel{\eqref{e17}}{=}\big(H(I-2S)^{-1}\big)(t,t)-\frac{2}{1+2\tau(t)}\Big(H(I-2S)^{-1}\big(\delta_t\otimes (H\delta_t)\big)(I-2S)^{-1}\Big)(t,t)\\
	&=-\big((I-2S^{\ast})^{-1}H\big)(t,t)-\frac{2}{1+2\tau(t)}\big(H(I-2S)^{-1}\delta_t\big)(t)\big((I-2S^{\ast})^{-1}H\delta_t\big)(t)\\
	&=-\tau(t)+\frac{2\tau^2(t)}{1+2\tau(t)},
\end{align*}
as claimed in \eqref{z25}. The proof is complete.
\end{proof}
What results from the quadratic \eqref{z26} with boundary constraint $\tau(+\infty)=0$ is \eqref{c11}.
\begin{proof}[Proof of \eqref{c11}]
This follows from \eqref{z25}, since \eqref{z26} yields two solutions $\tau_1(t)=0$ and $\tau_2(t)=-1$, of which only one links to $\lim_{\tau\rightarrow+\infty}\tau(t)=0$. After that, one factors out $I-Q$ from $D_{2S}(t)$ in \eqref{z25}, and obtains \eqref{c11}.
\end{proof}

\subsection*{The orthogonal derived class} We now turn to the regularised $2$-determinant $D_{M,2}(t)$ of \eqref{c15}, first showing that 
\begin{equation*}
	\rho^{-1}S\rho, \rho^{-1}G\rho^{-1},\rho H\rho:L^2(t,\infty)\rightarrow L^2(t,\infty)
\end{equation*}
are trace class.
\begin{prop}\label{conjkernel} If $\rho^{-1}S\rho,\rho^{-1}G\rho^{-1},\rho H\rho:L^2(t,\infty)\rightarrow L^2(t,\infty)$ denote the integral operators with kernels
\begin{align*}
	\frac{1}{\rho(x)}&\,S(x,y)\rho(y),\ \ \frac{1}{\rho(x)}G(x,y)\frac{1}{\rho(y)},\ \ \rho(x)H(x,y)\rho(y); \ \ \ \textnormal{where}\ \ \  G(x,y)=-\frac{\partial}{\partial y}S(x,y)\ \ \ \textnormal{and}\\
	 H(x,y)=&\,-\int_x^{\infty}\left[\int_0^{\infty}\phi(z+u)\phi(u+y)\d u\right]\d z+\frac{1}{2}\int_y^x\phi(z)\d z+\frac{1}{2}\int_x^{\infty}\phi(z)\d z\int_y^{\infty}\phi(z)\d z,
\end{align*}
with $S(x,y)$ as in \eqref{c13}, then all three are trace class and we have that
\begin{equation*}
	G(x,y)=-G(y,x),\ \ \ \ \ H(x,y)=-H(y,x)\ \ \ \ \textnormal{on}\ \ \ (t,\infty)\times(t,\infty).
\end{equation*}

\end{prop}
\begin{proof} We note $\frac{\partial}{\partial x}H(x,y)=S(x,y)$ on $(t,\infty)\times(t,\infty)$ and by \eqref{c6},\eqref{c9},
\begin{equation*}
	G(x,y)=-\int_0^{\infty}\phi(x+u)(D\phi)(u+y)\d u-\frac{1}{2}\phi(x)\phi(y).
\end{equation*}
Thus $\rho^{-1}G\rho^{-1}$ is a linear combination of a Hilbert-Schmidt composition integral operator and a rank-$1$ integral operator, both on $L^2(t,\infty)$, and so trace class on the same space. Moreover,
\begin{equation*}
	G(x,y)+G(y,x)=-\phi(x+u)\phi(u+y)\Big|_{u=0}^{\infty}-\phi(x)\phi(y)=0\ \ \ \ \textnormal{on}\ \ \ (t,\infty)\times(t,\infty).
\end{equation*}
Next, with $\Phi(x):=\int_x^{\infty}\phi(z)\d z\in L^2(t,\infty)$, exploiting here the exponential decay of $\phi$ at $+\infty$,
\begin{equation*}
	H(x,y)=\int_0^{\infty}\Phi(x+u)(D\Phi)(u+y)\d u-\frac{1}{2}\Phi(x)+\frac{1}{2}\Phi(y)+\frac{1}{2}\Phi(x)\Phi(y),
\end{equation*}
i.e. $\rho H\rho$ is also a linear combination of a Hilbert-Schmidt composition integral operator and three rank-$1$ integral operators, all on $L^2(t,\infty)$, i.e., $\rho H\rho$ is trace class on the same space. Also,
\begin{equation*}
	H(x,y)+H(y,x)=\Phi(x+u)\Phi(u+y)\Big|_{u=0}^{\infty}+\Phi(x)\Phi(y)=0\ \ \ \ \textnormal{on}\ \ \ (t,\infty)\times(t,\infty).
\end{equation*}
It remains to establish the trace class property of $\rho^{-1}S\rho$ with kernel
\begin{equation*}
	(t,\infty)\times(t,\infty)\ni (x,y)\mapsto \int_0^{\infty}\frac{\phi(x+u)}{\rho(x)}\phi(u+y)\rho(y)\d u+\frac{1}{2}\frac{\phi(x)}{\rho(x)}\bigg(1-\int_y^{\infty}\phi(z)\d z\bigg)\rho(y).
\end{equation*}
This is straightforward given the decay properties of $\phi$ at $+\infty$, the polynomial boundedness of $\frac{1}{\rho}$, and that $\rho\in L^2(\mathbb{R})$. The proof is complete.
\end{proof}
Seeing that $D_{M,2}(t)$ with $M$ as in \eqref{c15} is well-defined by Proposition \ref{conjkernel}, one sets out to simplify the same utilising \eqref{c7}, working first with $m=1,a_2=t$ and $a_2>t$ finite, afterwards taking
\begin{equation*}
	D_{M,2}(t)=\lim_{a_2\rightarrow\infty}D_{M,2}((t,a_2)),
\end{equation*}
where $D_{M,2}((t,a_2))$ is the regularised $2$-determinant of $M$ in \eqref{c15} on $L^2(t,a_2)\oplus L^2(t,a_2)$.
\begin{prop} We have,
\begin{equation}\label{z61}
	D_{M,2}((t,a_2))=\det\bigg(I-S^{\ast}+\frac{1}{2}\big(H\delta_t+H\delta_{a_2}\big)\otimes\delta_t-\frac{1}{2}\big(H\delta_t+H\delta_{a_2}\big)\otimes\delta_{a_2}\bigg),
	%\frac{1}{2}\delta_t\otimes H\delta_{a_2}-\frac{1}{2}\delta_{a_2}\otimes H\delta_t\bigg),
\end{equation}
where the right hand side is the Fredholm determinant of the underlying operator on $L^2(t,a_2)$.
\end{prop}
\begin{proof} As in the workings leading to \eqref{c7},
\begin{align*}
	D_{M,2}&\,((t,a_2))=\lim_{n\rightarrow\infty}\det_2\bigg(I-\begin{bmatrix}\rho^{-1}S\rho & \rho^{-1}G\rho^{-1}\\ \rho H\rho-\rho\eta_n\rho & \rho S^{\ast}\rho^{-1}\end{bmatrix}\bigg)\\
	=&\,\,\lim_{n\rightarrow\infty}\det\bigg(I-\begin{bmatrix}\rho^{-1}&0\\ 0 & \rho\end{bmatrix}\begin{bmatrix}S\rho & G\rho^{-1}\\ H\rho-\eta_n\rho & S^{\ast}\rho^{-1}\end{bmatrix}\bigg)=\lim_{n\rightarrow\infty}\det\bigg(I-\begin{bmatrix}S&G\\ H-\eta_n&S^{\ast}\end{bmatrix}\bigg)\\
	=&\,\,\det\bigg(I-S^{\ast}+\frac{1}{2}\big(H\delta_t+H\delta_{a_2}\big)\otimes\delta_t-\frac{1}{2}\big(H\delta_t+H\delta_{a_2}\big)\otimes\delta_{a_2}\bigg),
\end{align*}
where $I$ denotes, depending on the context, the identity either on $L^2(t,a_2)\oplus L^2(t,a_2)$ or on $L^2(t,a_2)$. The second equality utilises the extension property of the regularised $2$-determinant, and \eqref{c2} is used in the third equality, and block determinant manipulations in the fourth. Also, the fourth equality made explicit use of \eqref{z56},\eqref{z57}. The proof is complete.
\end{proof}

Moving ahead we split off the symmetric part of $S$ in \eqref{c13}, via $S=Q+\frac{1}{2}\phi\otimes\beta$, using
\begin{equation*}
	Q(x,y)=\int_0^{\infty}\phi(x+u)\phi(u+y)\d u,\ \ \ \ \ \ \ \ \ \beta(y):=1-\int_y^{\infty}\phi(z)\d z,
\end{equation*}
where $Q:L^2(t,\infty)\rightarrow L^2(t,\infty)$ is trace class and $\phi\in L^2(t,\infty)$, both by the decay properties of $\phi$ at $+\infty$. Then, if $I-Q$ is invertible on $L^2(t,\infty)\subset L^2(t,a_2)$, we factor out $I-Q$ from the right hand side in \eqref{z61}, resulting in
\begin{equation}\label{z62}
	D_{M,2}((t,a_2))=D_Q((t,a_2))\det\big(\delta_{jk}-\langle f_j,g_k\rangle\big)_{j,k=1}^3.
\end{equation}
The factor $D_Q((t,a_2))$ is the Fredholm determinant of $Q$ on $L^2(t,a_2)$ and the functions $\{f_j,g_j\}_{j=1}^3\subset L^2(t,a_2)$ are
\begin{equation*}
	f_1:=\frac{1}{2}\beta,\ f_2:=-\frac{1}{2}(H\delta_t+H\delta_{a_2}),\ f_3:=-f_2,\ \ \ g_1:=(I-Q)^{-1}\phi,\ g_2:=(I-Q)^{-1}\delta_t,\ g_3:=(I-Q)^{-1}\delta_{a_2}.
	%f_1:=(I-Q)^{-1}\alpha,\ \ f_2:=-\frac{1}{2}(I-Q)^{-1}\delta_t,\ \ f_3:=\frac{1}{2}(I-Q)^{-1}\delta_{a_2},\ \ g_1:=\beta,\ \ g_2:=H\delta_{a_2},\ \ g_3:=H\delta_t.
\end{equation*}
Note that $\langle f,g\rangle=\int_t^{a_2}f(x)g(x)\d x$ in \eqref{z62}, and both $(I-Q)^{-1}$ and $H$ act on $L^2(t,a_2)$. We now evaluate the bilinear forms in the right hand side of \eqref{z62} in the limit $a_2\rightarrow\infty$.
\begin{prop}\label{limi1} Assume $I-Q$ is invertible on $L^2(t,\infty)$. Then, as $a_2\rightarrow\infty$, %noting $(I-Q)^{-1}\phi\in L^1(t,\infty)\cap L^2(t,\infty)$ by the properties of $\phi$,
\begin{equation*}
	\langle f_1,g_1\rangle\rightarrow\frac{1}{2}\int_t^{\infty}\big((I-Q)^{-1}\phi\big)(x)\beta(x)\d x,\ \ \ \ \langle f_1,g_2\rangle\rightarrow\frac{1}{2}\int_t^{\infty}\big((I-Q)^{-1}\delta_t\big)(x)\beta(x)\d x,\ \ \ \ \langle f_1,g_3\rangle\rightarrow\frac{1}{2},
	%-\frac{1}{4}\int_t^{\infty}\big((I-Q)^{-1}\phi\big)(x)\Phi(x)\d x,
\end{equation*}
%\begin{equation*}
%	\langle f_1,g_3\rangle\rightarrow \frac{1}{2}\int_t^{\infty}\big((I-Q)^{-1}\phi\big)(x)H(x,t)\d x,
%\end{equation*}
followed by
\begin{align*}
	\langle f_2,g_1\rangle\rightarrow&\,-\frac{1}{2}\int_t^{\infty}\Big[H(x,t)-\frac{1}{2}\Phi(x)\Big]\big((I-Q)^{-1}\phi\big)(x)\d x,\\
	\langle f_2,g_2\rangle\rightarrow&\,-\frac{1}{2}\int_t^{\infty}\Big[H(x,t)-\frac{1}{2}\Phi(x)\Big]\big((I-Q)^{-1}\delta_t\big)(x)\d x,\ \ \ \ \ \ \ \langle f_2,g_3\rangle\rightarrow-\frac{1}{4}\Phi(t),
	%
	%
	%\big((I-Q)^{-1}\delta_t\big)(x)\beta(x)\d x,\ \ \ \ \langle f_2,g_2\rangle\rightarrow\frac{1}{4}\int_t^{\infty}\big((I-Q)^{-1}\delta_t\big)(x)\Phi(x)\d x,
\end{align*}
%\begin{equation*}
%	\langle f_2,g_3\rangle\rightarrow-\frac{1}{2}\int_t^{\infty}\big((I-Q)^{-1}\delta_t\big)(x)H(x,t)\d x,
%\end{equation*}
and concluding with 
\begin{align*}
	\langle f_3,g_1\rangle\rightarrow&\,\frac{1}{2}\int_t^{\infty}\Big[H(x,t)-\frac{1}{2}\Phi(x)\Big]\big((I-Q)^{-1}\phi\big)(x)\d x,\\
	\langle f_3,g_2\rangle\rightarrow&\,\frac{1}{2}\int_t^{\infty}\Big[H(x,t)-\frac{1}{2}\Phi(x)\Big]\big((I-Q)^{-1}\delta_t\big)(x)\d x,\ \ \ \ \ \ \ \langle f_3,g_3\rangle\rightarrow\frac{1}{4}\Phi(t).
	%
	%
	%\big((I-Q)^{-1}\delta_t\big)(x)\beta(x)\d x,\ \ \ \ \langle f_2,g_2\rangle\rightarrow\frac{1}{4}\int_t^{\infty}\big((I-Q)^{-1}\delta_t\big)(x)\Phi(x)\d x,
\end{align*}
Throughout, $\Phi(x)=\int_x^{\infty}\phi(z)\d z\in L^1(t,\infty)\cap L^2(t,\infty)$.
\end{prop}
In order to simplify the limiting, as $a_2\rightarrow\infty$, finite-size determinant in \eqref{z62}, it will be convenient to introduce the following abbreviations,
\begin{equation*}
	u:=\frac{1}{2}\int_t^{\infty}\big((I-Q)^{-1}\phi\big)(x)\Phi(x)\d x,\ \ \ p:=\frac{1}{2}\int_t^{\infty}\big((I-Q)^{-1}\phi\big)(x)\d x,\ \ \ v:=\frac{1}{2}\big((I-Q)^{-1}\Phi\big)(t),
\end{equation*}
and, with $R(x,y)$ denoting the kernel of the resolvent $R=Q(I-Q)^{-1}=(I-Q)^{-1}-I$ on $L^2(t,\infty)$, if it exists, also the shorthand
\begin{equation*}
	r:=\frac{1}{2}\int_t^{\infty}R(x,t)\d x.
\end{equation*}
Note that $u,p,v,r$ are indeed well-defined for $(I-Q)^{-1}\phi\in L^1(t,\infty)\cap L^2(t,\infty)$, by the properties of $\phi$.
\begin{prop} We have, as $a_2\rightarrow\infty$ and provided $\|Q\|<1$ in operator norm on $L^2(t,\infty)$,
\begin{align}
	\det\big(\delta_{jk}-\langle f_j,g_k\rangle\big)_{j,k=1}^3\rightarrow&\,\,\det\begin{bmatrix}1-p+u & v-r-\frac{1}{2} & -\frac{1}{2}\smallskip\\ v-u-\frac{1}{2}\Phi(t)(1-p+u)&1-v+u-\frac{1}{2}\Phi(t)(v-r-\frac{1}{2}) &\frac{1}{4}\Phi(t) \smallskip\\ u-v+\frac{1}{2}\Phi(t)(1-p+u)& v-u+\frac{1}{2}\Phi(t)(v-r-\frac{1}{2}) & 1-\frac{1}{4}\Phi(t)\end{bmatrix}\nonumber\\
	=&\,\,\det\begin{bmatrix}1-p+u &v-r&-\frac{1}{2}\smallskip\\ 0  & 0 & 1\smallskip\\ u-v&v-u-1&1\end{bmatrix}\label{z63}.
\end{align}
\end{prop}
\begin{proof} To obtain the limit in \eqref{z63}, we use the limits in Proposition \ref{limi1} and rewrite them in terms of $(u,p,v,r)$. This is possible in light of skew-symmetry of $H(x,t)=-H(t,x)$ and the identities
\begin{align*}
	\int_t^{\infty}\big((I-Q)^{-1}\phi\big)(x)\phi(x+y)\d x=&\,\,\big((I-Q)^{-1}Q\big)(t,y+t),\ \ \ y>0,\\
	\int_t^{\infty}\big((I-Q)^{-1}\delta_t\big)(x)\phi(x+y)\d x=&\,\,\big((I-Q)^{-1}\phi\big)(y+t),\ \ \ y>0.
\end{align*}
To obtain the simplified determinant in \eqref{z63}, one adds appropriate linear combinations of rows and columns inside the determinant to one another; this does not affect the value of the determinant.
\end{proof}
While \eqref{z63} seemingly depends on four variables $(u,p,v,r)$, two of them are redundant. Indeed, what transpires from the proof workings of \cite[Lemma $4.1$]{BB} are the simple relations
\begin{equation}\label{z633}
	p=v\ \ \ \ \ \textnormal{and}\ \ \ \ \ \ r=u,
\end{equation}
which allow us to express the right hand side in \eqref{z63} solely in terms of $(u,p)$, say. 
\begin{proof}[Proof of \eqref{c16}] By \eqref{z63} and \eqref{z633},
\begin{equation}\label{z64}
	\lim_{a_2\rightarrow\infty}\det\big(\delta_{jk}-\langle f_j,g_k\rangle\big)_{j,k=1}^3=1+2u-2p.%=1-\int_t^{\infty}\big((I-Q)^{-1}\phi\big)(x)\bigg(1-\int_x^{\infty}\phi(z)\d z\bigg)\d x
\end{equation}
Now recall $D_{M,2}(t)=\lim_{a_2\rightarrow\infty}D_{M,2}((t,a_2))$ and equations \eqref{z61},\eqref{z62},\eqref{z63}, and \eqref{z64}. What results is \eqref{c16}.
\end{proof}

\subsection*{RHP and asymptotic results} We are left to derive \eqref{c19} and \eqref{c20}, and this will be achieved through a nonlinear steepest descent analysis. 

\subsubsection{Asymptotic resolution of RHP \ref{HankelRHP}} For real-valued $\phi:\mathbb{R}\rightarrow\mathbb{R}$, we have 
\begin{equation}\label{c35}
	\bar{r}(z)=\im\int_{-\infty}^{\infty}\phi(y)\e^{\im zy}\d y=:r_2(z),\ \ \ \ \ \ \ r_1(z):=r(z),
\end{equation}
for any $z\in\mathbb{R}$. With \eqref{c18} in place, $z\mapsto r_j(z)$ in \eqref{c35} admits analytic continuation to the closed horizontal sector $|\Im z|\leq\epsilon$, with $|r_j(z)|<1$ and $|zr_j(z)|<1$ in the same sector. In turn, the $t\rightarrow+\infty$ resolution of RHP \ref{HankelRHP} is trivial: we simply set
\begin{equation}\label{R3}
	Y(z;t,\phi):=X(z;t,\phi)\begin{cases}\begin{bmatrix}1 & 0\\ -r_1(z)\e^{\im t z}&1\end{bmatrix},&\Im z\in(0,\epsilon)\smallskip\\ \begin{bmatrix}1 & -r_2(z)\e^{-\im tz}\\ 0 & 1\end{bmatrix},&\Im z\in(-\epsilon,0)\\ \mathbb{I},&\textnormal{else} \end{cases},
\end{equation}
and obtain the below small norm problem.
\begin{problem}\label{master2} Let $t>0$ and $\phi\in W^{1,1}(\mathbb{R})\cap L^{\infty}(\mathbb{R})$ be as in \eqref{c18}. The function $Y(z)=Y(z;t,\phi)\in\mathbb{C}^{2\times 2}$ defined in \eqref{R3} has the following properties.
\begin{enumerate}
	\item[(1)] $z\mapsto Y(z)$ is analytic for $z\in\mathbb{C}\setminus\Sigma_Y$ with $\Sigma_Y:=(\mathbb{R}-\im\epsilon)\cup(\mathbb{R}+\im\epsilon)$, and extends continuously on the closure of $\mathbb{C}\setminus\Sigma_Y$.
	\item[(2)] The limiting values $Y_{\pm}(z):=\lim_{\epsilon\downarrow 0}Y(z\pm\im\epsilon)$ on $\Sigma_Y\ni z$ satisfy $Y_+(z)=Y_-(z)G_Y(z)$ where the jump matrix $G_Y(z)=G_Y(z;t,\phi)$ is of the form
	\begin{equation*}
		G_Y(z)=\begin{bmatrix}1&0\\ r_1(z)\e^{\im tz}&1\end{bmatrix},\ \ \Im z=\epsilon;\hspace{1cm}G_Y(z)=\begin{bmatrix}1&-r_2(z)\e^{-\im tz}\\ 0 & 1\end{bmatrix},\ \ \Im z=-\epsilon.
	\end{equation*}
	\item[(3)] $Y(z)\rightarrow\mathbb{I}$ as $z\rightarrow\infty,z\notin\Sigma_Y$.
\end{enumerate}
\end{problem}
The above RHP is a small norm problem, as demonstrated by the deduced estimates
\begin{equation*}
	\|G_Y(\cdot;t,\phi)-\mathbb{I}\|_{L^2(\Sigma_Y)}\leq\sqrt{\frac{2\pi}{\epsilon}}\,\e^{-t\epsilon},\ \ \ \ \ \|G_Y(\cdot;t,\phi)-\mathbb{I}\|_{L^{\infty}(\Sigma_Y)}\leq \e^{-t\epsilon}\ \ \ \forall\,t>0,\ \epsilon>0,
\end{equation*}
and thus is uniquely solvable as $t\rightarrow+\infty$ by general theory, cf. \cite{DZ}. What also transpires is the estimate
\begin{equation}\label{R4}
	Y(z)=\mathbb{I}+\mathcal{O}\bigg(\frac{1}{\sqrt{\epsilon}}\,\e^{-t\epsilon}\bigg)\ \ \ \ \ \textnormal{as}\ \ t\rightarrow+\infty,
\end{equation}
which holds uniformly for $z\in\mathbb{C}\setminus\Sigma_Y$ and $\epsilon>0$. To obtain a similar result at $t=-\infty$, we change the triangularity of the factorisation of the jump matrix $G_X(z)$ in RHP \ref{HankelRHP} that underwrites the transformation \eqref{R3}. Precisely, we introduce the $g$-function
\begin{equation}\label{R5}
	g(z)=g(z;\phi):=\frac{\im}{2\pi}\int_{-\infty}^{\infty}\ln\big(1-r_1(\lambda)r_2(\lambda)\big)\frac{\d\lambda}{\lambda-z},\ \ \ z\in\mathbb{C}\setminus\mathbb{R},
\end{equation}
noting that the integral in \eqref{R5} is absolutely convergent since $|r_j(z)|<1$ and $|zr_j(z)|<1$ on $\mathbb{R}$. Moreover,
\begin{equation*}
	g_+(z)-g_-(z)=-\ln\big(1-r_1(z)r_2(z)\big),\ \ z\in\mathbb{R};\hspace{1cm}g_{\pm}(z):=\lim_{\epsilon\downarrow 0}g(z\pm\im\epsilon),\ z\in\mathbb{R},
\end{equation*}
as well as
\begin{equation*}
	g(z)=\frac{1}{2\pi\im z}\int_{-\infty}^{\infty}\ln\big(1-r_1(\lambda)r_2(\lambda)\big)\d\lambda+o\big(z^{-1}\big),\ \ \ \ z\rightarrow\infty,\ z\notin\mathbb{R}.
\end{equation*}
What results for the transformed function
\begin{equation}\label{R6}
	S(z;t,\phi):=X(z;t,\phi)\e^{g(z;\phi)\sigma_3},\ \ \ z\in\mathbb{C}\setminus\mathbb{R},
\end{equation}
is summarised in the following problem.
\begin{problem}\label{master3} Let $t\in\mathbb{R}$ and $\phi\in W^{1,1}(\mathbb{R})\cap L^{\infty}(\mathbb{R})$ be as in \eqref{c18}. The function $S(z)=S(z;t,\phi)\in\mathbb{C}^{2\times 2}$ defined in \eqref{R6} is uniquely determined by the following properties.
\begin{enumerate}
	\item[(1)] $z\mapsto S(z)$ is analytic for $z\in\mathbb{C}\setminus\mathbb{R}$ and extends continuously on the closure of $\mathbb{C}\setminus\mathbb{R}$.
	\item[(2)] $z\mapsto S(z)$ admits continuous limiting values $S_{\pm}(z)=\lim_{\epsilon\downarrow 0}S(z\pm\im\epsilon)$ on $\mathbb{R}\ni z$ that satisfy $S_+(z)=S_-(z)G_S(z)$, with $G_S(z)=G_S(z;t,\phi)$ equal to
	\begin{equation*}
		G_S(z)=\begin{bmatrix}1 & -\eta_2(z)\e^{-\im tz-2g_+(z)}\smallskip\\ \eta_1(z)\e^{\im tz+2g_-(z)} & 1-r_1(z)r_2(z)\end{bmatrix};\ \ \  \eta_k(z)=\eta_k(z;\phi):=\frac{r_k(z)}{1-r_1(z)r_2(z)},\ \ |\Im z|\leq\epsilon.
	\end{equation*}
	\item[(3)] As $z\rightarrow\infty$ in $\mathbb{C}\setminus\mathbb{R}$,
	\begin{equation*}
		S(z)=\mathbb{I}+\frac{1}{z}S_1+o\big(z^{-1}\big);\ \ \ S_1=S_1(t,\phi)=X_1(t,\phi)+\frac{\sigma_3}{2\pi\im}\int_{-\infty}^{\infty}\ln\big(1-r_1(\lambda)r_2(\lambda)\big)\d\lambda.
	\end{equation*}
\end{enumerate}
\end{problem}
Since $z\mapsto\eta_k(z)$ are analytic for $|\Im z|\leq\epsilon$, with $z\eta_k(z)=\mathcal{O}(1)$ as $|z|\rightarrow\infty$ in the same sector, we can factorise
\begin{equation*}
	G_S(z)=\begin{bmatrix}1 & 0\\
	\eta_1(z)\e^{\im tz+2g_-(z)} & 1\end{bmatrix}\begin{bmatrix}1 & -\eta_2(z)\e^{-\im tz-2g_+(z)}\\ 0 & 1\end{bmatrix},\ \ z\in\mathbb{R},
\end{equation*}
and transform RHP \ref{master3} as follows. We introduce
\begin{equation}\label{R7}
	T(z;t,\phi):=S(z;t,\phi)\begin{cases}\begin{bmatrix}1 & \eta_2(z)\e^{-\im tz-2g(z)}\\ 0 & 1\end{bmatrix},&\Im z\in(0,\epsilon)\smallskip\\ \begin{bmatrix}1 & 0\\ \eta_1(z)\e^{\im tz+2g(z)} & 1\end{bmatrix},&\Im z\in(-\epsilon,0)\\ \mathbb{I},&\textnormal{else}\end{cases},
\end{equation}
and collect its properties below.
\begin{problem}\label{master4} Let $t<0$ and $\phi\in W^{1,1}(\mathbb{R})\cap L^{\infty}(\mathbb{R})$ be as in \eqref{c18}. The function $T(z)=T(z;t,\phi)\in\mathbb{C}^{2\times 2}$ defined in \eqref{R7} has the following properties.
\begin{enumerate}
	\item[(1)] $z\mapsto T(z)$ is analytic for $z\in\mathbb{C}\setminus\Sigma_Y$, with $\Sigma_Y$ as in RHP \ref{master2}, and extends continuously on the closure of $\mathbb{C}\setminus\Sigma_Y$.
	\item[(2)] $z\mapsto T(z)$ admits continuous limiting values $T_{\pm}(z)=\lim_{\epsilon\downarrow 0}T(z\pm\im\epsilon)$ on $\Sigma_Y\ni z$ that satisfy $T_+(z)=T_-(z)G_T(z)$, with $G_T(z;t,\phi)$ equal to 
	\begin{equation*}
		G_T(z)=\begin{bmatrix}1&-\eta_2(z)\e^{-\im tz-2g(z)}\\ 0 & 1\end{bmatrix},\ \ \Im z=\epsilon;\hspace{1cm}G_T(z)=\begin{bmatrix}1 & 0\\ \eta_1(z)\e^{\im tz+2g(z)} & 1\end{bmatrix},\ \ \Im z=-\epsilon.
	\end{equation*}
	\item[(3)] As $z\rightarrow\infty$ with $z\notin\Sigma_Y$,
	\begin{equation*}
		T(z)=\mathbb{I}+\frac{1}{z}S_1+o\big(z^{-1}\big).
	\end{equation*}
\end{enumerate}
\end{problem}
Noting that $|\e^{\mp\im tz}|=\e^{t\epsilon}$ for $\Im z=\pm\epsilon$ and that $\eta_2(z)\e^{-2g(z)},\eta_1(z)\e^{2g(z)}$ are bounded on $\Im z=\epsilon,\Im z=-\epsilon$, we deduce, for any $\epsilon>0$, existence of $c=c(\epsilon)>0$ such that
\begin{equation*}
	\|G_T(\cdot;t,\phi)-\mathbb{I}\|_{L^2\cap L^{\infty}(\Sigma_Y)}\leq c\,\e^{-\epsilon|t|}\ \ \ \ \forall\,t<0.
\end{equation*}
In short, RHP \ref{master4} is a small norm problem and so uniquely solvable as $t\rightarrow-\infty$, cf. \cite{DZ}. What's more, we have
\begin{equation}\label{R8}
	T(z)=\mathbb{I}+\mathcal{O}\big(\e^{-\epsilon|t|}\big)\ \ \ \ \textnormal{as}\ \ t\rightarrow-\infty,
\end{equation}
uniformly for $z\in\mathbb{C}\setminus\Sigma_Y$, for any fixed $\epsilon>0$. We now utilise \eqref{R4} and \eqref{R8} to derive $t\rightarrow-\infty$ asymptotics for \eqref{c11} and \eqref{c16}. 
\subsubsection{Derivation of asymptotics for $D_M(t)$ and $D_{M,2}(t)$}
First, by \eqref{R6} and \eqref{R8},
\begin{equation}\label{R9}
	\frac{\d}{\d t}\ln D_Q(t)=\im X_1^{11}(t,\phi)=-\frac{1}{2\pi}\int_{-\infty}^{\infty}\ln\big(1-r_1(\lambda)r_2(\lambda)\big)\d\lambda+\im S_1^{11}(t,\phi)\stackrel{t\rightarrow-\infty}{=}s(0)+\mathcal{O}\big(t^{-\infty}\big),
\end{equation}
with the shorthand, using again that $\phi:\mathbb{R}\rightarrow\mathbb{R}$ is real-valued and so $r_2(\lambda)=\overline{r_1(\lambda)},\lambda\in\mathbb{R}$,
\begin{equation*}
	s(x)=-\frac{1}{2\pi}\int_{-\infty}^{\infty}\ln\big(1-|r(\lambda)|^2\big)\e^{\im x\lambda}\d\lambda,\ \ \ x\in\mathbb{R};\ \ \ \ \ r(\lambda)=r_1(\lambda).
\end{equation*}
Consequently, after indefinite $t$-integration in \eqref{R9}, 
\begin{equation}\label{R10}
	\ln D_Q(t)=ts(0)+\varpi+\mathcal{O}\big(t^{-\infty}\big),\ \ \ t\rightarrow-\infty,
\end{equation}
where $\varpi$ is $t$-independent. To compute this constant, one can introduce a parameter $\gamma\in[0,1]$ into RHP \ref{HankelRHP} via the rule $\phi\mapsto\sqrt{\gamma}\phi$, and afterwards exploit general theory underwriting RHP \ref{master2}, cf. \cite{IIKS}. By exploiting in particular $\|Q\|<1$ in operator norm, compare the proof workings of \cite[Theorem $2.16$]{Bo}, we arrive at the explicit formula
\begin{equation}\label{R11}
	\varpi=\int_0^{\infty}xs(x)s(-x)\d x-\frac{1}{2\pi}\int_{-\infty}^{\infty}\Im\bigg\{\frac{r'(\lambda)}{r(\lambda)}\bigg\}\ln\big(1-|r(\lambda)|^2\big)\d\lambda.
\end{equation}
Moving now to the equation immediately below \eqref{c17}, we require
\begin{equation*}
	\omega(t)=\int_t^{\infty}q(s)\d s=\int_t^{\infty}X_1^{21}(s,\phi)\d s\stackrel{\eqref{R6}}{=}\int_t^{\infty}S_1^{21}(s,\phi)\d s\stackrel{\eqref{R8}}{=}\int_{-\infty}^{\infty}S_1^{21}(s,\phi)\d s+\mathcal{O}(t^{-\infty}),\ \ \ t\rightarrow-\infty,
\end{equation*}
and thus the total integral
\begin{equation*}
	P=P(\phi):=\int_{-\infty}^{\infty}X_1^{21}(s,\phi)\d s=\int_{-\infty}^{\infty}X_1^{12}(s,\phi)\d s.
\end{equation*}
The evaluation of this is a standard exercise, cf. \cite{BBDI}, that uses the Zhakharov-Shabat system
\begin{equation}\label{R12}
	\frac{\partial W}{\partial t}=\left\{-\frac{\im z}{2}\sigma_3+\im\begin{bmatrix}0 & X_1^{12}\\ -X_1^{12}& 0\end{bmatrix}\right\}W,
\end{equation}
for the function $W(z;t,\phi):=X(z;t,\phi)\e^{-\frac{\im}{2}tz\sigma_3},z\in\mathbb{C}\setminus\mathbb{R}$. Indeed, with
\begin{equation*}
	V(t;\phi):=\lim_{\substack{z\rightarrow 0\\ \Im z<0}}X(z;t,\phi),
\end{equation*}
system \eqref{R12} yields, in particular,
\begin{equation}\label{R13}
	V(t,\phi)=\begin{bmatrix}\cosh\nu & \im \sinh\nu\\ -\im\sinh\nu & \cosh\nu\end{bmatrix}V(t_0,\phi),\ \ \ \ \nu=\nu(t,t_0,\phi):=\int_{t_0}^tX_1^{12}(s,\phi)\d s,\ \ t,t_0\in\mathbb{R}.
\end{equation}
Evidently, the sought after $P$ is the limit of $\nu(t,t_0,\phi)$ as $t\rightarrow+\infty$ and $t_0\rightarrow-\infty$, so we first calculate
\begin{equation*}
	\lim_{t\rightarrow+\infty}V(t,\phi)=\lim_{t\rightarrow+\infty}\left\{\lim_{\substack{z\rightarrow 0\\ \Im z<0}}X(z;t,\phi)\right\}\stackrel{\eqref{R3}}{=}\lim_{t\rightarrow+\infty}Y(0;t,\phi)\begin{bmatrix}1&r_2(0)\\ 0 & 1\end{bmatrix}\stackrel{\eqref{R4}}{=}\begin{bmatrix}1 & r_2(0)\\ 0 & 1\end{bmatrix}.
\end{equation*}
After that we calculate
\begin{align*}
	\lim_{t_0\rightarrow-\infty}V(t_0,\phi)\stackrel{\eqref{R6}}{=}\lim_{t_0\rightarrow-\infty}&\,\left\{\lim_{\substack{z\rightarrow 0\\ \Im z<0}}S(z;t_0,\phi)\e^{-g(z)\sigma_3}\right\}\stackrel{\eqref{R7}}{=}\lim_{t_0\rightarrow-\infty}T(0;t_0,\phi)\e^{-g_-(0)\sigma_3}\begin{bmatrix}1&0\\ -\eta_1(0)& 1\end{bmatrix}\\
	&\,\stackrel{\eqref{R8}}{=}\e^{-g_-(0)\sigma_3}\begin{bmatrix}1 & 0\\ -\eta_1(0)&1\end{bmatrix}=\begin{bmatrix}1&0\\ -r_1(0)&1\end{bmatrix}\big(1-r_1(0)r_2(0)\big)^{-\frac{1}{2}\sigma_3},
\end{align*}
and combine our results,
\begin{equation*}
	V(+\infty,\phi)V(-\infty,\phi)^{-1}=\frac{1}{\sqrt{1-r_1(0)r_2(0)}}\begin{bmatrix}1 & r_2(0)\\ r_1(0)&1\end{bmatrix}.
\end{equation*}
In turn, using again that $\phi:\mathbb{R}\rightarrow\mathbb{R}$ is real-valued,
\begin{equation*}
	\cosh\nu(+\infty,-\infty,\phi)=\frac{1}{\sqrt{1-|r(0)|^2}},\ \ \ \ \ \ \sinh\nu(+\infty,-\infty,\phi)=\frac{\im r(0)}{\sqrt{1-|r(0)|^2}},
\end{equation*}
and so
\begin{equation*}
	\exp\left[-\int_{-\infty}^{\infty}X_1^{12}(s,\phi)\d s\right]=\cosh\nu(+\infty,-\infty,\phi)-\sinh(+\infty,-\infty,\phi)=\frac{1-\im r(0)}{\sqrt{1-|r(0)|^2}}=\sqrt{\frac{1-\im r(0)}{1+\im r(0)}}.
\end{equation*}
This implies, as $t\rightarrow-\infty$,
\begin{equation}\label{R113}
	\e^{-\omega(t)}=\sqrt{\frac{1-\im r(0)}{1+\im r(0)}}+\mathcal{O}(t^{-\infty}),\ \ \ \ \cosh^2\bigg(\frac{\omega(t)}{2}\bigg)=\Bigg(\frac{1}{2}\sqrt[4]{\frac{1+\im r(0)}{1-\im r(0)}}+\frac{1}{2}\sqrt[4]{\frac{1-\im r(0)}{1+\im r(0)}}\Bigg)^2+\mathcal{O}(t^{-\infty}),
\end{equation}
and we summarise our workings below.
\begin{proof}[Proof of \eqref{c19} and \eqref{c20}] There is nothing left to prove; simply combine \eqref{R10},\eqref{R11} and \eqref{R113} with \eqref{c11} and \eqref{c16}.
\end{proof}

\section{Wiener-Hopf type kernels - proof of \eqref{c23},\eqref{c26} and \eqref{c30},\eqref{c31}, as well as \eqref{c33},\eqref{c34}}\label{cool}
We start by proving \eqref{c23} and \eqref{c26} through an application of \eqref{c3} and \eqref{c7}
\subsection*{The symplectic derived class} One begins by showing that the Fredholm determinant $D_M(t)$ of $M$ in \eqref{z5} with $S,G,H$ as in \eqref{c21},\eqref{z6},\eqref{c22} is well-defined.
\begin{prop}\label{pz4} If $S,G,H:L^2(-t,t)\rightarrow L^2(-t,t)$ denote the integral operators with kernels \eqref{c21},\eqref{z6}, and \eqref{c22}, then all three are trace class and we have that
\begin{equation*}
	G(x,y)=-G(y,x),\ \ \ \ \ H(x,y)=-H(y,x)\ \ \ \textnormal{on}\ \ \ (-t,t)\times(-t,t).
\end{equation*}
\end{prop}
\begin{proof} The operator $S$ is a truncated Wiener-Hopf operator with kernel $S(x,y)$ in \eqref{c21} that is continuous on $(-t,t)\times(-t,t)$, since $\phi\in L^1(\mathbb{R})\cap L^{\infty}(\mathbb{R})\subset L^2(\mathbb{R})$, and that satisfies $S(x,y)=S(y,x)\in\mathbb{R}$ for all $x,y\in(-t,t)$. Moreover,
\begin{equation*}
	\int_{-t}^t\int_{-t}^tS(x,y)f(x)\overline{f(y)}\d x\d y=\frac{1}{4\pi}\int_{-\infty}^{\infty}\bigg|\int_{-t}^tf(x)\e^{\im ux}\d x\bigg|^2\phi(u)\d u\ \ \ \ \forall\,f\in L^2(-t,t).
\end{equation*}
Hence, $S:L^2(-t,t)\rightarrow L^2(-t,t)$ is trace class by \cite[Theorem $2.12$]{S} provided $\phi$ is non-negative. Since every $\phi\in L^1(\mathbb{R})\cap L^{\infty}(\mathbb{R})$ is a linear combination of non-negative functions, we obtain that $S$ is trace class on $L^2(-t,t)$ in the general case. Next, subject to our assumptions,
\begin{equation*}
	G(x,y)=\frac{\im}{4\pi}\int_{-\infty}^{\infty}u\phi(u)\e^{\im u(x-y)}\d u=-G(y,x)\ \ \ \ \forall\,x,y\in(-t,t),
\end{equation*}
and thus $G$ is trace class on $L^2(-t,t)$ by the same arguments made for $S$. Lastly,
\begin{equation*}
	H(x,y)=\frac{1}{4\pi}\int_0^{x-y}\bigg[\int_{-\infty}^{\infty}\phi(u)\e^{\im uz}\d u\bigg]\d z=-H(y,x)\ \ \ \ \forall\,x,y\in(-t,t),
\end{equation*}
and
\begin{equation}\label{z30}
	H(x,y)=\int_0^xS(z,y)\d z+\int_y^0S(z,y)\d z,
\end{equation}
where the second summand in \eqref{z30} is the kernel of a finite rank integral operator on $L^2(-t,t)$. The first summand in \eqref{z30} can be realised as a Hilbert-Schmidt composition kernel
\begin{equation*}
	\int_{-\infty}^{\infty}K_1(x,u)K_2(u,y)\d u,\ \ \ \ \ K_1:L^2(\mathbb{R})\rightarrow L^2(-t,t),\ \ \ \ \ K_2:L^2(-t,t)\rightarrow L^2(\mathbb{R}),
\end{equation*}
provided $\phi$ is non-negative. If it is, then $H$ is trace class on $L^2(-t,t)$ by \eqref{z30}; if it is not, then we again write $\phi\in L^1(\mathbb{R})\cap L^{\infty}(\mathbb{R})$ as linear combination of non-negative functions, and obtain the general assertion via the same methods as with \eqref{z30}.
\end{proof}
For the application of \eqref{c3} to \eqref{c21}, which is admissible by Proposition \ref{pz4}, we take $m=1,a_1=-t$, and $a_2=t$. Noting that
\begin{equation*}
	S(x,y)=S(y,x)=S(-x,-y),\ \ \ \ \ \ H(x,y)=-H(y,x)=-H(-x,-y)\ \ \ \ \textnormal{on}\ \ (-t,t)\times(-t,t),
\end{equation*}
we deduce $F_{11}((-t,t))=F_{22}((-t,t)),F_{12}((-t,t))=F_{21}((-t,t))$. Thus, from \eqref{c3},
\begin{equation}\label{z31}
	D_M(\Delta)=D_{2S}(t)\prod_{a=\pm 1}\Big(1-\big(F_{22}((-t,t))+aF_{21}((-t,t))\big)\Big),
\end{equation}
where $D_{2S}(t)$ is the Fredholm determinant of $2S$ on $L^2(-t,t)$, with the kernel of $S$ written in \eqref{c21}.

\begin{proof}[Proof of \eqref{c23}] Since $Q=2S$ by \eqref{c23}, the first factor $D_{2S}(t)$ in the right hand side of \eqref{z31} is $D_Q(t)$. For the second factor, involving $F_{22}-F_{21}$, we compute, utilising the shorthands $F_{jk}\equiv F_{jk}((-t,t))$ and $\langle f,g\rangle=\int_{-t}^tf(x)g(x)\d x$,
\begin{align*}
	F_{22}-F_{21}=\big\langle (I-Q)^{-1}H\delta_t,\delta_t+\delta_{-t}\big\rangle=\big\langle H\delta_t,(I-Q)^{-1}(\delta_t+\delta_{-t})\big\rangle,
\end{align*}
where, writing $Q(x,y)=:W(x-y)$,
\begin{equation}\label{z33}
	(H\delta_t)(x)=H(x,t)=\frac{1}{2}\int_0^{x-t}W(z)\d z,\ \ \ \ x\in(-t,t).
\end{equation}
Therefore, using the function $1(x):=1$ identically one on $(-t,t)$,
\begin{align}\label{z34}
	(S1)(x)=\frac{1}{2}\int_{-t}^tW(x-y)\d y=\frac{1}{2}\int_{x-t}^{x+t}W(u)\d u\stackrel{\eqref{z33}}{=}-2H(x,t)+\Big(H(x,t)+H(x,-t)\Big).
\end{align}
Noting that $(-t,t)\ni x\mapsto H(x,t)+H(x,-t)$ is odd while $(-t,t)\ni x\mapsto (I-Q)^{-1}(x,t)+(I-Q)^{-1}(x,-t)$ is even, we deduce
\begin{align*}
	F_{22}-F_{21}=&\,\,\big\langle H\delta_t,(I-Q)^{-1}(\delta_t+\delta_{-t})\big\rangle\stackrel[\eqref{z34}]{\eqref{z33}}{=}-\frac{1}{2}\big\langle S1,(I-Q)^{-1}(\delta_t+\delta_{-t})\big\rangle\\
	=&\,\,-\frac{1}{2}\big\langle(I-Q)^{-1}1,S(\delta_t+\delta_{-t})\big\rangle=-\frac{1}{4}\int_{-t}^t\Big[\big((I-Q)^{-1}Q\big)(x,t)+\big((I-Q)^{-1}Q\big)(x,-t)\Big]\d x\\
	&\hspace{4.46cm}=-\frac{1}{2}\int_{-t}^t\big((I-Q)^{-1}Q\big)(x,t)\d x,
\end{align*}
by utilising, in the third equality, that $(I-Q)^{-1}$ and $S=\frac{1}{2}Q$ commute, and the symmetry $Q=Q^{\ast}$ throughout. Likewise, and relevant to the third factor, which depends on $F_{22}+F_{21}$ in \eqref{z31},
\begin{align*}
	F_{22}+F_{21}=\big\langle H\delta_t,(I-Q)^{-1}(\delta_t-\delta_{-t})\big\rangle=\big\langle H(\delta_t+\delta_{-t}),(I-Q)^{-1}\delta_t\big\rangle,
\end{align*}
where
\begin{equation*}
	\big(H(\delta_t+\delta_{-t})\big)(x)=H(x,t)-H(-x,t)\stackrel{\eqref{z33}}{=}\frac{1}{2}\int_{t-x}^{t+x}W(z)\d z=\frac{1}{2}\int_{-x}^xQ(z,-t)\d z=\frac{1}{2}\int_0^x\big(Q(\delta_t+\delta_{-t})\big)(z)\d z,
\end{equation*}
is odd in $x\in(-t,t)$. Consequently,
\begin{align*}
	F_{22}+F_{21}=&\,\,\big\langle H(\delta_t+\delta_{-t}),(I-Q)^{-1}(I-Q+Q)\delta_t\big\rangle\\
	=&\,\,\frac{1}{2}\int_0^t\big(Q(\delta_t+\delta_{-t})\big)(z)\d z+\big\langle H(\delta_t+\delta_{-t}),(I-Q)^{-1}Q\delta_t\big\rangle\\
	%=&\,\,\frac{1}{2}\int_{-t}^tQ(x,t)\d x+\frac{1}{2}\big\langle H(\delta_t+\delta_{-t}),(I-Q)^{-1}Q(\delta_t-\delta_{-t})\big\rangle\\
	=&\,\,\frac{1}{2}\int_{-t}^tQ(x,t)\d x+\frac{1}{2}\int_{-t}^t\bigg[\int_{-x}^xQ(z,t)\d z\bigg]\big((I-Q)^{-1}Q\big)(x,t)\d x,
\end{align*}
where we use the
%in the third equality that $\frac{1}{2}Q(\delta_t-\delta_{-t})$ and $Q\delta_t$ differ by an even function, we used \eqref{z34} in the fourth equality and 
symmetry $Q=Q^{\ast}$ throughout. The proof of \eqref{c23} is now complete.
\end{proof}

%%%%%%%%%%%%%%%%%%%%%%%%%%%%%

\subsection*{The orthogonal derived class} We have already established that the regularised $2$-determinant $D_{M,2}(t)$ of $M$ in \eqref{c5} with $S,G,H$ as in \eqref{c25},\eqref{z6} and \eqref{c22} is well-defined - simply quote Proposition \ref{pz4}, for our current $S$ is $2S$ therein. Having prepared ourselves in this fashion, we can immediately apply \eqref{c7},
\begin{equation}\label{z67}
	D_{M,2}(\Delta)=D_S(t)\det\Big(\delta_{jk}-G_{jk}((-t,t))\Big)_{j,k=1}^2,\ \ \ \ \ \ G_{jk}((-t,t))=:\langle g_j,f_k\rangle=\int_{-t}^tg_j(x)f_k(x)\d x.
\end{equation}
Here, $D_S(t)$ denotes the Fredholm determinant of $S$ on $L^2(-t,t)$, and we use
\begin{equation*}
	f_1:=-\frac{1}{2}(H\delta_{-t}+H\delta_t),\ \ f_2:=-f_1,\hspace{1cm}g_1:=(I-S)^{-1}\delta_{-t},\ \ g_2:=(I-S)^{-1}\delta_t.
\end{equation*}
As before, $S(x,y)=S(y,x)=S(-x,-y)$ and $H(x,y)=-H(y,x)=-H(-x,-y)$ on $(-t,t)\times(-t,t)$, so
\begin{equation*}
	\langle g_1,f_1\rangle=\langle g_2,f_2\rangle\ \ \ \ \ \ \textnormal{and}\ \ \ \ \ \ \langle g_1,f_2\rangle=\langle g_2,f_1\rangle.
\end{equation*}
Consequently, from \eqref{z67},
\begin{equation*}
	D_{M,2}(t)=D_S(t)\prod_{a=\pm 1}\big(1-\langle g_2,f_2+af_1\rangle\big),
\end{equation*}
and which leads to the below proof of \eqref{c26}.
\begin{proof}[Proof of \eqref{c26}] Observe that $S=Q$ by \eqref{c25} and $\langle g_2,f_2+f_1\rangle=0$ since $f_2=-f_1$. Also,
%\begin{equation*}
%	\langle f_2,g_2+g_1\rangle=\frac{1}{2}\langle H(\delta_t+\delta_{-t}),(I-S)^{-1}(\delta_t+\delta_{-t})\rangle,\ \ \ \ \langle f_2,g_2-g_1\rangle=\frac{1}{2}\langle H(\delta_t+\delta_{-t}),(I-S)^{-1}(\delta_t-\delta_{-t})\rangle
%\end{equation*}
%where $(-t,t)\ni x\mapsto H(x,t)+H(x,-t)$ is odd and $(-t,t)\ni x\mapsto (I-S)^{-1}(x,t)+(I-S)^{-1}(x,-t)$ even. Consequently, $\langle f_2,g_2+g_1\rangle=0$ and
\begin{align}\label{z69}
	\langle g_2,f_2-f_1\rangle=\langle (I-S)^{-1}\delta_t,H(\delta_t+\delta_{-t})\rangle=\langle H(\delta_t+\delta_{-t}),(I-S)^{-1}(I-S+S)\delta_t\rangle,
%
%	
%	
%	\langle H\delta_t,(I-S)^{-1}(\delta_t-\delta_{-t})\rangle-\frac{1}{2}\langle H(\delta_t-\delta_{-t}),(I-S)^{-1}(\delta_t-\delta_{-t})\rangle,
\end{align}
which can be evaluated as in the proof workings of \eqref{c23},
\begin{align*}
	\langle g_2,f_2-f_1\rangle
	=&\,\,\int_0^t\big(S(\delta_{-t}+\delta_t)\big)(z)\d z+\langle H(\delta_t+\delta_{-t}),(I-S)^{-1}S\delta_t\rangle\\
	=&\,\,\int_{-t}^tS(x,t)\d x+\frac{1}{2}\langle S(\delta_t+\delta_{-t}),(I-S)^{-1}S(\delta_t-\delta_{-t})\rangle\\
	=&\,\,\int_{-t}^tS(x,t)\d x+\int_{-t}^t\bigg[\int_{-x}^xS(z,t)\d z\bigg]\big((I-S)^{-1}S\big)(x,t)\d x.
\end{align*}
The proof of \eqref{c26} is now complete.
\end{proof}

\subsection*{RHP and asymptotic results} First things first, we establish \eqref{c30} and \eqref{c31}.
\begin{proof}[Proof of \eqref{c30} and \eqref{c31}] These follow from a careful re-examination of the workings in \cite[page $731-736$]{TW}, noting that the workings on those pages exploit 
\begin{equation*}
	\left(\frac{\partial}{\partial x}+\frac{\partial}{\partial y}\right)Q(x,y)=0,\ \ \ \ \ Q(x,y)=Q(y,x)=Q(-x,-y),
\end{equation*}
but no integrable structures of $Q$. Thus, those workings apply to our \eqref{c24} and they yield \eqref{c30} and \eqref{c31}.
\end{proof}
%\end{lem}
\subsubsection{Asymptotic resolution of RHP \ref{WHRHP}} One can approach the $t\downarrow 0$ regime via a Fredholm and Neumann series argument. That approach yields, for instance,
\begin{equation*}
	D_Q(t)=1-\frac{t}{\pi}\int_{-\infty}^{\infty}\phi(u)\d u+\mathcal{O}\big(t^2\big)\ \ \ \ \ \Rightarrow\ \ \ \ \ X_1^{11}(t,\phi)\sim-\frac{\im}{2\pi}\int_{-\infty}^{\infty}\phi(u)\d u,
\end{equation*}
followed by
\begin{equation*}
	X_1^{21}(t,\phi)=-X_1^{12}(t,\phi)\sim-\im Q(-t,t)\sim-\frac{\im}{2\pi}\int_{-\infty}^{\infty}\phi(u)\d u.
\end{equation*}
Furthermore, by the known integral representation for $X(z)$, cf. \cite{IIKS}, with $F_j,f_k$ defined below \eqref{R19},
\begin{equation*}
	X(z)=\mathbb{I}-\int_{-\infty}^{\infty}\begin{bmatrix}F_1(\lambda)\\ F_2(\lambda)\end{bmatrix}\begin{bmatrix}f_2(\lambda)\\ -f_1(\lambda)\end{bmatrix}^{\top}\frac{\d\lambda}{\lambda-z},\ \ \ z\in\mathbb{C}\setminus\mathbb{R},
\end{equation*}
which yields, by evenness of $\mathbb{R}\ni x\mapsto\phi(x)$,
\begin{equation}\label{Rz}
	X_-(0)\sim\frac{1}{2}\begin{bmatrix}2+\phi(0) & -\phi(0)\\ \phi(0) & 2-\phi(0)\end{bmatrix},\ \ \ t\downarrow 0.
\end{equation}
For $t\rightarrow+\infty$, a nonlinear steepest descent analysis of RHP \ref{WHRHP} goes as follows. We modify RHP \ref{WHRHP} by replacing $Q$ on $L^2(-t,t)$ with $Q\mapsto Q_{\gamma}:=\gamma Q,\gamma\in[0,1]$, so that the Fredholm determinant $D_Q(t;\gamma)$ of $Q_{\gamma}$ becomes the Fredholm determinant of $V_{\gamma}$ on $L^2(\mathbb{R})$, see \eqref{c28}. Here, $V_{\gamma}:L^2(\mathbb{R})\rightarrow L^2(\mathbb{R})$ is trace class with kernel
\begin{equation*}
	V_{\gamma}(x,y):=\gamma V(x,y)=\gamma\sqrt{\phi(x)}\,\frac{\sin t(x-y)}{\pi(x-y)}\sqrt{\phi(y)},\ \ \ \ \ x,y\in\mathbb{R},
\end{equation*}
and we denote with $X(z)=X(z;t,\gamma\phi)\in\mathbb{C}^{2\times 2}$, from now on, the solution of RHP \ref{WHRHP}, where $\phi$ is replaced with $\gamma\phi$. Note that $D_Q(t)=D_Q(t;1)$ and $D_Q(t;0)\equiv 1$; moreover
\begin{equation}\label{R19}
	\ln D_Q(t)=\int_0^1\frac{\partial}{\partial\gamma}\ln D_Q(t;\gamma)\d\gamma=-\int_0^1\left[\int_{-\infty}^{\infty}R_{\gamma}(x,x)\d x\right]\frac{\d\gamma}{\gamma},
\end{equation}
with $R_{\gamma}:=(I-V_{\gamma})^{-1}-I:L^2(\mathbb{R})\rightarrow L^2(\mathbb{R})$, which exists certainly since we assume $\|Q\|<1$ in operator norm and $\gamma\in[0,1]$. We shall use \eqref{R19} in our asymptotic study, recalling the important fact that $V_{\gamma}$ is an integrable operator, cf. \cite{IIKS}, and so $R_{\gamma}$ has kernel
\begin{equation*}
	R_{\gamma}(x,y)=\frac{F_1(x)G_1(y)+F_2(x)G_2(y)}{x-y},\ \ \ x,y\in\mathbb{R},
\end{equation*}
with, using a fixed branch for the square roots throughout,
\begin{equation*}
	\begin{bmatrix}F_1(x)\\ F_2(x)\end{bmatrix}=\begin{bmatrix}((I-V_{\gamma})^{-1}f_1)(x)\smallskip\\ ((I-V_{\gamma})^{-1}f_2)(x)\end{bmatrix}=X_{\pm}(x)\begin{bmatrix}f_1(x)\\ f_2(x)\end{bmatrix},\ x\in\mathbb{R};\ \ \ f_1(x):=\sqrt{\frac{\gamma\phi(x)}{2\pi\im}}\e^{\im tx},\ \ \ f_2(x):=\sqrt{\frac{\gamma\phi(x)}{2\pi\im}}\e^{-\im tx},
\end{equation*}
and, by symmetry of $V_{\gamma}$,
\begin{equation*}
	\begin{bmatrix}G_1(x)\\ G_2(x)\end{bmatrix}=X_{\pm}^{\top}(x)^{-1}\begin{bmatrix}f_2(x)\\ -f_1(x)\end{bmatrix},\ \ x\in\mathbb{R}.%G_2(x)=-F_1(x),\ \ \ \ G_1(x)=F_2(x),\ \ x\in\mathbb{R}.
\end{equation*}
Thus, for \eqref{R19}, we find the simplified expression
\begin{equation}\label{R20}
	\ln D_Q(t)=-\int_0^1\left[\int_{-\infty}^{\infty}\big(F_1'(x)G_1(x)+F_2'(x)G_2(x)\big)\d x\right]\frac{\d\gamma}{\gamma},\ \ \ \ \ \ F_j'(x)=\frac{\d}{\d x}F_j(x),\ \ x\in\mathbb{R}.
\end{equation}
Moving forward with the actual asymptotic resolution of RHP \ref{WHRHP}, we define the $g$-function
\begin{equation*}
	g(z)=g(z;\gamma\phi):=\frac{\im}{2\pi}\int_{-\infty}^{\infty}\ln\big(1-\gamma\phi(\lambda)\big)\frac{\d\lambda}{\lambda-z},\ \ \ z\in\mathbb{C}\setminus\mathbb{R},
\end{equation*}
so that by the imposed properties on $\phi$, we have at once
\begin{equation*}
	g_+(z)-g_-(z)=-\ln\big(1-\gamma\phi(z)\big),\ z\in\mathbb{R};\ \ \ g_{\pm}(z):=\lim_{\epsilon\downarrow 0}g(z\pm\im\epsilon),
\end{equation*}
followed by
\begin{equation*}
	g(z)=\frac{1}{2\pi\im z}\int_{-\infty}^{\infty}\ln\big(1-\gamma\phi(\lambda)\big)\d\lambda+o\big(z^{-1}\big),\ \ \ \ z\rightarrow\infty,\ z\notin\mathbb{R}.
\end{equation*}
What results for the transformed function
\begin{equation}\label{R21}
	S(z;t,\gamma\phi):=X(z;t,\gamma\phi)\e^{g(z;\gamma\phi)\sigma_3},\ \ \ z\in\mathbb{C}\setminus\mathbb{R},
\end{equation}
is summarised below.
\begin{problem}\label{master6} Let $t>0$ and $\phi$ be as in Theorem \ref{Kactheo}. The function $S(z)=S(z;t,\gamma\phi)\in\mathbb{C}^{2\times 2}$ defined in 
\eqref{R21} is uniquely determined by the following properties.
\begin{enumerate}
	\item[(1)] $z\mapsto S(z)$ is analytic for $z\in\mathbb{C}\setminus\mathbb{R}$ and extends continuously on the closure of $\mathbb{C}\setminus\mathbb{R}$.
	\item[(2)] $z\mapsto S(z)$ admits continuous limiting values $S_{\pm}(z)=\lim_{\epsilon\downarrow 0}S(z\pm\im\epsilon)$ on $\mathbb{R}\ni z$ that satisfy $S_+(z)=S_-(z)G_S(z)$, with $G_S(z)=G_S(z;t,\gamma\phi)$ equal to
	\begin{equation*}
		G_S(z)=\begin{bmatrix}1 & \eta(z)\e^{2\im tz-2g_+(z)}\\ -\eta(z)\e^{-2\im tz+2g_-(z)} & 1-\gamma^2\phi^2(z)\end{bmatrix};\ \ \ \eta(z)=\eta(z;\gamma\phi):=\frac{\gamma\phi(z)}{1-\gamma\phi(z)},\ z\in\mathbb{R}.
	\end{equation*}
	\item[(3)] As $z\rightarrow\infty$ in $\mathbb{C}\setminus\mathbb{R}$,
	\begin{equation*}
		S(z)=\mathbb{I}+\frac{1}{z}S_1+o\big(z^{-1}\big);\ \ \ \ S_1=S_1(t,\gamma\phi)=X_1(t,\gamma\phi)+\frac{\sigma_3}{2\pi\im}\int_{-\infty}^{\infty}\ln\big(1-\gamma\phi(\lambda)\big)\d\lambda.
	\end{equation*}
\end{enumerate}
\end{problem}
By \eqref{c32}, there exists $\epsilon>0$ such that $z\mapsto\eta(z)$ is analytic for $|\Im z|\leq\epsilon$ and $\lim_{|z|\rightarrow\infty}\eta(z)=0$ in the same strip. Thus, we factorise
\begin{equation*}
	G_S(z)=\begin{bmatrix}1 & 0\\ -\eta(z)\e^{-2\im tz+2g_-(z)} & 1\end{bmatrix}\begin{bmatrix}1 & \eta(z)\e^{2\im tz-2g_+(z)}\\ 0 & 1\end{bmatrix},\ \ \ z\in\mathbb{R},
\end{equation*}
and transform RHP \ref{master6} as follows. Consider
\begin{equation}\label{R22}
	T(z;t,\gamma\phi):=S(z;t,\gamma\phi)\begin{cases}\displaystyle \begin{bmatrix}1 & -\eta(z)\e^{2\im tz-2g(z)}\\ 0 & 1\end{bmatrix},&\Im z\in(0,\epsilon)\smallskip\\
	\begin{bmatrix}1 & 0\\ -\eta(z)\e^{-2\im tz+2g(z)} & 1\end{bmatrix},&\Im z\in(-\epsilon,0)\smallskip\\ \mathbb{I},&\textnormal{else}\end{cases},
\end{equation}
which has the following analytic and asymptotic properties.
\begin{problem}\label{master7} Let $t>0$ and $\phi$ be as in Theorem \ref{Kactheo}. The function $T(z)=T(z;t,\gamma\phi)\in\mathbb{C}^{2\times 2}$ defined in \eqref{R22} is uniquely determined by the following properties.
\begin{enumerate}
	\item[(1)] $z\mapsto T(z)$ is analytic for $z\in\mathbb{C}\setminus\Sigma_T$, with $\Sigma_T:=\{z\in\mathbb{C}:\ \Im z=\epsilon,\ \Im z=-\epsilon\}$, and extends continuously on the closure of $\mathbb{C}\setminus\Sigma_T$.
	\item[(2)] $z\mapsto T(z)$ admits continuous limiting values $T_{\pm}(z)=\lim_{\epsilon\downarrow 0}T(z\pm\im\epsilon)$ on $\Sigma_T\ni z$ that satisfy $T_+(z)=T_-(z)G_T(z)$, with $G_T(z)=G_T(z;t,\gamma\phi)$ equal to
	\begin{equation*}
		G_T(z)=\begin{bmatrix}1 & \eta(z)\e^{2\im tz-2g(z)}\\ 0 & 1\end{bmatrix},\ \ \Im z=\epsilon;\hspace{1cm}G_T(z)=\begin{bmatrix}1 & 0\\ -\eta(z)\e^{-2\im tz+2g(z)}&1\end{bmatrix},\ \ \Im z=-\epsilon.
	\end{equation*}
	\item[(3)] As $z\rightarrow\infty$, with $z\notin\Sigma_T$,
	\begin{equation*}
		T(z)=\mathbb{I}+\frac{1}{z}S_1+o\big(z^{-1}\big).
	\end{equation*}
\end{enumerate}
\end{problem}
Since $|\e^{\mp 2\im tz}|=\e^{-2t\epsilon}$ for $\Im z=\pm\epsilon$ and $\eta(z)\e^{-2g(z)},\eta(z)\e^{2g(z)}$ are bounded on $\Im z=\epsilon,\Im z=-\epsilon$, we deduce the existence of $c=c(\epsilon)>0$ such that
\begin{equation*}
	\|G_t(\cdot;t,\gamma\phi)-\mathbb{I}\|_{L^2\cap L^{\infty}(\Sigma_T)}\leq c\,\e^{-2\epsilon t}\ \ \ \ \forall\,t>0,\ \gamma\in[0,1].
\end{equation*}
This says that RHP \ref{master7} is a small norm problem and so uniquely solvable as $t\rightarrow+\infty$, cf. \cite{DZ}. In addition,
\begin{equation}\label{R23}
	T(z)=\mathbb{I}+\mathcal{O}\big(\e^{-2\epsilon t}\big)\ \ \ \ \textnormal{as}\ \ t\rightarrow+\infty,
\end{equation}
uniformly for $z\in\mathbb{C}\setminus\Sigma_T$, with some $\epsilon>0$. At this point, we can return to \eqref{R20} and \eqref{c27}.
\subsubsection{Derivation of asymptotics for $D_M(t)$ and $D_{M,2}(t)$} First, by \eqref{c27},\eqref{R21}, and \eqref{R22},
\begin{equation}\label{R24}
	\frac{\d}{\d t}\ln D_Q(t)=-2\im X_1^{11}(t,\phi)=\frac{1}{\pi}\int_{-\infty}^{\infty}\ln\big(1-\phi(\lambda)\big)\d\lambda-2\im S_1^{11}(t,\phi)\stackrel{t\rightarrow+\infty}{=}-s(0)+\mathcal{O}\big(t^{-\infty}\big),
\end{equation}
with the shorthand
\begin{equation*}
	s(x)\equiv s(x;\phi)=-\frac{1}{\pi}\int_{-\infty}^{\infty}\ln\big(1-\phi(\lambda)\big)\e^{\im x\lambda}\d\lambda,\ \ \ \ x\in\mathbb{R}.
\end{equation*}
Consequently, after indefinite $t$-integration in \eqref{R24},
\begin{equation}\label{R25}
	\ln D_Q(t)=-ts(0)+\tau+\mathcal{O}\big(t^{-\infty}\big),\ \ \ t\rightarrow+\infty,
\end{equation}
where $\tau$ is $t$-independent. The compute this constant, we now utilise \eqref{R20}. First, taking limiting values from the $(-)$ side for definiteness, and using $f=[f_1,f_2]^{\top},F=[F_1,F_2]^{\top}$, and $G=[G_1,G_2]^{\top}$, we have
\begin{align*}
	F(x)=&\,\,X_-(x)f(x)\stackrel{\eqref{R21}}{=}S_-(x)\e^{-g_-(x)\sigma_3}f(x)\stackrel{\eqref{R22}}{=}T(x)\begin{bmatrix}1 & 0\\ \eta(x)\e^{-2\im tx+2g_-(x)} & 1\end{bmatrix}\e^{-g_-(x)\sigma_3}f(x)\\
	=&\,\,T(x)\e^{-g_-(x)\sigma_3}\sqrt{\frac{\gamma\phi(x)}{2\pi\im}}\begin{bmatrix}\e^{\im tx}\\ \e^{-\im tx}(1-\gamma\phi(x))^{-1}\end{bmatrix},\ \ x\in\mathbb{R}.
\end{align*}
Similarly
\begin{equation*}
	G(x)=\big(T^{-1}(x)\big)^{\top}\e^{g_-(x)\sigma_3}\sqrt{\frac{\gamma\phi(x)}{2\pi\im}}\begin{bmatrix}\e^{-\im tx}(1-\gamma\phi(x))^{-1}\\ -\e^{\im tx}\end{bmatrix},\ \ x\in\mathbb{R}.
\end{equation*}
Here, by H\"older continuity of $\phi$ and Plemelj's formula,
\begin{equation*}
	g_-(x)=\frac{1}{2}\ln\big(1-\gamma\phi(x)\big)-\frac{1}{2\pi\im}\,\textnormal{pv}\int_{-\infty}^{\infty}\ln\big(1-\gamma\phi(\lambda)\big)\frac{\d\lambda}{\lambda-x},\ \ x\in\mathbb{R}.
\end{equation*}
Focusing thus on the innermost integrand in \eqref{R20}, we find as $t\rightarrow+\infty$, using in particular \eqref{R23},
\begin{equation*}
	\big(F'(x)\big)^{\top}G(x)=-\frac{t\gamma}{\pi}\frac{\partial}{\partial\gamma}\ln\big(1-\gamma\phi(x)\big)-\frac{1}{2\pi^2}\frac{\gamma\phi(x)}{1-\gamma\phi(x)}\frac{\d}{\d x}\left[\textnormal{pv}\int_{-\infty}^{\infty}\ln\big(1-\gamma\phi(\lambda)\big)\frac{\d\lambda}{\lambda-x}\right]+\mathcal{O}\big(\gamma|\phi(x)|t^{-\infty}\big),
\end{equation*}
uniformly for any $\gamma\in[0,1]$ and $x\in\mathbb{R}$. Here,
\begin{equation*}
	\int_0^1\left[\int_{-\infty}^{\infty}\frac{t\gamma}{\pi}\frac{\partial}{\partial\gamma}\ln\big(1-\gamma\phi(x)\big)\d x\right]\frac{\d\gamma}{\gamma}=-ts(0),
\end{equation*}
and, integrating by parts,
\begin{align}
	\frac{1}{2\pi^2}&\,\int_{-\infty}^{\infty}\frac{\phi(x)}{1-\gamma\phi(x)}\frac{\d}{\d x}\left[\textnormal{pv}\int_{-\infty}^{\infty}\ln\big(1-\gamma\phi(\lambda)\big)\frac{\d\lambda}{\lambda-x}\right]\d x=-\frac{1}{2\pi^2}\int_{-\infty}^{\infty}\Bigg(\frac{\d}{\d x}\bigg[\frac{\phi(x)}{1-\gamma\phi(x)}\bigg]\Bigg)\nonumber\\
	&\times\Bigg(\textnormal{pv}\int_{-\infty}^{\infty}\ln\big(1-\gamma\phi(\lambda)\big)\frac{\d\lambda}{\lambda-x}+\im\pi\ln\big(1-\gamma\phi(x)\big)\Bigg)\d x,\label{R26}
\end{align}
where we used that, by evenness of $x\mapsto\phi(x)$,
\begin{equation*}
	\int_{-\infty}^{\infty}\frac{\d}{\d x}\bigg[\frac{\phi(x)}{1-\gamma\phi(x)}\bigg]\ln\big(1-\gamma\phi(x)\big)\d x=0\ \ \ \ \forall\,\gamma\in[0,1].
\end{equation*}
Now, deforming the $x$-integration contour in \eqref{R26} to $\mathbb{R}+\im\epsilon$ by Cauchy's theorem, we deduce
\begin{align}
	\textnormal{RHS in}\ \eqref{R26}=&\,-\frac{1}{2\pi^2}\int_{\mathbb{R}+\im\epsilon}\Bigg(\frac{\d}{\d z}\bigg[\frac{\phi(z)}{1-\gamma\phi(z)}\bigg]\Bigg)\Bigg(\int_{-\infty}^{\infty}\ln\big(1-\gamma\phi(\lambda)\big)\frac{\d\lambda}{\lambda-z}\Bigg)\d z\nonumber\\
	=&\,\frac{1}{2\pi^2}\int_{\mathbb{R}+\im\epsilon}\frac{\phi(z)}{1-\gamma\phi(z)}\Bigg(\int_{-\infty}^{\infty}\ln\big(1-\gamma\phi(\lambda)\big)\frac{\d\lambda}{(\lambda-z)^2}\Bigg)\d z,\label{R27}
\end{align}
after another integration by parts. Since $(\lambda-z)^{-2}=-\int_0^{\infty}x\e^{\im(z-\lambda)x}\d x$ for $\lambda\in\mathbb{R},z\in\mathbb{R}+\im\epsilon$, we furthermore find
\begin{align*}
	\textnormal{RHS in}\ \eqref{R27}=\frac{1}{2}\int_0^{\infty}x\bigg[\frac{\partial}{\partial\gamma}s(x;\gamma\phi)\bigg]s(-x;\gamma\phi)\d x=\frac{1}{4}\frac{\partial}{\partial\gamma}\bigg[\int_0^{\infty}xs(x;\gamma\phi)s(-x;\gamma\phi)\d x\bigg].
\end{align*}
In summary, as $t\rightarrow+\infty$,
\begin{equation*}
	\ln D_Q(t)=-\int_0^1\bigg[\int_{-\infty}^{\infty}\big(F'(x)\big)^{\top}G(x)\d x\bigg]\frac{\d\gamma}{\gamma}=-ts(0)+\frac{1}{4}\int_0^{\infty}xs(x;\phi)s(-x;\phi)\d x+\mathcal{O}\big(t^{-\infty}\big),
\end{equation*}
which identities $\tau$ back in \eqref{R25} to be of the form
\begin{equation}\label{R28}
	\tau=\frac{1}{4}\int_0^{\infty}xs(x;\phi)s(-x;\phi)\d x.
\end{equation}
Moving on to \eqref{c30} and \eqref{c31}, we have
\begin{equation*}
	\omega(t)=2\im\int_0^tX_1^{21}(s,\phi)\d s\stackrel{\eqref{R21}}{=}2\im\int_0^tS_1^{21}(s,\phi)\d s\stackrel{\eqref{R23}}{=}2\im\int_0^{\infty}S_1^{21}(s,\phi)\d s+\mathcal{O}\big(t^{-\infty}\big),\ \ t\rightarrow+\infty,
\end{equation*}
and are thus in need of the total integral
\begin{equation*}
	P=P(\phi):=2\im\int_0^{\infty}X_1^{21}(s,\phi)\d s=-2\im\int_0^{\infty}X_1^{12}(s,\phi)\d s.
\end{equation*}
Similar to our calculations surrounding \eqref{R113}, we begin with the Zhakarov-Shabat system
\begin{equation}\label{R29}
	\frac{\partial W}{\partial t}=\bigg\{\im z\sigma_3+2\im\begin{bmatrix}0&X_1^{21}\\ X_1^{21} &0\end{bmatrix}\bigg\}W,
\end{equation}
for $W(z;t,\phi):=X(z;t,\phi)\e^{\im tz\sigma_3},z\in\mathbb{C}\setminus\mathbb{R}$. What results is that
\begin{equation*}
	V(t;\phi):=\lim_{\substack{z\rightarrow 0\\ \Im z<0}}X(z;t,\phi),\ \ \ t>0,
\end{equation*}
equals
\begin{equation}\label{R30}
	V(t;\phi)=\begin{bmatrix}\cosh\nu & \sinh\nu\\ \sinh\nu & \cosh\nu\end{bmatrix}V(t_0;\phi),\ \ \ \ \ \nu=\nu(t,t_0,\phi):=2\im\int_{t_0}^tX_1^{21}(s,\phi)\d s,\ \ \ t,t_0>0.
\end{equation}
Since $P$ is the limit of $\nu(t,t_0,\phi)$ as $t\rightarrow+\infty$ and $t_0\downarrow 0$, we begin by computing
\begin{align*}
	\lim_{t\rightarrow+\infty}V(t,\phi)=\lim_{t\rightarrow+\infty}&\left\{\lim_{\substack{z\rightarrow 0\\ \Im z<0}}X(z;t,\phi)\right\}\stackrel{\eqref{R22}}{=}\lim_{t\rightarrow+\infty}T(0;t,\phi)\e^{-g_-(0)\sigma_3}\begin{bmatrix}1&0\\ \eta(0)&1\end{bmatrix}\\
	\stackrel{\eqref{R23}}{=}&\,\e^{-g_-(0)\sigma_3}\begin{bmatrix}1&0\\ \eta(0)&1\end{bmatrix}=\begin{bmatrix}1&0\\ \phi(0)&1\end{bmatrix}\big(1-\phi(0)\big)^{-\frac{1}{2}\sigma_3}.
\end{align*}
On the other hand,
\begin{equation*}
	\lim_{t_0\downarrow 0}V(t_0,\phi)=\lim_{t_0\downarrow 0}\left\{\lim_{\substack{z\rightarrow 0\\ \Im z<0}}X(z;t_0,\phi)\right\}\stackrel{\eqref{Rz}}{=}\frac{1}{2}\begin{bmatrix}2+\phi(0)&-\phi(0)\\ \phi(0) & 2-\phi(0)\end{bmatrix},
\end{equation*}
and so in total,
\begin{equation*}
	V(+\infty,\phi)V(0,\phi)^{-1}=\frac{1}{2\sqrt{1-\phi(0)}}\begin{bmatrix}2-\phi(0)& \phi(0)\\ \phi(0) & 2-\phi(0)\end{bmatrix}.
\end{equation*}
In turn, by \eqref{R30},
\begin{equation*}
	\cosh\nu(+\infty,0,\phi)=\frac{2-\phi(0)}{2\sqrt{1-\phi(0)}},\ \ \ \ \ \ \ \sinh\nu(+\infty,0,\phi)=\frac{\phi(0)}{2\sqrt{1-\phi(0)}},
\end{equation*}
which implies
\begin{equation*}
	\exp\left[2\im\int_0^{\infty}X_1^{21}(s,\phi)\d s\right]=\e^{\nu(+\infty,0,\phi)}=\cosh\nu(+\infty,0,\phi)+\sinh\nu(+\infty,0,\phi)=\frac{1}{\sqrt{1-\phi(0)}},
\end{equation*}
and consequently, as $t\rightarrow\infty$,
\begin{align}\label{R31}
	\e^{-\omega(t)}=\sqrt{1-\phi(0)}+\mathcal{O}\big(t^{-\infty}\big),\ \ \cosh^2\bigg(\frac{\omega(t)}{2}\bigg)=\bigg(\frac{1}{2}\sqrt[4]{1-\phi(0)}+\frac{1}{2\sqrt[4]{1-\phi(0)}}\bigg)^2+\mathcal{O}\big(t^{-\infty}\big).
\end{align}
We are left to summarise our workings.
\begin{proof}[Proof of \eqref{c33} and \eqref{c34}] Simply combine \eqref{R25},\eqref{R28}, and \eqref{R31} with \eqref{c30},\eqref{c31} to arrive at \eqref{c33} and \eqref{c34}.
\end{proof}

\section{Comment about multiplicative Hankel composition kernels}\label{notknow}

One might expect that, besides \eqref{c8},\eqref{c13} and \eqref{c21},\eqref{c25}, there exists a third class of main kernels which firstly induce algebraic simplifications of the Fredholm Pfaffian as in \eqref{c11},\eqref{c16} and \eqref{c23},\eqref{c26}, and secondly admit a Riemann-Hilbert problem characterisation. This is not quite the case, however, as we will now explain, specifying to the symplectic derived class. A similar discussion can be held in the orthogonal derived class.

\subsection*{A symplectic derived example} We begin with the following concrete example in which $\Delta=(0,t)\subset\mathbb{R}$, with $t>0$ and, as in \cite[Theorem $1.1$]{DGKV},
\begin{equation}\label{z35}
	S(x,y):=\frac{1}{2}\bigg(\frac{1}{4}\int_0^1J_{\alpha}(\sqrt{xu})J_{\alpha}(\sqrt{uy})\d u-\frac{J_{\alpha-1}(\sqrt{x})}{8\sqrt{x}}\int_0^yJ_{\alpha+1}(\sqrt{z})\frac{\d z}{\sqrt{z}}\bigg),\ \ \ \ x,y\in(0,t),
\end{equation}
with $J_{\alpha}:(0,\infty)\rightarrow\mathbb{R}$ as a Bessel function of order $\alpha>1$. The asymmetric appearance of $\alpha$ in \eqref{z35} is a new feature unseen in \eqref{c8}  but a necessary one, as shown in the following proof workings.
\begin{prop}\label{Bessfish} If $S,G,H:L^2(0,t)\rightarrow L^2(0,t)$ denote the integral operators with kernels \eqref{z35},\eqref{z6} and
\begin{equation}\label{z355}
	H(x,y)=\int_0^xS(z,y)\d z,
\end{equation}
then all three are trace class and we have that
\begin{equation*}
	G(x,y)=-G(y,x),\ \ \ \ \ \ H(x,y)=-H(y,x)\ \ \ \ \textnormal{on}\ \ \ (0,t)\times(0,t).
\end{equation*}
\end{prop}
\begin{proof} By special function properties, for all $t>0$, and for all $\alpha>1$,
\begin{equation*}
	\int_0^t\int_0^1\big|J_{\alpha}(\sqrt{xu})\big|^2\d u\d x<\infty,\ \ \ \ \int_0^t\big|J_{\alpha-1}(\sqrt{x})\big|^2\frac{\d x}{x}<\infty,\ \ \ \ \int_0^t\bigg|\int_0^yJ_{\alpha+1}(\sqrt{z})\frac{\d z}{\sqrt{z}}\bigg|^2\d y<\infty,
\end{equation*}
and so $S$ is a linear combination of a Hilbert-Schmidt composition integral operator and a rank-$1$ integral operator, both on $L^2(0,t)$; this makes $S$ trace class on the same space. Since, also for all $t>0$ and $\alpha>1$,
\begin{equation*}
	\int_0^t\int_0^1\frac{u}{y}\big|J_{\alpha}'(\sqrt{uy})\big|^2\d u\d y<\infty,\ \ \ \ \int_0^t\big|J_{\alpha+1}(\sqrt{y})\big|^2\frac{\d y}{y}<\infty,
\end{equation*}
the same goes for the integral operator $G:L^2(0,t)\rightarrow L^2(0,t)$. Lastly, with $\alpha>0$ being admissible,
\begin{equation*}
	\int_0^t\int_0^1\bigg|\int_0^xJ_{\alpha}(\sqrt{zu})\d z\bigg|^2\d u\d x<\infty,\ \ \ \ \int_0^t\bigg|\int_0^xJ_{\alpha-1}(\sqrt{z})\frac{\d z}{\sqrt{z}}\bigg|^2\d x<\infty\ \ \ \ \ \forall\,t>0,
\end{equation*}
$H:L^2(0,t)\rightarrow L^2(0,t)$ is also trace class on $L^2(0,t)$. While the trace class properties of $S,G,H$ rely solely on local integrability of $(0,\infty)\ni x\mapsto J_{\alpha}(x)$ and the endpoint behaviour
\begin{equation*}
	 J_{\alpha}(x)\sim c_0(\alpha)x^{\alpha},\ \ \ x\downarrow 0\ \ \ \ \ \ \textnormal{with some}\ \ \ c_0(\alpha)\neq 0,
\end{equation*}
skew-symmetry of $G(x,y),H(x,y)$ constrains the special functions appearing in \eqref{z35} much more. Indeed, assuming $J_{\alpha}(x)$ in \eqref{z35} is of the form $J_{\alpha}(x)=x^{\alpha}\sum_{k=0}^{\infty}c_k(\alpha)x^{2k}$, uniformly and absolutely convergent on compact subsets of $(0,\infty)$, with some coefficients $c_k(\alpha)$, we have
\begin{equation*}
	S(x,y)=\frac{1}{8}(xy)^{\frac{\alpha}{2}}\sum_{k,\ell=0}^{\infty}\bigg(\frac{c_k(\alpha)c_{\ell}(\alpha)}{k+\ell+\alpha+1}x^ky^{\ell}-\frac{c_k(\alpha-1)c_{\ell}(\alpha+1)}{2\ell+\alpha+2}x^{k-1}y^{\ell+1}\bigg),
\end{equation*}
and so upon $y$-differentiation,
\begin{equation}\label{z36}
	G(x,y)=-\frac{1}{8}(xy)^{\frac{\alpha}{2}}\sum_{k,\ell=0}^{\infty}\bigg(\frac{c_k(\alpha)c_{\ell}(\alpha)}{k+\ell+\alpha+1}\Big(\ell+\frac{\alpha}{2}\Big)x^ky^{\ell-1}-\frac{1}{2}c_k(\alpha-1)c_{\ell}(\alpha+1)x^{k-1}y^{\ell}\bigg),
\end{equation}
followed by, upon $x$-integration,
\begin{equation}\label{z37}
	H(x,y)=\frac{1}{8}(xy)^{\frac{\alpha}{2}}\sum_{k,\ell=0}^{\infty}\bigg(\frac{c_k(\alpha)c_{\ell}(\alpha)}{k+\ell+\alpha+1}\frac{x^{k+1}y^{\ell}}{k+\frac{\alpha}{2}+1}-\frac{c_k(\alpha-1)c_{\ell}(\alpha+1)}{2\ell+\alpha+2}\frac{x^ky^{\ell+1}}{k+\frac{\alpha}{2}}\bigg).
\end{equation}
Assuming the $\alpha$-recursion
\begin{equation}\label{z38}
	c_k(\alpha+1)=\frac{c_k(\alpha)}{k+\alpha+1},\ \ k\in\mathbb{Z}_{\geq 0}\ \ \ \ \ \ \textnormal{with some}\ \ \ c_0(\alpha)\neq 0,
\end{equation}
is in place, we can use the resulting identity
\begin{equation}\label{z388}
	c_k(\alpha)c_{\ell}(\alpha)=c_k(\alpha-1)c_{\ell}(\alpha+1)\frac{\ell+\alpha+1}{k+\alpha},\ \ \ k,\ell\in\mathbb{Z}_{\geq 0},
\end{equation}
to obtain, after first relabelling $k\leftrightarrow\ell$ in the sums for $G(x,y)$ in \eqref{z36},
\begin{equation}\label{z39}
	G(x,y)=\frac{1}{8}(yx)^{\frac{\alpha}{2}}\sum_{k,\ell=0}^{\infty}\bigg(\frac{c_k(\alpha)c_{\ell}(\alpha)}{k+\ell+\alpha+1}\Big(\ell+\frac{\alpha}{2}\Big)d_{k\ell}(\alpha)y^kx^{\ell-1}-\frac{1}{2}c_k(\alpha-1)c_{\ell}(\alpha+1)\frac{x^{\ell}y^{k-1}}{d_{\ell k}(\alpha)}\bigg).
\end{equation}
Here, for any $k,\ell\in\mathbb{Z}_{\geq 0}$,
\begin{align*}
	d_{k\ell}(\alpha)=&\,\frac{k+\ell+\alpha+1}{2\ell+\alpha}\frac{\ell+\alpha}{k+\alpha+1}=1-\frac{\ell(k-\ell+1)}{(2\ell+\alpha)(k+\alpha+1)},\\
	\frac{1}{d_{\ell k}(\alpha)}=&\,\frac{2k+\alpha}{k+\ell+\alpha+1}\frac{\ell+\alpha+1}{k+\alpha}=1+\frac{k(\ell-k+1)}{(k+\alpha)(k+\ell+\alpha+1)},
\end{align*}
and so \eqref{z39} translates into the following identity, valid for $(x,y)\in(0,t)\times(0,t)$,
\begin{equation}\label{z40}
	G(x,y)+G(y,x)=-\frac{1}{16}(xy)^{\frac{\alpha}{2}}\sum_{k,\ell=0}^{\infty}\frac{c_k(\alpha)c_{\ell}(\alpha)}{k+\ell+\alpha+1}\bigg(\frac{\ell(k-\ell+1)}{k+\alpha+1}y^kx^{\ell-1}+\frac{k(\ell-k+1)}{\ell+\alpha+1}x^{\ell}y^{k-1}\bigg).
\end{equation}
For the right hand side in \eqref{z40} to vanish identically, it is sufficient to impose, with \eqref{z38}, the $k$-recursion
\begin{equation}\label{z41}
	c_{k+1}(\alpha)=\frac{c_k(\alpha)}{(k+1)(k+\alpha+1)},\ \ \ k\in\mathbb{Z}_{\geq 0}\ \ \ \ \ \ \textnormal{with some}\ \ \ c_0(\alpha)\neq 0.
\end{equation}
For then
\begin{equation}\label{z42}
	c_{k-1}(\alpha)c_{\ell+1}(\alpha)\frac{\ell+1}{k+\alpha}=c_k(\alpha)c_{\ell}(\alpha)\frac{k}{\ell+\alpha+1},\ \ \ k\in\mathbb{N},\ \ \ell\in\mathbb{Z}_{\geq 0},
\end{equation}
which implies the following formula: 
\begin{align*}
	\sum_{k,\ell=0}^{\infty}\frac{c_k(\alpha)c_{\ell}(\alpha)}{k+\ell+\alpha+1}&\,\frac{\ell(k-\ell+1)}{k+\alpha+1}y^kx^{\ell-1}=\sum_{k=1}^{\infty}\sum_{\ell=0}^{\infty}\frac{c_{k-1}(\alpha)c_{\ell+1}(\alpha)}{k+\ell+\alpha+1}\frac{(\ell+1)(k-\ell-1)}{k+\alpha}y^{k-1}x^{\ell}\\
	&\,=\sum_{k,\ell=0}^{\infty}\frac{c_k(\alpha)c_{\ell}(\alpha)}{k+\ell+\alpha+1}\frac{k(k-\ell-1)}{\ell+\alpha+1}y^{k-1}x^{\ell},
\end{align*}
where we shifted summation indices according to the rule $k\rightarrow k-1,\ell\rightarrow\ell+1$ in the first equality and used \eqref{z42} in the second. In turn, the two remaining sums in the right hand side of \eqref{z40} cancel out identically, and so skew-symmetry of $G(x,y)$ follows. The argument for $H(x,y)$ is analogous: using only \eqref{z388}, one deduces
\begin{align}
	H(x,y)+H(y,x)=\frac{1}{8}(xy)^{\frac{\alpha}{2}}\sum_{k,\ell=0}^{\infty}&\,\frac{c_k(\alpha)c_{\ell}(\alpha)}{k+\ell+\alpha+1}\bigg(\frac{\ell(k-\ell+1)y^{k+1}x^{\ell}}{(2\ell+\alpha)(k+\alpha+1)(k+\frac{\alpha}{2}+1)}\nonumber\\
	&\,\,+\frac{k(\ell-k+1)y^kx^{\ell+1}}{(2k+\alpha)(\ell+\alpha+1)(\ell+\frac{\alpha}{2}+1)}\bigg),\ \ \ (x,y)\in(0,t)\times(0,t),\label{z43}
\end{align}
and then utilises \eqref{z42},
\begin{align*}
	\sum_{k,\ell=0}^{\infty}\frac{c_k(\alpha)c_{\ell}(\alpha)}{k+\ell+\alpha+1}&\,\,\frac{\ell(k-\ell+1)y^{k+1}x^{\ell}}{(2\ell+\alpha)(k+\alpha+1)(k+\frac{\alpha}{2}+1)}=\sum_{k=1}^{\infty}\sum_{\ell=0}^{\infty}\frac{c_{k-1}(\alpha)c_{\ell+1}(\alpha)}{k+\ell+\alpha+1}\frac{(\ell+1)(k-\ell-1)y^kx^{\ell+1}}{(2k+\alpha)(k+\alpha)(\ell+\frac{\alpha}{2}+1)}\\
	=&\,\,\sum_{k,\ell=0}^{\infty}\frac{c_k(\alpha)c_{\ell}(\alpha)}{k+\ell+\alpha+1}\frac{k(k-\ell-1)y^kx^{\ell+1}}{(2k+\alpha)(\ell+\alpha+1)(\ell+\frac{\alpha}{2}+1)}.
\end{align*}
So, the remaining sum in the right hand side of \eqref{z43} vanishes identically, and thus skew-symmetry of $H(x,y)$ is established. Lastly, one can notice that the recursions \eqref{z38} and \eqref{z41} determine $c_k(\alpha)$ to be of the form
\begin{equation*}
	c_k(\alpha)=\frac{c}{k!\,\Gamma(k+\alpha+1)},\ \ \ \ \ \textnormal{with some}\ \ c\neq 0,
\end{equation*}
i.e. Proposition \ref{Bessfish} is now certainly established for the Bessel function choice $J_{\alpha}$, noticing that the proof of skew-symmetry only requires \eqref{z388},\eqref{z42}, rather than \eqref{z38},\eqref{z41}.
\end{proof}
We can now apply \eqref{c3} to \eqref{z35} because of Proposition \ref{Bessfish}. We take $m=1,a_1=0$ and $a_2=t$. Note that, for $\alpha>1$,
\begin{equation*}
	f_1=(I-2S^{\ast})^{-1}H\delta_{a_1}=0\in L^2(0,t),
\end{equation*}
and so the right hand side in \eqref{c3} becomes
\begin{equation}\label{z46}
	D_M(\Delta)=D_{2S}(t)\big(1-\tau(t)\big),\hspace{1cm}\tau(t):=\big((I-2S^{\ast})^{-1}H\big)(t,t),
\end{equation}
where $D_{2S}(t)$ is the Fredholm determinant of $2S$ on $L^2(0,t)$, with the kernel of $S$ written in \eqref{z35}. By the properties of $J_{\alpha}$, $\lim_{t\downarrow 0}\tau(t)=0$, and $\tau(t)$ satisfies the following equation.
\begin{prop} Assume $I-2S^{\ast}$ is invertible on $L^2(0,t)$. Then $\tau(t)$ satisfies
\begin{equation}\label{z47}
	\tau(t)=-\tau(t)-\frac{2\tau(t)^2}{1-2\tau(t)},\ \ \ t>0.
\end{equation}
\end{prop}
\begin{proof} From \eqref{z27}, we deduce that invertibility of $I-2S^{\ast}$ yields invertibility of $I-2S+2\sum_{k=1}^{2m}(-1)^k(\delta_{a_k}\otimes\delta_{a_k})H$ and so, by Lemma \ref{SM}, once specialised to $m=1,a_1=0,a_2=t$,
\begin{equation*}
	1-2\tau(t)=1-2\big((I-2S^{\ast})^{-1}H\big)(t,t)=1-2\int_0^t(H\delta_t)(x)\big((I-2S)^{-1}\delta_t\big)(x)\d x\neq 0.
\end{equation*}
Also, by \eqref{e17},
\begin{align*}
	\tau(t)=&\,\,\big((I-2S^{\ast})^{-1}H\big)(t,t)\stackrel{\eqref{z27}}{=}\Big(H\big(I-2S-2\delta_t\otimes H\delta_t\big)^{-1}\Big)(t,t)\\
	=&\,\,\big(H(I-2S)^{-1}\big)(t,t)+\frac{2}{1-2\tau(t)}\Big(H(I-2S)^{-1}(\delta_t\otimes H\delta_t)(I-2S)^{-1}\Big)(t,t)\\
	=&\,\,-\tau(t)+\frac{2}{1-2\tau(t)}\big(H(I-2S)^{-1}\delta_t\big)(t)\big((I-2S^{\ast})^{-1}H\delta_t\big)(t)=-\tau(t)-\frac{2\tau(t)^2}{1-2\tau(t)},
\end{align*}
as claimed in \eqref{z47}. The proof is complete.
\end{proof}

\begin{cor} Assume $I-2S^{\ast}$ and $I-Q$ are invertible on $L^2(0,t)$, with $Q:L^2(0,t)\rightarrow L^2(0,t)$ denoting the multiplicative Hankel composition integral operator with trace class Bessel kernel
\begin{equation}\label{z477}
	Q(x,y):=\frac{1}{4}\int_0^1J_{\alpha}(\sqrt{xu})J_{\alpha}(\sqrt{uy})\d u,\ \ \ \ \ x,y\in(0,t).
\end{equation}
Then we have for the Fredholm determinant $D_M(t)$ of $M$ on $L^2(0,t)\oplus L^2(0,t)$ in \eqref{z5},\eqref{z6},\eqref{z35},\eqref{z355},
\begin{equation}\label{z48}
	D_M(t)=D_Q(t)\bigg[1+\frac{1}{2}\int_0^t\big((I-Q)^{-1}\phi\big)(x)\Phi(x)\d x\bigg],\ \ \ \phi(x):=\frac{J_{\alpha-1}(\sqrt{x})}{2\sqrt{x}},\ \Phi(x):=\int_0^x\frac{J_{\alpha+1}(\sqrt{z})}{2\sqrt{z}}\d z,
\end{equation}
where $D_Q(t)$ is the Fredholm determinant of $Q$ with kernel \eqref{z477} on $L^2(0,t)$.
\end{cor}
\begin{proof} Equation \eqref{z47} admits two solutions, $\tau_1(t)=0$ and $\tau_2(t)=1$, but only one links to $\lim_{t\downarrow 0}\tau(t)=0$. So, \eqref{z46} implies $D_M(\Delta)=D_{2S}(t)$ and here we factor out $I-Q$ from $D_{2S}(t)$. This yields \eqref{z48}.
\end{proof}
Identity \eqref{z48} plays the analogue of \eqref{c11} for \eqref{z35}. Still, \eqref{z48} is not on the same level as \eqref{c11} because of the asymmetric appearance of $\alpha$ in the second factor in the right hand side of \eqref{z48}. What's more, if we attempt to study more general multiplicative Hankel composition kernels of the form
\begin{equation}\label{z488}
	S(x,y)=\frac{1}{2}\bigg(\int_0^1\phi(xu)\phi(uy)\d u-\frac{\psi(x)}{2\sqrt{x}}\int_0^y\vartheta(z)\frac{\d z}{\sqrt{z}}\bigg),\ \ \ \ \ x,y\in(0,t),
\end{equation}
for some appropriate special functions $\phi,\psi,\vartheta:(0,\infty)\rightarrow\mathbb{R}$, then skew-symmetry of $G(x,y)$ and $H(x,y)$ constrains those functions significantly, as argued below.
\begin{rem} What the proof workings of Proposition \ref{Bessfish} show is that, with \eqref{z35}, once $J_{\alpha}$ is of the form 
\begin{equation}\label{z44}
	J_{\alpha}(x)=x^{\alpha}\sum_{k=0}^{\infty}c_k(\alpha)x^{2k},\ \ \ x>0,
\end{equation}
conditions \eqref{z388},\eqref{z42} are sufficient for the skew-symmetry of $G(x,y)$ and $H(x,y)$, and that constrains $J_{\alpha}$ to be essentially a Bessel function. On the other hand, imposing skew-symmetry on $G(x,y)$ and $H(x,y)$ within the class \eqref{z44}, we learn from \eqref{z36},\eqref{z37} that $G(x,x)=H(x,x)\equiv 0$ imply the equations
\begin{align*}
	\sum_{\ell=0}^m\bigg(\frac{c_{m-\ell}(\alpha)c_{\ell}(\alpha)}{m+\alpha+1}-\frac{1}{2}c_{m-\ell}(\alpha-1)c_{\ell}(\alpha+1)\bigg)=&\,\,0,\\
	\sum_{\ell=0}^m\bigg(\frac{c_{m-\ell}(\alpha)c_{\ell}(\alpha)}{m+\alpha+1}\frac{1}{m-\ell+\frac{\alpha}{2}+1}-\frac{c_{m-\ell}(\alpha-1)c_{\ell}(\alpha+1)}{2\ell+\alpha+2}\frac{1}{m-\ell+\frac{\alpha}{2}}\bigg)=&\,\,0,
\end{align*}
valid for all $m\in\mathbb{Z}_{\geq 0}$. These yield linear $\alpha$-recursions from which we deduce, after some manipulations,
\begin{equation}\label{z45}
	c_k(\alpha)=\frac{d_k}{2^{\alpha}\Gamma(k+\alpha+1)},\ \ k\in\mathbb{Z}_{\geq 0}\ \ \ \ \ \textnormal{with some}\ \ (d_k)_{k=0}^{\infty},
\end{equation}
and thus, \eqref{z388} is in fact a consequence of skew-symmetry. Moreover, enforcing $G(x,y)+G(y,x)\equiv 0$ and $H(x,y)+H(y,x)\equiv 0$, \eqref{z40} and \eqref{z43} yield, in particular, on $(0,t)\ni x$,
\begin{align*}
	\sum_{k,\ell=0}^{\infty}\frac{c_k(\alpha)c_{\ell}(\alpha)}{k+\ell+\alpha+1}\bigg(\frac{\ell(k-\ell+1)}{k+\alpha+1}+\frac{k(\ell-k+1)}{\ell+\alpha+1}\bigg)x^{\ell+k}\equiv&\, 0,\\
	\sum_{k,\ell=0}^{\infty}\frac{c_k(\alpha)c_{\ell}(\alpha)}{k+\ell+\alpha+1}\bigg(\frac{\ell(k-\ell+1)}{(2\ell+\alpha)(k+\alpha+1)(k+\frac{\alpha}{2}+1)}+\frac{k(\ell-k+1)}{(2k+\alpha)(\ell+\alpha+1)(\ell+\frac{\alpha}{2}+1)}\bigg)x^{k+\ell}\equiv&\, 0,
\end{align*}
and thus the equations
\begin{align*}
	\sum_{\ell=0}^mc_{m-\ell}(\alpha)c_{\ell}(\alpha)\bigg(\frac{\ell(m-2\ell+1)}{m-\ell+\alpha+1}+\frac{(m-\ell)(2\ell-m+1)}{\ell+\alpha+1}\bigg)=&\,\,0,\\
	\sum_{\ell=0}^mc_{m-\ell}(\alpha)c_{\ell}(\alpha)\bigg(\frac{\ell(m-2\ell+1)}{(2\ell+\alpha)(m-\ell+\alpha+1)(m-\ell+\frac{\alpha}{2}+1)}+\frac{(m-\ell)(2\ell-m+1)}{(2m-2\ell+\alpha)(\ell+\alpha+1)(\ell+\frac{\alpha}{2}+1)}\bigg)=&\,\,0,
\end{align*}
valid for all $m\in\mathbb{Z}_{\geq 0}$. Inserting here \eqref{z45}, we obtain $k$-recursions from which we deduce, after some manipulations,
\begin{equation*}
	d_k=\frac{(-1)^kc}{4^kk!},\ \ \ k\in\mathbb{Z}_{\geq 0}\ \ \ \ \ \textnormal{with some}\ \ c\in\mathbb{R},
\end{equation*}
and so, all together, skew-symmetry of $G(x,y),H(x,y)$ within the class \eqref{z44} necessarily enforces
\begin{equation*}
	c_k(\alpha)=\frac{(-1)^kc}{2^{\alpha}4^kk!\,\Gamma(k+\alpha+1)},
\end{equation*}
i.e. it enforces $J_{\alpha}$ to be a multiple of the Bessel function of order $\alpha$. 
\end{rem}

It is in the sense of the last Remark that the class of kernels \eqref{z488} is narrower when compared to \eqref{c8},\eqref{c21}. We close out this section with a final comment on the relation between \eqref{c8} and a broader class of multiplicative Hankel composition kernels than the special one \eqref{z35}.
\begin{rem} View \eqref{c8} as kernel of $S:L^2(t,\infty)\rightarrow L^2(t,\infty)$ with $t\in\mathbb{R}$ but take $\phi(x)=\psi(\e^{-x})\e^{-\frac{1}{2}x}$, with $\psi:(0,\infty)\rightarrow\mathbb{R}$ reasonably well behaved in it. What results is
\begin{equation}\label{z49}
	S(x,y)=\frac{1}{2}\bigg(\int_0^1\e^{-\frac{1}{2}x}\psi\big(\e^{-x}v\big)\psi\big(v\e^{-y}\big)\e^{-\frac{1}{2}y}\d v-\frac{1}{2}\psi(\e^{-x})\e^{-\frac{1}{2}x}\int_0^{\e^{-y}}\psi(w)\frac{\d w}{\sqrt{w}}\bigg),\ \ \ x,y\in(t,\infty),
\end{equation}
and so, if $\mathcal{M}:L^2(0,\e^{-t})\rightarrow L^2(t,\infty)$ and $\mathcal{M}^{-1}:L^2(t,\infty)\rightarrow L^2(0,\e^{-t})$ denote the bounded linear transformations
\begin{align*}
	(\mathcal{M}f)(x):=&\,\e^{-\frac{1}{2}x}f(\e^{-x}),\ \ \ x\in(t,\infty),\ \ \ f\in L^2(0,\e^{-t});\\
	(\mathcal{M}^{-1}f)(x):=&\,\frac{1}{\sqrt{x}}f(-\ln x),\ \ x\in(0,\e^{-t}),\ \ \ f\in L^2(t,\infty),
\end{align*}
we arrive at the factorisation identity $S=\mathcal{M}T\mathcal{M}^{-1}$, valid on $L^2(t,\infty)$ and where $T:L^2(0,\e^{-t})\rightarrow L^2(0,\e^{-t})$ is the integral operator with kernel
\begin{equation}\label{z50}
	T(x,y):=\frac{1}{2}\bigg(\int_0^1\psi(xv)\psi(vy)\d v-\frac{1}{2}\psi(x)\frac{1}{\sqrt{y}}\int_0^y\psi(w)\frac{\d w}{\sqrt{w}}\bigg),\ \ \ x,y\in(0,\e^{-t}).
\end{equation}
The similarities between \eqref{z35} and \eqref{z50} are apparent and Fredholm determinants of scalar multiplicative Hankel composition operators of the type \eqref{z50} were investigated in \cite{Bo}. However, it is not immediately clear if $S=\mathcal{M}T\mathcal{M}^{-1}$ could be used in the definition of an appropriate matrix operator \eqref{z5} with $S$ of the type \eqref{z50} and appropriate skew-symmetric $G$ and $H$ so that $G(x,y)=-\frac{\partial}{\partial y}S(x,y)$ and $\frac{\partial}{\partial x}H(x,y)=S(x,y)$.

\end{rem}

%%%%%%%%%%%%%%%%%%%%%%%%%%%%%%%%%%%%%%%%%%%%%%%%%%%%%%%%%%%%%%

\begin{appendix}
\section{One useful definition and one useful tool}\label{appA}

Fredholm determinants on $L^2(X,\d\mu)$, for some separable measure space, are well-known objects, and \cite{S0,GGK} standard references for their properties. Fredholm Paffians on $L^2(X,\d\mu)\oplus L^2(X,\d\mu)$, on the other hand, are perhaps less well-known. Here is an easy-to-work-with definition.
\begin{definition} If $K(x,y)$ in \eqref{a3} is a skew-symmetric $2\times 2$ matrix-valued kernel that induces a trace class integral operator $K$ on $L^2(X,\d\mu)\oplus L^2(X,\d\mu)$, then the Fredholm Pfaffian of $K$ on $L^2(X,\d\mu)\oplus L^2(X,\d\mu)$ is defined as
\begin{equation}\label{p1}
	\textnormal{pf}\,(J-K):=1+\sum_{\ell=1}^{\infty}\frac{(-1)^{\ell}}{\ell!}\int_{X^{\ell}}\textnormal{Pf}\big[K(x_j,x_k)\big]_{j,k=1}^{\ell}\prod_{j=1}^{\ell}\d\mu(x_j),
\end{equation}
where $\textnormal{Pf}$ evaluates the ordinary Paffian of the $2\ell\times 2\ell$ skew-symmetric matrix $[K(x_j,x_k)]_{j,k=1}^{\ell}$.

\end{definition}

Lastly, we list the following infinite-dimensional version of the Sherman-Morrison identity.
\begin{lem}\label{SM} Suppose $K$ is a trace class integral operator on $L^2(X,\d\mu)$ for some separable measure space. Assume $I-K$ is invertible. Then its rank-one perturbation $I-K+\alpha\otimes\beta$ with $\alpha,\beta\in L^2(X,\d\mu)$ is also invertible if and only if
\begin{equation*}
	\langle \beta,(I-K)^{-1}\alpha\rangle:=\int_X\beta(y)\big((I-K)^{-1}\alpha\big)(y)\d\mu(y)\neq -1.
\end{equation*}
Moreover, if $I-K+\alpha\otimes\beta$ is invertible on $L^2(X,\d\mu)$, then
\begin{equation}\label{e17}
	(I-K+\alpha\otimes\beta)^{-1}=(I-K)^{-1}-\frac{(I-K)^{-1}(\alpha\otimes\beta)(I-K)^{-1}}{1+\langle\beta,(I-K)^{-1}\alpha\rangle}.
\end{equation}
\end{lem}
\begin{proof} One part is an explicit calculation: assuming $1+\langle\beta,(I-K)^{-1}\alpha\rangle\neq 0$, one shows that the right hand side of \eqref{e17} constitutes a right inverse of $I-K+\alpha\otimes\beta$ and that the same right hand side is a bounded linear operator on $L^2(X,\d\mu)$. The other part uses the factorization
\begin{equation*}
	I-K+\alpha\otimes\beta=(I-K)\big(I+(I-K)^{-1}\alpha\otimes\beta\big),
\end{equation*}
which yields
\begin{equation*}
	\det(I-K+\alpha\otimes\beta)=\det(I-K)\Big[1+\langle\beta,(I-K)^{-1}\alpha\rangle\Big],
\end{equation*}
and thus says $I-K+\alpha\otimes\beta$ invertible if and only if $1+\langle\beta,(I-K)^{-1}\alpha\rangle\neq 0$, see \cite[Theorem $3.5$]{S}. The proof is complete.
\end{proof}

\end{appendix}

%%%%%%%%%%%%%%%%%%%%%%%%%%%%%%%%%%%%%%%%%%%%%%%%%%%%%%%%%%%%%%

\end{document}